\documentclass{article}
\usepackage[utf8]{inputenc}
\usepackage{amssymb}
\usepackage{MnSymbol}
\usepackage{amsfonts}
\usepackage{mathtools}
\usepackage{amsmath}
\usepackage{amsthm}
\usepackage{xcolor}
\usepackage{hyperref}
\usepackage[shortlabels]{enumitem}
\usepackage{soul}
\usepackage{geometry}
\usepackage{bbm}
\usepackage{comment}
\usepackage{booktabs}
\usepackage{physics}

\newcommand\mc[1]{\mathcal{#1}}

\newcommand{\cD}{{\mc{D}}}

\newcommand{\cH}{{\mc{H}}}

\newcommand{\cI}{{\mc{I}}}

\newcommand{\cL}{{\mc{L}}}
\newcommand{\cM}{{\mc{M}}}

\newcommand{\cO}{{\mc{O}}}
\newcommand{\cP}{{\mc{P}}}

\newcommand{\aP}{{\widetilde{\mc{P}}}}
\newcommand{\pP}{{\overline{\mc{P}}}}
\newcommand{\cJ}{{\mc{J}}}

\newcommand{\cS}{{\mc{S}}}
\newcommand{\cT}{{\mc{T}}}

\newcommand{\hf}{{\widehat{f}}}

\newcommand\hV{{\widehat{V}}}

\newcommand{\ta}{{\widetilde{a}}}
\newcommand{\tb}{{\widetilde{b}}}

\newcommand{\tv}{{\widetilde{v}}}

\newcommand{\tw}{{\widetilde{w}}}

\newcommand{\tA}{{\widetilde{A}}}

\newcommand{\tL}{{\widetilde{L}}}

\newcommand{\tT}{{\widetilde{T}}}

\newcommand{\oL}{{\overline{L}}}
\newcommand{\oa}{{\overline{a}}}

\newcommand{\oS}{{\overline{S}}}

\newcommand{\oP}{{\overline{P}}}
\newcommand{\oQ}{{\overline{Q}}}

\newcommand{\cU}{{\mc{U}}}

\newcommand{\soP}{{\mathfrak{so}|_{\aP^n}}}
\newcommand{\su}{{\mathfrak{su}}}
\newcommand{\spf}{{\mathfrak{sp}}}
\newcommand{\sof}{{\mathfrak{so}}}
\newcommand{\pf}{{\mathfrak{p}}}
\newcommand{\lf}{{\mathfrak{l}}}

\newcommand{\LH}{{\mathrm{Mat}_{2^n}(\C)}}
\newcommand{\LHsa}{{\mathrm{Herm}_{2^n}(\C)}}

\newcommand{\Den}{{\mc{D}(\cH)}}

\newcommand{\UH}{{\cU(\cH)}}

\newcommand{\R}{{\mathbb R}}
\newcommand{\C}{{\mathbb C}}

\newcommand{\CP}{{\mathbb{CP}}}

\newcommand{\Cl}{{\mathbf{C}}}
\newcommand{\pCl}{{\overline{\Cl}}}
\newcommand{\one}{{\mathbbm{1}}}

\newcommand{\zz}{{\mathbb{Z}}}
\newcommand{\ra}{{\rightarrow}}

\newcommand{\sym}{{\mathrm{sym}}}
\newcommand{\Sym}{{\mathrm{Sym}}}

\newcommand{\Ad}{{\mathrm{Ad}}}
\newcommand{\id}{{\mathrm{id}}}
\newcommand{\ad}{{\mathrm{ad}}}
\newcommand{\SO}{{\mathrm{SO}}}
\newcommand{\OG}{{\mathrm{O}}}
\newcommand{\SU}{{\mathrm{SU}}}
\newcommand{\U}{{\mathrm{U}}}
\newcommand{\PSU}{{\mathrm{PSU}}}

\newcommand{\Sp}{{\mathrm{Sp}}}

\newcommand{\PStab}{\mathrm{PStab}}

\newcommand{\fg}{{\mathfrak{g}}}

\newcommand{\SWAP}{{\mathrm{SWAP}}}
\newcommand{\vect}{{\mathrm{vec}}}
\newcommand{\ch}{{\mathrm{char}}}
\newcommand{\Arf}{{\mathrm{Arf}}}
\newcommand{\End}{{\mathrm{End}}}

\newcommand{\Comm}{{\mathrm{Comm}}}
\newcommand{\Var}{{\mathrm{Var}}}

\newcommand{\Lag}{{\mathrm{Lag}}}
\newcommand{\PG}{{\mathrm{PG}}}

\newcommand{\ind}{\mathbf 1}

\newtheorem{definition}{Definition}
\newtheorem{lemma}{Lemma}
\newtheorem{corollary}{Corollary}
\newtheorem{proposition}{Proposition}
\newtheorem{theorem}{Theorem}

\newcommand{\kett}[1]{\ensuremath{\vert#1\rangle\rangle}}
\newcommand{\braa}[1]{\ensuremath{\langle\langle #1\vert}}

\begin{document}

\title{An algebraic characterisation of $n$-qubit pure-state invariants}
\author{Markus Frembs\thanks{\href{mailto:markus.frembs@itp.uni-hannover.de}{markus.frembs@itp.uni-hannover.de}}\\
\small{\textit{Institut f\"ur Theoretische Physik, Leibniz Universit\"at Hannover}}\\ \small{\textit{Appelstraße 2, 30167 Hannover, Germany}}}
\date{}

\maketitle
\vspace{.25cm}
\begin{abstract}
    We study the set of $n$-qubit pure quantum states. Specifically, we consider their incarnation as density matrices in Bloch representation, that is, in terms of their expansion with respect to the basis of Hermitian Pauli operators. Building on prior work \cite{HoehnWever2017,FrembsHoehnNataleWever2026a,FrembsHoehnNataleWever2026b}, we define a family of quadratic functions in the Pauli expectation values, and fully characterise all instances that result in state-independent purity constraints. The resulting family of invariants is surprisingly rich in structure, with strong ties to the geometry of symplectic polar spaces.
    
    Having charted the general landscape, we next define a sub-family of pure-state invariants that is not only kinematical, but in addition satisfies the dynamical constraint of being stable under state update. We show this family to coincide with that recently found in Ref.~\cite{FrembsHoehnNataleWever2026a,FrembsHoehnNataleWever2026b}, and prove that it corresponds with (different embeddings of) symplectic Lie algebras $\spf(2^{n-1})\subset\su(2^n)$, which admit a generating set of Pauli operators. We carefully distinguish the symplectic case from its orthogonal counterpart, using the Arf invariant of quadratic refinements of the symplectic form on the projective Pauli group, which we further identify with quaternionic structures on $\C^{2^n}$ that are both generated by and restricted to Pauli operators.
    
    As a consequence of this analysis, we reveal previously largely overlooked structure in the Pauli group. We highlight this in a number of examples. We characterise all state-independent uncertainty relations in variance for Pauli observables, as well as all Pauli channels for which the output fidelity and output purity are state-independent. We further show that the orbits of the projective stabiliser groups of pure-state invariants, which for $n\geq2$ are proper subgroups of the Clifford group, form complex projective state $2$-designs, while their unitary lifts form symplectic $2$-designs.

    In summary, this shows that not only the pure-state geometry, but also a substantial part of the algebraic structure of the Pauli group, and of quantum theory as a whole, is encoded by pure-state invariants. Together with their relatively simple structure, our work thus provides a practically useful toolkit - in a first instance, for the study of purity, and in a second, for various applications of the Pauli group in quantum information science.
\end{abstract}

\newpage
\tableofcontents
\newpage

\section{Introduction}

The Bloch (or Pauli) representation is one of the most concrete ways of turning questions about multipartite quantum states into questions about real geometry. Any $n$-qubit density operator $\rho$ can be expanded as $\rho=\frac{1}{2^n}\sum_{P\in\aP_n}\alpha_P P$, where $\aP_n=\{\one,X,Y,Z\}^{\otimes n}$ are the phase-free Hermitian $n$-qubit Pauli operators which form a basis for the space of Hermitian operators $\LHsa$ on the Hilbert space of $n$-qubits $\cH=\C^{2^n}$; the real expansion coefficients $\alpha_P=\tr[\rho P]$ collectively define the Bloch vector $\alpha\in\R^{4^n}$ corresponding to $\rho$.\footnote{Normalisation, $\tr[\rho]=1$, fixes the identity component $\alpha_\one=1$, and we may thus take $\alpha\in\R^{4^n-1}$.} This representation is ubiquitous in quantum information, stemming from the many useful properties of the Pauli group and its automorphism group within the unitary group, called the Clifford group. In particular, Pauli and Clifford gates form the backbone of many quantum computing architectures and of quantum error correction, appear in error models, benchmarking and simulation algorithms, and underlie many constructions in condensed matter physics and statistical mechanics, to name but a few applications. All of these directly or indirectly make use of the above representation of quantum states. Yet, despite its apparent simplicity, the geometry carved out by the manifold of pure quantum states quickly becomes intricate for increasing $n$; the set of $n$-qubit pure states when expressed in terms of Bloch vectors $\alpha$ is a highly nontrivial real-algebraic subset of $\R^{4^n-1}$. Indeed, the scalar purity condition $\tr[\rho^2]=1$ becomes a simple quadratic constraint in Pauli coordinates, $1=\tr[\rho^2]=\frac{1}{2^n}\sum_{P}\alpha_P^2$, so that pure states are naturally inscribed into the `Bloch sphere' of radius $\sum_{P\neq\one}\alpha_P^2=2^n-1$. In the single-qubit case, this is the only purity constraint and pure states fill the entire Bloch sphere. Yet, for $n\geq 2$, purity is no longer described by a single trace condition, but instead by the idempotency relation $\rho^2=\rho$, which comprises $4^n-1$ linearly dependent quadratic relations between the Bloch components of a given quantum state.\footnote{Of course, for density operators, $\tr[\rho^2]=1$ still certifies purity. From this perspective the richer purity constraints capture in part the interplay of positivity and purity.} Compared to the tracial constraint, the latter involves mixed Pauli expectation terms and is thus much harder to work with. The resulting ad hoc parametrisation does not in any obvious way exploit the geometric properties of the pure-state manifold.

In a companion paper \cite{FrembsHoehnNataleWever2026a}, we find a family of quadratic invariants in Pauli expectation values, that simplifies the structure of the idempotency relation $\rho^2=\rho$ in a number of ways: the constraints of this family are (i) state-independent, (ii) diagonal in Bloch components, (iii) supported on strict subsets of Pauli observables only, and (iv) exhibit structural properties that remain opaque on the level of the relation $\rho^2=\rho$. In the present paper, we extend these findings in various ways.\\

\noindent\textbf{Summary of main results.}
\begin{itemize}
    \item We characterise all quadratic pure-state invariants satisfying (i)-(iii), and relate them to known open problems in the geometry of symplectic polar spaces (see Thm.~\ref{thm: all weighted quadratic invariants}, Thm.~\ref{thm: all Boolean pure-state invariants}).
    \item We define the sub-family of these invariants that are stable under state update, which mimics the properties of the invariants for two qubits in Ref.~\cite{HoehnWever2017}, and prove that this sub-family corresponds with the generalisation in Ref.~\cite{FrembsHoehnNataleWever2026a} (see Cor.~\ref{cor: completeness}).
    \item We analyse the structure of these pure-state invariants from a variety of perspectives, in particular, we relate them both with quadratic refinements of the symplectic form of the projective Pauli group (see Thm.~\ref{thm: pure-state invariants as quadratic refinements}), as well as with (compact) symplectic Lie algebras which admit generating sets given by subsets of $\aP_n$ (see Thm.~\ref{thm: all update-closed pure-state invariants}).
\end{itemize}
We then showcase a number of applications. While likely not an exhaustive list, it hints at the broader versatility of our results, beyond their immediate relevance for the geometry of pure states.
\begin{itemize}
    \item We obtain all state-independent uncertainty relations in variance for Pauli observables (see Cor.~\ref{cor: all variance uncertainty relations}).\footnote{For a more detailed analysis of pure-state invariants and complementarity, see also Ref.~\cite{FrembsHoehnNataleWever2026b}.}
    \item We characterise all Pauli channels with state-independent output fidelity and state-independent output purity (see Cor.~\ref{cor: state-independent contraction} and Cor.~\ref{cor: state-independent unitarity}).
    \item We identify the symmetries of the pure-state invariants with subgroups of the Clifford group (Lm.~\ref{lm: transitivity} and Cor.~\ref{cor: Clifford intersections}), and prove them to be symplectic $2$-designs, as well as giving rise to (complex projective) state $2$-designs (see Thm.~\ref{thm: 2-design from invariant symmetries}).
\end{itemize}

Throughout, we will draw heavily on the finite symplectic geometry of the Pauli group (as well as its relation to classical Lie algebras in Sec.~\ref{sec: structure}). These structures have been explored in much detail in Ref.~\cite{ShultThas1994,Thas1989,Taylor1992,Grove2002}, and more recently in Ref.~\cite{Cuypers2026,GargiuloHerringerZeier2026}. Our contribution is complementary to those: we use the structure theory to characterise pure-state invariants, specifically we show that those stable under state update form a single family (Thm.~\ref{thm: pure-state invariants as quadratic refinements}), while the general classification turns out to be equivalent to an open problem in finite geometry (see Sec.~\ref{sec: Boolean pure-state invariants}).\\

\textbf{Outline.} Sec.~\ref{sec: preliminaries} introduces some necessary background. We first introduce our notation for and review the relevant structures in the Pauli group (Sec.~\ref{sec: Pauli observables}), before defining a generalised notion of the quadratic measure that underlies the family of pure-state invariants in Ref.~\cite{HoehnWever2017,FrembsHoehnNataleWever2026a} (Sec.~\ref{sec: quadratic quantity}).

Sec.~\ref{sec: characterisation} contains our main results, Thm.~\ref{thm: all weighted quadratic invariants}, Thm.~\ref{thm: all Boolean pure-state invariants} and Thm.~\ref{thm: all update-closed pure-state invariants}: namely, a full characterisation of these quadratic quantities both in the real-valued (Sec.~\ref{sec: weighted case}) and Boolean case (Sec.~\ref{sec: Boolean pure-state invariants}), as well as under an additional dynamical constraint: stability under state update (Sec.~\ref{sec: update-closed Boolean invariants}).

In Sec.~\ref{sec: structure}, we analyse the structure of these pure-state invariants in detail, starting with an explicit representation in Weyl form. We identify them as generating sets of Pauli Lie algebras (Sec.~\ref{sec: Cartan decomposition}), before singling out the symplectic case, first, in terms of the Arf invariant of quadratic refinements (Sec.~\ref{sec: quadratic refinement}), and, second, in terms of their quaternionic structures on the $n$-qubit Hilbert space (Sec.~\ref{sec: quaternionic structure}). We analyse symmetry properties of pure-state invariants in Sec.~\ref{sec: invariant set stabiliser}.

Sec.~\ref{sec: applications} collects a number of applications of our main results. We obtain a characterisation of all variance-based, state-independent uncertainty relations for Pauli operators (Sec.~\ref{sec: uncertainty relations}). This complements our analysis of complementarity in the $n$-qubit Pauli group in Ref.~\cite{FrembsHoehnNataleWever2026b}. Next, we relate pure-state invariants to the theory of (randomised and unitarity) benchmarking of Pauli channels (Sec.~\ref{sec: Pauli channels}). Moreover, we provide an alternative perspective on the closure under state update in terms of Pauli projections (Sec.~\ref{sec: positive Pauli projection maps}). We then turn to symmetries, and prove that the projective stabiliser groups of pure-state invariants give rise to complex projective state $2$-designs (Sec.~\ref{sec: 2-designs}), thus providing a first explicit example for the recent observation in Ref.~\cite{WestEtAl2024,ShayaEtAl2025} that such state designs can be obtained not only from unitary but already from symplectic $k$-designs. We discuss this result with closely related results obtained independently recently \cite{GargiuloHerringerZeier2026}.

Sec.~\ref{sec: conclusion} concludes. Two appendices provide further details and additional results.

App.~\ref{app: cubic purity relations} translates our quadratic pure-state invariants into cubic ones, which exhibits their intimate connection with Jordan algebras.

App.~\ref{app: symplectic representation} provides further background on the representations of the symplectic Pauli Lie algebras in Bell basis. In particular, we relate pure-state invariants with quadratic Casimirs (App.~\ref{app: pure-state invariants from quadratic Casimirs}), and relate the different representations under the explicit parametrisation in Def.~\ref{def: family of pure-state invariants} (App.~\ref{app: changing representation}).

\section{Definition of pure-state invariants}\label{sec: preliminaries}

In Sec.~\ref{sec: Pauli observables}, we collect some basic facts about and introduce our notation for Hermitian Pauli operators up to phase. In Sec.~\ref{sec: quadratic quantity}, we define the quadratic invariants that are our main subject.

\subsection{Preliminaries - Bloch representation and Pauli operators}\label{sec: Pauli observables}

Recall that the ($n$-qubit) Pauli group is given by $\cP_n=\langle\one,X,Y,Z\rangle^{\otimes n}$, where
\begin{align}\label{eq: Paulis}
    \one&=\begin{pmatrix}
        1 & 0 \\ 0 & 1
    \end{pmatrix} &
    X&=\begin{pmatrix}
        0 & 1 \\ 1 & 0
    \end{pmatrix} &
    Y&=\begin{pmatrix}
        0 & -i \\ i & 0
    \end{pmatrix} &
    Z&=\begin{pmatrix}
        1 & 0 \\ 0 & -1
    \end{pmatrix}\; .
\end{align}
The product in $\cP_n$ is given by matrix multiplication, and its central phase group is $Z(\cP_n)=\{\pm\one,\pm i\one\}$.\footnote{We will overload notation and write $\one$ both for the single-qubit identity matrix in Eq.~(\ref{eq: Paulis}), as well as for the identity in $\cP_n$. Its meaning will always be clear from context.} Since two $n$-qubit Pauli operators either commute or anti-commute, it will be useful to consider these two operations separately, and restrict to (subsets of) Hermitian, respectively anti-Hermitian Pauli operators.\footnote{Alternatively, we will sometimes consider Pauli operators as generators of the special unitary Lie algebra, and thus identify phase-free Pauli operators with skew-Hermitian operators. Taking normalisation into account, the respective subsets are related via multiplication by the factor $\frac{i}{2}$ (see App.~\ref{app: symplectic representation}).\label{fn: su-generators}} The Hermitian Pauli operators $\aP_n:=\{\one,X,Y,Z\}^{\otimes n}$ form an operator basis in the space of Hermitian matrices, $\LHsa=\{M\in\LH\mid M^\dagger=M\}$. The expansion of a density matrix $\rho\in\Den:=\{\rho\in\LHsa\mid\rho\geq 0,\tr[\rho]=1\}$ with $\cH=\C^{2^n}$ in this basis is called its \emph{Bloch representation $\rho=\frac{1}{2^n}(\one+\sum_{\one\neq P\in\aP_n}\alpha_P P)$} whose Bloch components $\alpha_P\in\R$ obey $\alpha_\one=1$ and $\|\alpha\|^2=\sum_{\one\neq P\in\aP_n}\alpha^2_P\leq 2^n-1$ with equality if and only if $\rho$ is a pure state. This follows immediately from $\tr[\rho]=1$ and $\tr[\rho^2]\leq 1$ (with $\tr[\rho^2]=1$ for $\rho$ pure), together with the orthogonality relation $\tr[PQ]=2^n\delta_{QP}$, and the fact that all non-identity Pauli operators are traceless.

Note that $\pm\aP_n\cup\{0\}$ is closed under (normalised) anti-commutators, $\{P,Q\}:=\frac{1}{2}(PQ+QP)$; when instead equipped with the (normalised) commutator, $[P,Q]:=\frac{1}{2}(PQ-QP)$, skew-Hermitian Pauli operators $\pm i\aP^\circ_n$ form a generating set of the special unitary Lie algebra $\mathfrak{su}(2^n)$, that is, $[\pm i\aP^\circ_n,\pm i\aP^\circ_n]=\pm i\aP^\circ_n\cup\{0\}$. We write $\aP^\circ_n=\aP_n\backslash\{\one\}$ (and, more generally $\cS^\circ=\cS\backslash\{\one\}$ for $\one\in\cS\subset\aP_n$).

Below, we will be interested in Pauli operators up to phase. One way to avoid carrying along unnecessary phase terms is by restricting to the \emph{projective Pauli group $\pP_n:=\cP_n/Z(\cP_n)$}, obtained from $\cP_n$ by disregarding phases (see Ref.~\cite{SarkarVanDenBerg2021}). As a vector space, $\pP_n\cong\zz^{2n}_2$, in particular, $\pP_n$ forgets the multiplicative structure in $\cP_n$. Since we will often require this additional data, we imprint the (anti-)commutation relation in $\cP_n$ on the quotient $\pP_n$. To this end, we define the commutator map\footnote{$[\cdot,\cdot]$ carves out the ambient commutation relations within the Abelian group $\pP_n=\zz^{2n}_2$. We overload the notation and use the same symbol for the commutator map induced on the quotient.} $[\cdot,\cdot]$ defined on $\pP_n$ by $[\oP,\oQ]:=\overline{[P,Q]}$ for $P,Q\in\cP_n$ representatives of $\oP,\oQ\in\pP_n$. Note that this is well-defined on the phase-free quotient, and simultaneously imprints the anti-commutator structure on $\pP_n$. Abstractly, the resulting structure defines the symplectic vector space $(V,\omega)$ for $V=\zz^{2n}_2$ and $\omega:V\times V\ra\zz_2$ a symplectic, that is, nondegenerate, alternating, bilinear form. In turn, $(\pP_n,[\cdot,\cdot])$ defines a representation of $(V,\omega)$ in terms of phase-free Pauli operators.

To distinguish between the two, we will write $\{P_v\}_{v\in V}$ for generic Pauli representatives. For any choice of lift $v\mapsto P_v$, $\omega$ then coincides with the commutator map on $\pP_n$, that is, $\omega(v,w)=0$ if and only if $[P_v,P_w]=0$. If the underlying Pauli labels are not relevant, we will at times omit writing them explicitly. More generally, if the choice of central phase is irrelevant, we will write $\cP_n\ni P\sim R\in\aP_n$ for the unique element in $\aP_n$ that corresponds with a given Pauli up to phase, in particular, we will write $[P,Q]\in X\subset\aP_n$ for the unique element $R\in X$ such that $[P,Q]\sim R$ (and similarly for $\{P,Q\}\in X\subset\aP_n$). For Pauli subsets, set-valued commutators and anti-commutators will be understood projectively: vanishing brackets are discarded, while every nonzero bracket is identified with its unique Hermitian Pauli representative up to central phase.

In particular, we will often restrict to Hermitian Pauli representatives. To do so, we need to define a phase function on $\pP_n\cong V$. A common convention is the Weyl representation, defined for all $v=(a,b)\in\zz^n_2\times\zz^n_2=\zz^{2n}_2=V$, by
\begin{align}\label{eq: Weyl operators}
    \cP_n\ni W_v
    :=i^{q_W(v)}\otimes_{i=1}^n Z^{a_i}_iX^{b_i}_i\; ,
\end{align}
and where $q_W(v):=a\cdot b:=\sum_{i=1}^na_ib_i$. The symplectic form on $V$ then takes the explicit form $\omega(v,w):=v^T\sigma w\mod 2$, where $\sigma=\begin{pmatrix}
    0 & \one \\- \one & 0
\end{pmatrix}$ is the standard symplectic matrix.

For every $v,w\in V=\zz^{2n}_2$, the following properties are readily verified from Eq.~(\ref{eq: Weyl operators}):
\begin{align}\label{eq: Weyl commutation relations}
    W_vW_w
    &=(-1)^{\omega(v,w)}W_wW_v &
    [W_v,W_w]&=0\ \Longleftrightarrow\ \omega(v,w)=0\; .
\end{align}
The phase convention in Eq.~(\ref{eq: Weyl operators}) is encoded in the quadratic form $q_W:V\ra\zz_2$, $q_W(v)=a\cdot b$. Eq.~(\ref{eq: Weyl commutation relations}) thus expresses that $q_W$ defines a quadratic refinement of the symplectic form $\omega$, that is, $q_W(v+w)+q_W(v)+q_W(w)=\omega(v,w)\mod 2$ for all $v,w\in V$ (see Sec.~\ref{sec: quadratic refinement} for more details).

Equivalently, $q_W$ defines a central extension of the elementary Abelian $2$-group $V=\zz^{2n}_2$,\footnote{Let $p$ be a prime number. A group is called a $p$-group if the order of all its elements is a power of $p$. An Abelian $p$-group in which all non-identity elements have the same order is called an elementary (Abelian) $p$-group. Elementary $p$-groups correspond with vector spaces over the cyclic group $\zz_p$. An extraspecial $p$-group is a non-Abelian $\zz_p$-central extension of an elementary $p$-group whose center is cyclic of order $p$. Extraspecial $p$-groups may be seen as finite analogues of the Heisenberg-Weyl group.}
\begin{align}\label{eq: extraspecial 2-groups}
    1\quad\ra\quad\zz_2\quad\ra\quad E\quad\ra\quad V\quad\ra\quad 1\; ,
\end{align}
with $Z(E)=[E,E]=\zz_2$, that is, an extraspecial $2$-group $E$. Indeed, up to isomorphism such an extension is fixed by the relations $\tv^2=(-1)^{q(v)}$ and $[\tv,\tw]=(-1)^{\omega(v,w)}$ for any choice of lift $V\ni v\mapsto\tv\in E$. While $q_W$ in Eq.~(\ref{eq: Weyl operators}) is a convenient choice in many applications, we stress that - from the perspective of central extensions in Eq.~(\ref{eq: extraspecial 2-groups}) - it is not merely a conventional choice. In particular, there are two isomorphism classes of extraspecial $2$-groups, denoted by $2^{1+2n}_+$ and $2^{1+2n}_-$ \cite{Quillen1971}, where two extensions (equivalently, quadratic refinements $q,q'$ of the symplectic form $\omega$ on $V$) are isomorphic if $q'=q\circ S^{-1}$ for $S\in\Sp(2n,\zz_2)$ a symplectic transformation of $(V,\omega)$ (that is, a linear map $S:V\ra V$ preserving the symplectic structure, $\omega(v,w)=\omega(Sv,Sw)$ for all $v,w\in V$). $q_W=a\cdot b$ in Eq.~(\ref{eq: Weyl operators}) then corresponds with the class denoted by $E=2^{1+2n}_+$, which is a central product of $n$ copies of the dihedral group $D_8$. A representative of the alternative class $E=2^{1+2n}_-$, which is a central product of $n-1$ copies of $D_8$ and single copy of the quaternion group $Q_8=\{\pm 1,\pm i,\pm j,\pm k\mid i^2=j^2=k^2=ijk=-1\}$, is given by the quadratic refinement $q'=a\cdot b+a_1+b_1$. In Sec.~\ref{sec: quadratic refinement}, we will find that it is the isomorphism class corresponding to $q'$ which enters the characterisation of (update-closed) pure-state invariants.

Finally, we remark that the distinction between the two types of central extension of the projective Pauli group $V\cong\aP_n$ above disappears on the level of the full Pauli group $\cP_n$, which is itself a central product of (a representative of) either one of the extraspecial $2$-groups and the phase group $Z(\cP_n)=\{\pm 1,\pm i\}\cong Z_4$, that is, $\cP_n\cong Z_4\circ 2^{1+2n}_+\cong Z_4\circ 2^{1+2n}_-$. Note that the symplectic form $\omega$ on $V$ is independent of the choice of phase convention.\\

\textbf{Notation.} Since for most of our analysis we will merely use that (Hermitian) Pauli operators form a basis of $\LH$ ($\LHsa$) and keep track of their commutation relations, we will either label Pauli operators abstractly as elements of the symplectic vector space $V$ (of `Pauli labels') equipped with the symplectic form $\omega$, or write explicit representatives $\{P_v\}_{v\in V}$; for Hermitian representatives, we instead write $\{W_v\}_{v\in V}$, which correspond with the explicit Weyl representation in Eq.~(\ref{eq: Weyl operators}) up to an index-relabelling given by a symplectic transformation and a $\pm 1$ phase convention. Consequently, $\{W_v\}_{v\in V}\cong\aP_n$, and we thus sometimes write $\aP_n$ to indicate any Hermitian lift of $V$.

Depending on the specific context at hand, we will switch between the abstract and concrete viewpoints above. We hope that in doing so we do not impair overall readability, but help make the presentation feel more natural. For visual aid, we will use lower case letters for Pauli labels $v\in V$ and upper case letters for their Pauli representatives. Tab.~\ref{tab: notation} summarises our conventions.

\begin{table}[t]
    \centering
    \begin{tabular}{@{}ll@{}}
        \toprule
        Notation & Description \\
        \midrule
        $\cP_n=\langle\one,X,Y,Z\rangle^{\otimes n}$
        & $n$-qubit Pauli group \\

        $Z(\cP_n)$
        & central phase subgroup $\{\pm\one,\pm i\one\}$ \\

        $\pP_n=\cP_n/Z(\cP_n)$
        & projective Pauli group, $\pP_n\simeq V=\mathbb Z_2^{2n}$ \\

        $V=\zz^{2n}_2\cong\pP_n$
        & space of abstract Pauli labels\\

        $\aP_n=\{\one,X,Y,Z\}^{\otimes n}$
        & Hermitian Pauli representatives of $V$ \\

        $\aP_n^\circ$
        & non-identity Hermitian Pauli representatives, $\aP_n\backslash\{\one\}$ \\

        $W_v$
        & Hermitian Weyl representative of $v\in V$ \\

        $P_v$
        & generic Pauli representative of $v\in V$ \\

        $P\sim Q$
        & equality up to a central Pauli phase \\

        $\omega(v,w)$
        & symplectic form encoding Pauli commutation relations \\

        $\chi_w(v)$
        & multiplicative character $\chi_w(v)=(-1)^{\omega(v,w)}$ \\

        $q(v)$ ($q_W(v)$)
        & (Weyl) quadratic refinement of $\omega$ ($q_W(a,b)=a\cdot b$) \\

        $\tau(v)$
        & transpose sign, $W_v^T=\tau(v)W_v$, $\tau(v)=(-1)^{q_W(v)}$ \\

        $T_\pm$
        & level sets $T_\pm=\{v\in V\mid\tau(v)=\pm1\}$ \\
        \bottomrule
    \end{tabular}
    \caption{Notation for Pauli group and its underlying symplectic vector space.}
    \label{tab: notation}
\end{table}

\subsection{Quadratic pure-state invariants}\label{sec: quadratic quantity}

Our concern in this work is to study pure-state invariants of $n$-qubit quantum states $\rho\in\cD(\C^{2^n})$. Purity is of course encoded in the single relation $\rho^2=\rho$, which translates into a set of quadratic equations in the Bloch vector $\alpha$ (see Sec.~\ref{sec: Pauli observables}). However, by itself this does not provide much insight into the geometry of pure-state manifold. Generalising relations found first for two-qubit pure states in Ref.~\cite{Hoehn2017,HoehnWever2017}, in  Ref.~\cite{FrembsHoehnNataleWever2026a} we present a family of constraints diagonal in $\alpha$, and find that they reveal rich algebraic structure. In the present work, we substantially extend this analysis and fully characterise the family of quadratic quantities defined as follows.

\begin{definition}\label{def: I}
    Let $\pP\cong V=\zz^{2n}_2$, $\zeta\in\R^V$ and $\rho=\frac{1}{2^n}\sum_{v\in V}\alpha_v W_v\in\cD(\C^{2^n})$, and define the quadratic
    \begin{align}\label{eq: I}
          I_\rho(\zeta)
          :=\sum_{v\in V}\zeta(v)(\tr[\rho W_v])^2
          =\sum_{v\in V}\zeta(v)\alpha^2_v\; .
    \end{align}
    For a pure state $\psi\in\mathbb{CP}^{2^n-1}$, let $p_\psi:=\dyad{\psi}$ denote the corresponding rank-one projector, and write $I_\psi(\zeta):=I_{p_\psi}(\zeta)$. Whenever $\zeta=\ind_{V'}\in\zz^V_2$ is the indicator function of some subset $V'\subset V$ (that is, $\ind_{V'}(v)=1$ if $v\in V'$ and $\ind_{V'}(v)=0$ if $v\notin V'$), we will also write $I_\rho(V'):=I_\rho(\ind_{V'})=\sum_{v\in V'}(\tr[\rho W_v])^2$ as well as $I_\rho(\overline{V'})=I_\rho(V\backslash V')$.
    
    $I_\rho(\zeta)$ or simply $\zeta$ is called a \emph{pure-state invariant} if $I_\psi(\zeta)=\mathrm{const}$ for all $\psi\in\CP^{2^n-1}$, and a \emph{Boolean pure-state invariant} if $\zeta\in\zz^V_2$, equivalently if $I_\psi(\zeta)=I_\psi(\ind_{V'})=I_\psi(V')$ for $V'=\zeta^{-1}(1)$.
\end{definition}

We stop for some remarks regarding this definition. First, Def.~\ref{def: I} is essentially the Brukner-Zeilinger information measure from Ref.~\cite{BruknerZeilinger1999,BruknerZeilinger2001,Hoehn2017,HoehnWever2017,FrembsHoehnNataleWever2026a}, generalised to include arbitrary weighted averages with respect to $\zeta\in\R^{4^n}$. (Since $\alpha_\one=1$, $\zeta_\one$ defines an independent trivial invariant, we may fix it by setting $\zeta_\one=1$, and restrict to $\zeta\in\R^{V^\circ}\cong\R^{4^n-1}$.) While the latter somewhat obscures the interpretation of $I$ as an information measure, we include it in Def.~\ref{def: I} for added generality in our main result, before adopting the Boolean case later on (see Sec.~\ref{sec: applications}). Second, note that since $I_\rho(\zeta)$ in Eq.~(\ref{eq: I}) is diagonal in the Bloch components, it is independent of the choice of phase convention of Hermitian Pauli observables. Anticipating our analysis below, we note that for pure-state invariants $I_\psi(V')=\mathrm{const}$, this leaves a freedom in the choice of (imaginary) phases on elements $v\notin V'$, thus departing from the Weyl convention in Eq.~(\ref{eq: Weyl operators}). We will discuss the freedom in choosing such phase conventions and how they relate with different sub-structures in the Pauli group in more detail in Sec.~\ref{sec: structure}. Third, note that there always exists at least one pure-state invariant, since for every $n$-qubit pure state $\psi\in\CP^{2^n-1}$,
\begin{align}\label{eq: total information}
    I_\psi(V)
    =\sum_{v\in V}(\tr[p_\psi W_v])^2
    =\sum_{v\in V}\alpha^2_v
    =2^n\; .
\end{align}
Eq.~(\ref{eq: total information}) is of course nothing but the tracial purity constraint $\tr[\rho^2]=1$. Note also that normalisation, $\tr[\rho]=1$, imposes the constraint $I_\psi(\one)=1$. For a single qubit, $\tr[\rho^2]=1$ is the only purity constraint. For two qubits, Ref.~\cite{Hoehn2017,HoehnWever2017} identify maximal sets of mutually anti-commuting Pauli operators as pure-state invariants, and show that these exhaust all Boolean pure-state invariants. Since for $n=2$ such sets have cardinality $2n+1=5$, they are called `pentagon identities' \cite{HoehnWever2017}. However, mutually anti-commuting Pauli operators do not define pure-state invariants for $n>2$. Instead, in Ref.~\cite{FrembsHoehnNataleWever2026a} we obtain a family of nontrivial Boolean pure-state invariants that generalises the two-qubit case to arbitrary $n$. Yet, Ref.~\cite{FrembsHoehnNataleWever2026a} leaves open whether this family is complete, that is, it leaves open the following general characterisation problem:
\begin{equation}\label{*}
    \mathrm{For\ which\ } \zeta\in\R^V \mathrm{\ is\ Eq.~(\ref{eq: I})\ a \ \emph{purity\ invariant},\ that\ is, \ } I_\psi(\zeta)=\mathrm{const}\mathrm{\ for\ all\ \psi\in\CP^{2^n-1}}.\tag{$*$}
\end{equation}

We answer this question in Thm.~\ref{thm: all weighted quadratic invariants}, and more specifically for $\zeta=\ind_{V'}$ with $V'\subset V$ in Thm.~\ref{thm: all Boolean pure-state invariants} and Thm.~\ref{thm: all update-closed pure-state invariants}. In addition, we study in detail the family of solutions to Eq.~(\ref{*}) in Ref.~\cite{FrembsHoehnNataleWever2026a}, which exhibit various structural properties that provide novel insight into the geometry of $n$-qubit pure states, as well as of the $n$-qubit Pauli group, which is of interest in its own right. Indeed, in Sec.~\ref{sec: applications} we will prove a number of corollaries to the above theorems, that relate this structure to uncertainty relations (in variance), benchmarking of Pauli channels and complex projective $2$-designs.

\section{Characterisation of pure-state invariants}\label{sec: characterisation}

In this section, we will answer the question posed in (\ref{*}). We first consider the case of general weighted invariants, that is, $I_\psi(\zeta)=\mathrm{const}$ for $\zeta\in\R^V$ (see Def.~\ref{def: I}) in Sec.~\ref{sec: weighted case}, before restricting our attention to the case of Boolean pure-state invariants with $\zeta\in\zz^V_2$ in Sec.~\ref{sec: Boolean pure-state invariants}, and finally to the special case of pure-state invariants compatible with state update in Sec.~\ref{sec: update-closed Boolean invariants}.

\subsection{General pure-state invariants}\label{sec: weighted case}

We first translate what it means for $I_\rho(\zeta)$ with $\zeta\in\R^V$ in Def.~\ref{def: I} to be a pure-state invariant into a property of an operator associated with it.\\

\textbf{Projectors on symmetric subspaces.} It is instructive to re-write $I_\rho(\zeta)$ in Eq.~(\ref{eq: I}) as follows
\begin{align}\label{eq: diagonal invariant operator}
    I_\rho(\zeta)
    =\sum_{v\in V}\zeta(v)\tr[\rho W_v]^2
    =\sum_{v\in V}\zeta(v)\tr[(\rho\otimes\rho)(W_v\otimes W_v)]
    =\tr[(\rho\otimes \rho)M(\zeta)]\; ,
\end{align}
where we define the operator $M(\zeta):=\sum_{v\in V}\zeta(v)(W_v\otimes W_v)$. This immediately suggests to decompose the tensor product Hilbert space $\C^{2^{2n}}$ into its (anti-)symmetric subspaces, $\Sym^2(\C^{2^n})$ and $ \Lambda^2(\C^{2^n})$, defined as the $\pm 1$ eigenspaces of the operator $\SWAP:\C^{2^n}\otimes\C^{2^n}\ra\C^{2^n}\otimes \C^{2^n}$, $\SWAP(\psi\otimes\phi)=\phi\otimes\psi$ for all $\psi,\phi\in\C^{2^n}$, with respective projections $\Pi_\sym:=\frac{1}{2}(\id+\SWAP)$ and $\Pi_\Lambda:=\frac{1}{2}(\id-\SWAP)$,
\begin{align}\label{eq: (anti-)symmetric subspace}
    \Sym^2(\C^{2^n})
    &=\langle\psi\in\C^{2^{2n}}\mid\SWAP(\psi)=\psi\rangle &
    \Lambda^2(\C^{2^n})
    &=\langle\psi\in\C^{2^{2n}}\mid\SWAP(\psi)=-\psi\rangle \; .
\end{align}
The dimension of these subspaces are given by
\begin{align*}
    \dim(\Sym^2(\C^{2^n}))
    &=\tr[\Pi_\sym]
    =\binom{2^n+1}{2}
    =2^{n-1}(2^n+1) \\
    \dim(\Lambda^2(\C^{2^n}))
    &=\tr[\Pi_\Lambda]
    =\binom{2^n}{2}
    =2^{n-1}(2^n-1)\; .
\end{align*}
We will repeatedly use the following representation of SWAP in terms of (diagonal) Weyl operators,
\begin{align}\label{eq: Pauli SWAP}
    \SWAP
    =\frac{1}{2^n}\sum_{v\in V}W_v\otimes W_v\; .
\end{align}
Using these relations, we obtain the following expression for pure-state invariants.

\begin{lemma}\label{lm: invariant explicit}
    Let $I_\rho(\zeta)$ in Def.~(\ref{eq: I}) be a pure-state invariant. Then
    \begin{align}\label{eq: explicit invariant}
        I_\psi(\zeta)=
        \zeta(0)+\frac{1}{2^n+1}\sum_{0\neq v\in V}\zeta(v)\; .
    \end{align}
\end{lemma}

\begin{proof}
    Note that all non-identity terms in $M(\zeta)$ have the same overlap with $\Pi_\sym$ and $\SWAP$,
    \begin{align}\label{eq: constant overlap}
        \tr[\Pi_\sym (W_v\otimes W_v)]
        =\frac{1}{2}\tr[\SWAP (W_v\otimes W_v)]
        =\frac{1}{2^{n+1}}\sum_{v'\in V}\delta_{vv'}\tr[\one]^2
        =\frac{2^{2n}}{2^{n+1}}
        =2^{n-1}\; ,
    \end{align}
    for all $0\neq v\in V$, whereas $\tr[\Pi_\sym]=\frac{1}{2}(4^n+2^n)=2^{n-1}(2^n+1)$. Consequently, for any $\zeta\in\R^V$,
    \begin{align*}
        I_\psi(\zeta)
        &=\int_{\CP^{2^n-1}}d\psi \tr[(p_\psi\otimes p_\psi)M(\zeta)] \\
        &=\tr[\frac{\Pi_\sym}{\tr[\Pi_\sym]}M(\zeta)] \\
        &=\frac{1}{\tr[\Pi_\sym]}\left(\tr[\Pi_\sym]\zeta(0)+2^{n-1}\sum_{0\neq v\in V}\zeta(v)\right)
        =\zeta(0)+\frac{1}{2^n+1}\sum_{0\neq v\in V}\zeta(v)\; ,
    \end{align*}
    where $d\psi$ is the (unique $\U(d)$-invariant, normalised) Fubini-Study measure.
\end{proof}

Next, we express what it means for $I_\rho(\zeta)$ to be a pure-state invariant in terms of $M(\zeta)$.

\begin{lemma}\label{lm: symmetric decomposition}
    Let $M(\zeta)=\sum_{v\in V}\zeta(v)(W_v\otimes W_v)$ with $\zeta\in\R^V$. Then the following are equivalent:
    \begin{itemize}
        \item[(i)] $I_\rho(\zeta)=\tr[(\rho\otimes \rho)M(\zeta)]$ is a pure-state invariant.
        \item[(ii)] $M(\zeta)|_\sym:=\Pi_\sym M(\zeta) \Pi_\sym=\lambda_\zeta \Pi_\sym$ with $\lambda_\zeta=I_\psi(\zeta)$.
        \item[(iii)] $M(\zeta)=\lambda_\zeta \Pi_\sym\oplus M(\zeta)|_\Lambda$, where $M(\zeta)|_\Lambda:=\Pi_\Lambda M(\zeta)\Pi_\Lambda$, hence, $(M^k(\zeta))|_\sym=(M(\zeta)|_\sym)^k$ also defines a pure-state invariant for  $J_\psi(M^k(\zeta)):=\tr[p^{\otimes 2}_\psi M^k(\zeta)]=\lambda^k_\zeta$ for all $k\in\mathbb{N}$.
    \end{itemize}
    Moreover, $[M(\zeta),M(\zeta')]=0$ for all $\zeta,\zeta'\in\R^V$, that is, the $M(\zeta)$ share a common eigenbasis.
\end{lemma}

\begin{proof}
    (i)$\leftrightarrow$(ii) That $M(\zeta)|_\sym\sim\lambda_\zeta \Pi_\sym$ follows since the projectors $p^{\otimes 2}_\psi$ for $\psi\in\CP^{2^n-1}$ span the space of operators on $\Sym^2(\C^{2^n})$ (see Eq.~(\ref{eq: (anti-)symmetric subspace})). To show that the respective constants coincide, recall that $\Pi_\sym=\tr[\Pi_\sym]\int_{\CP^{2^n-1}} d\psi\ p^{\otimes 2}_\psi$ with $d\psi$ the normalised (unitarily-invariant) Fubini-Study measure on $\CP^{2^n-1}$ (equivalently, the push-forward of the normalised Haar measure on $\UH$), and $\tr[\Pi_\sym]=2^{n-1}(2^n+1)$. Since $I_\psi(\zeta)$ is a pure-state invariant,
    \begin{align*}
        I_\psi(\zeta)
        =\tr[(p_\psi\otimes p_\psi)M(\zeta)]
        &=\tr\left[\left(\int_{\CP^{2^n-1}} d\psi\ p^{\otimes 2}_\psi\right)M(\zeta)\right]\\
        &=\frac{\tr[\Pi_\sym M(\zeta)]}{\tr[\Pi_\sym]}
        =\frac{\tr[M(\zeta)|_\sym]}{\tr[\Pi_\sym]}
        =\frac{\tr[\lambda_\zeta \Pi_\sym]}{\tr[\Pi_\sym]}
        =\lambda_\zeta\; .
    \end{align*}
    (ii)$\ra$(iii) Clearly, $[M(\zeta),\SWAP]=0=[M(\zeta),\Pi_\sym]$ as this holds for every Pauli-diagonal term,
    \begin{align*}
        (W_v\otimes W_v)\SWAP
        =\SWAP(W_v\otimes W_v)\; .
    \end{align*}
    In other words, since $n$-qubit Pauli operators either commute or anti-commute, their squared tensor powers always commute. Since $M(\zeta)$ is a linear combination of diagonal Weyl operators it, too, is diagonal with respect to the decomposition $\C^{2^{2n}}=\Sym^2(\C^{2^n})\oplus\Lambda^2(\C^{2^n})$, which together with (ii) implies $M(\zeta)=\lambda_\zeta \Pi_\sym\oplus M(\zeta)|_\Lambda$. Consequently, $M^k(\zeta)=\lambda^k_\zeta \Pi_\sym\oplus (M(\zeta)|_\Lambda)^k$, hence,
    \begin{align*}
        J_\psi(M^k(\zeta))
        &:=\frac{1}{\tr[\Pi_\sym]}\tr[\Pi_\sym M^k(\zeta)]\\
        &=\frac{1}{\tr[\Pi_\sym]}\tr[\Pi_\sym (M^k(\zeta))|_\sym]\\
        &=\frac{1}{\tr[\Pi_\sym]}\tr[\Pi_\sym (M(\zeta)|_\sym)^k]
        =\frac{1}{\tr[\Pi_\sym]}\tr[\lambda^k_\zeta \Pi_\sym]
        =\lambda^k_\zeta\; ,
    \end{align*}
    for every $k\in\mathbb{N}$. The implication (iii)$\ra$(ii) is obvious.

    Finally, $[M(\zeta),M(\zeta')]=0$ for all $\zeta,\zeta'\in\R^V$ since from above $[W^{\otimes 2}_v,W^{\otimes 2}_w]=0$ for all $v,w\in V$.
\end{proof}

\textbf{Remarks.} For $\zeta=\ind_{V'}\in\zz^V_2$, we have $M(\ind_{V'})+M(\ind_{\overline{V'}})=2^n\SWAP$, and since $\tr[\rho^{\otimes 2}\SWAP]=\tr[\rho^2]=1$ for pure states, $I_\psi(V')+I_\psi(\overline{V'})=\lambda_{\ind_{V'}}+\lambda_{\ind_{\overline{V'}}}=\lambda_{\ind_V}=I_\psi(V)=2^n$.

Moreover, note that while $\mathrm{conv}\{p^{\otimes 2}_\psi\mid\psi\in\CP^{2^n-1}\}\subsetneq\mathrm{conv}\{\Psi\mid\Psi\in\CP^{2^{2n}-1}_\sym\}=\cD_\sym(\C^{2^{2n}})$, where we define $\CP^{2^{2n}-1}_\sym:=\CP^{2^{2n}-1}\cap\Sym^2(\C^{2^n})$ and $\cD_\sym(\C^{2^{2n}}):=\{\rho\in\cD(\C^{2^{4n}})\mid\mathrm{supp}(\rho)\subset\Sym^2(\C^{2^n})\}$, Lm.~\ref{lm: symmetric decomposition} implies that pure-state invariants $I_\psi(\zeta)$ lift to invariants on the full symmetric subspace.

\begin{corollary}\label{cor: symmetric extension}
    $I_\rho(\zeta)$ is a pure-state invariant if and only if $J_\Psi(M(\zeta))=I_\psi(\zeta)$ for all $\Psi\in\CP^{2^{2n}-1}_\sym$.
\end{corollary}

\begin{proof}
    This follows immediately from Lm.~\ref{lm: symmetric decomposition} (ii).
\end{proof}

In other words, the problem in (\ref{*}) is equivalent to the following generalisation:
\begin{center}
    For which $\zeta\in\R^V$ is $J_\Psi(M(\zeta))=\mathrm{const}$ for all $\Psi\in\CP^{2^{2n}-1}_\sym$?
\end{center}

Note that for any operator $M\in\mathrm{Mat}_{2^{2n}}(\C)$ with $M|_\sym=\lambda \Pi_\sym$ (but otherwise arbitrary), one has $J_\varrho(M):=\tr[\varrho M]=\mathrm{const}$ for all quantum states in the symmetric subspace $\varrho\in\cD_\sym(\C^{2^{2n}})$. From this perspective, Lm.~\ref{lm: symmetric decomposition} (iii) shows that the (diagonal) operators $M(\zeta)$ are strongly restricted. In Lm.~\ref{lm: anti-symmetric eigenvalues}, we will determine the eigenvalues $M(\zeta)|_\Lambda$ corresponding to pure-state invariants explicitly.

Our ensuing characterisation of pure-state invariants in Thm.~\ref{thm: all weighted quadratic invariants} is based on two steps: (i) following Lm.~\ref{lm: symmetric decomposition}, we choose a convenient basis of $\Sym^2(\C^{2^n})$ with respect to which $M(\zeta)$ is diagonal; (ii) we re-formulate problem (\ref{*}) in terms of a condition on the discrete Fourier transform over $\zz^V_2$.\\

\textbf{(i) Bell basis.} Define the linear map $\vect:\LH\ra\C^{2^n}\otimes\C^{2^n}$ by linear extension of
\begin{align}\label{eq: vec}
    \vect(W_v)
    =\kett{v}
    :=(\one\otimes W_v)\ket\Omega
    =\sum_{ij}(W_v)_{ji}|i\rangle\otimes|j\rangle\; ,
\end{align}
for all $v\in V$, and where $\ket{\Omega}=\sum_{i}\ket{i}\otimes\ket{i}$ is the unnormalised maximally entangled Bell state.\footnote{Alternatively, we could use the map $\vect^*:\LH\ra\C^{2^{2n}}$ defined by $\vect^*(W_v)=(W_v\otimes\one)\ket\Omega$.} Moreover, define the transpose sign $\tau:V\ra\{\pm 1\}$ by $W^T_v=\tau(v)W_v$, and note that $\tau=(-1)^{q_W}$, where $q_W(v)=a\cdot b$ as in Eq.~(\ref{eq: Weyl operators}). The following lemma collects some basic properties of the map $\vect$.

\begin{lemma}\label{lm: vec}
    \begin{itemize}
        \item[(i)] $\vect$ preserves orthogonality, that is, $\braa{w}\kett{v}=\tr[W^\dagger_wW_v]=2^n\delta_{vw}$, hence, $\{\vect(v)\mid v\in V\}$ is a basis of $\C^{2^{2n}}$, called the \emph{Bell basis},
        \item[(ii)] $\SWAP\kett{v}=\pm\kett{v}$ if and only if $W^T_v=\pm W_v$, hence,
        \begin{align*}
            \Sym^2(\C^{2^n})
            &=\langle\kett{v}\mid \tau(v)=+1\rangle &
            \Lambda^2(\C^{2^n})
            &=\langle\kett{v}\mid \tau(v)=-1\rangle\; ,
        \end{align*}
        \item[(iii)] $(Q\otimes R)\vect(P)=\vect(RPQ^T)$.
    \end{itemize}
\end{lemma}

\begin{proof}
    (i) By definition in Eq.~(\ref{eq: vec}), the identity $\langle\Omega|\one\otimes O|\Omega\rangle=\tr[O]$ and Eq.~(\ref{eq: Weyl commutation relations}),
    \begin{align*}
        \braa{w}\kett{v}
        =\langle\Omega|(\one\otimes W^\dagger_wW_v)\kett{\Omega}
        =\tr[W^\dagger_wW_v]
        =2^n\delta_{vw}\; .
    \end{align*}
    
    (ii) Since $\SWAP|\Omega\rangle=|\Omega\rangle$, this follows from
    \begin{align*}
        \SWAP\kett{v}
        =\SWAP(\one\otimes W_v)|\Omega\rangle
        =(W_v\otimes\one)\SWAP|\Omega\rangle
        =(W_v\otimes\one)|\Omega\rangle
        =\tau(v)\kett{v}\; .
    \end{align*}

    (iii) is immediate from the definition of $\vect$ in Eq.~(\ref{eq: vec}).
\end{proof}

With Lm.~\ref{lm: vec}, we can write the decomposition in Eq.~(\ref{eq: (anti-)symmetric subspace}) as a span of Pauli operators under $\vect$, that is, in terms of the Bell basis. Importantly, $M(\zeta)$ is diagonal with respect to this basis,
\begin{align}\label{eq: diagonal action in Bell basis}
    (W_v\otimes W_v)\vect(W_w)
    =\vect(W_vW_wW^T_v)
    =\chi_v(w)\tau(v)\vect(W_w)\; ,
\end{align}
where we define the characters $\chi_w:V\ra\{\pm 1\}$ by $\chi_w(v)=(-1)^{\omega(v,w)}$ for all $v,w\in V$.

This yields the following spectral condition for pure-state invariants.

\begin{lemma}\label{lm: symmetric eigenvalues}
    $I_\rho(\zeta)$ is a pure-state invariant if and only if for all $s\in\cT_+$,
    \begin{align*}
        \lambda_\zeta
        =\sum_{v\in V}\zeta(v)\chi_v(s)\tau(v)\; .
    \end{align*}
\end{lemma}

\begin{proof}
    By Lm.~\ref{lm: symmetric decomposition}, $I_\rho(\zeta)$ is a pure-state invariant if and only if $M(\zeta)\kett{s}=\lambda_\zeta\kett{s}$ for all $s\in\cT_+$, and by Lm.~\ref{lm: vec} (iii), $M(\zeta)\kett{s}
    =\sum_{v\in V}\zeta(v)(W_v\otimes W_v)\kett{s}
    =\sum_{v\in V}\zeta(v)\chi_s(v)\tau(v)\kett{s}$.
\end{proof}

In particular, $M(\zeta)$ is diagonal in the Bell basis. More generally, define the linear operators $L_v\in\cL(\C^{2^{2n}})$ by (linear extension of) $L_v(W_w):=W_vW_wW^T_v$ for all $v,w\in V$. From Lm.~\ref{lm: symmetric eigenvalues}, $M(\zeta)\vect(W_s)=\sum_{v\in V}\zeta(v)\vect(L_v(W_s))=\lambda_\zeta\vect(W_s)$. Moreover, $L_vL_{v'}=(-1)^{\omega(v,v')}L_{v+v'}$, hence, $v\mapsto L_v$ defines a faithful (projective) representation of $V=\zz^{2n}_2$, which decomposes as
\begin{align*}
    \C^V
    =\oplus_{w\in V}V_w
    =(\oplus_{s\in\cT_+}V_s)\oplus(\oplus_{a\in\cT_-}V_a)\; ,
\end{align*}
for one-dimensional eigenspaces $V_w=\langle W_w\rangle$, and where we defined $\cT_\pm:=\{w\in V\mid\tau(w)=\pm 1\}$.\\

\textbf{(ii) Fourier transform.} For any real-valued function $f\in\R^V$, define its Fourier transform by
\begin{align}\label{eq: Fourier transform}
    \hf(w):=\hf(\chi_w)
    =\frac{1}{\sqrt{|V|}}\sum_{v\in V}\chi_w(v)f(v)\; ,
\end{align}
for all $\chi_w(\cdot)=(-1)^{\omega(w,\cdot)}\in\hV\cong V$, where $\hV$ denotes the space of characters on $V$, that is, its Pontryagin dual. Eq.~(\ref{eq: Fourier transform}) defines a linear isomorphism between $\R^V$ and $\R^\hV\cong\R^V$ such that $\widehat{\hf}=f$.\\

\textbf{Characterisation.} Using the above ingredients, we characterise pure-state invariants in Def.~\ref{def: I}.

\begin{theorem}\label{thm: all weighted quadratic invariants}
    $I_\rho(\zeta)$ is a pure-state invariant if and only if $\zeta=\beta_0\zeta_0+\sum_{a\in\cT_-}\beta_a\xi_a$ for some $\beta_0,\beta_a\in\R$, and where $\zeta_0=\delta_{v0}$ and $\xi_a=\chi_a\tau$ for $a\in\cT_-=\{w\in V\mid\tau(w)=-1\}$.
\end{theorem}

\begin{proof}
    By Lm.~\ref{lm: symmetric eigenvalues}, $I_\rho(\zeta)$ is a pure-state invariant if and only if $\sum_{v\in V}\chi_s(v)\zeta(v)\tau(v)=\lambda_\zeta$ for all $s\in\cT_+$, equivalently (by Eq.~(\ref{eq: Fourier transform})) if and only if $\widehat{f}_\zeta(s)=\frac{\lambda_\zeta}{2^n}$ for all $s\in\cT_+$ with $f_\zeta=\zeta\tau$. Since the Fourier transform is a linear map, this defines a set of linear equations in its entries $v\in V$, and since $|\cT_+|=\frac{2^n(2^n+1)}{2}=2^{n-1}(2^n+1)$, the space of solutions has dimension $|V|+1-|\cT_+|=2^{n-1}(2^n-1)+1$.\footnote{Note that the value of the proportionality constant $\lambda_\zeta$ is a free parameter in the space of solutions. This is reflected in the existence of two solutions to the inhomogeneous equation $\widehat{f}_\zeta(s)=2^{-n}\lambda_\zeta$ for all $s\in\cT_+$, which may be identified with the identity and SWAP operators. One of which can be taken to fix $\lambda_\zeta$, equivalently, to fix a normalisation for $\zeta$.}
    
    We show that $\zeta_0:=\delta_{v0}$ and $\{\zeta=\xi_a\}_{a\in\cT_-}$ with $\xi_a=\chi_a\tau$ parametrise this space. Indeed,
    \begin{align*}
        \widehat{f_{\xi_a}}(w)
        =\frac{1}{2^n}\sum_{v\in V}\chi_w(v)f_{\xi_a}(v)
        =\frac{1}{2^n}\sum_{v\in V}\chi_w(v)\tau^2(v)\chi_a(v)
        =\frac{1}{2^n}\sum_{v\in V}\chi_{w+a}(v)
        =\begin{cases}
            2^n &\mathrm{if}\ w=a\\
            0 &\mathrm{otherwise}
        \end{cases}\; ,
    \end{align*}
    which for $w\in\cT_+$ implies that $\xi_a$ solves the homogeneous equation $\widehat{f_{\xi_a}}(s)=0$ for all $s\in\cT_+$. Moreover, $\zeta_0$ defines a solution to the inhomogeneous equation, as $\widehat{f_{\zeta_0}}(w)=\frac{1}{2^n}\sum_{v\in V}\chi_w(v)\delta_{v0}(v)\tau(v)=\frac{1}{2^n}$ for all $w\in V$. Finally, we need to show that these solutions are linearly independent. Let $0=\beta_0\delta_{v0}+\sum_{a\in\cT_-}\beta_a\xi_a$ for $\beta_0,\beta_a\in\R$, multiplying by $\tau$ and taking the Fourier transform yields,
    \begin{equation*}
        \beta_0\delta_{v0}+\sum_{a\in\cT_-}\beta_a\xi_a=0
        \quad\stackrel{\cdot\tau}{\Longrightarrow}\quad
        \beta_0\delta_{v0}+\sum_{a\in\cT_-}\beta_a\chi_a=0
        \quad\stackrel{\widehat{\cdot}}{\Longrightarrow}\quad
        \frac{\beta_0}{2^n}+2^n\sum_{a\in\cT_-}\beta_a\delta_{aw}=0\quad\forall w\in V\; ,
    \end{equation*}
    where we used that $\widehat{f_{\delta_{v0}}}=\frac{1}{2^n}$ and $\widehat{f_{\chi_a}}(w)=2^n 
    \delta_{aw}$ from above. Evaluating at $w\in\cT_+$ thus implies $\beta_0=0$. For the remaining parameters, we use the linear independence of the characters $\chi_a$,
    \begin{equation*}
        \sum_{a\in\cT_-}\beta_a\xi_a=0
        \quad\stackrel{\cdot\tau}{\Longrightarrow}\quad
        \sum_{a\in\cT_-}\beta_a\chi_a=0
        \quad\Longleftrightarrow\quad \beta_a=0\quad\forall a\in\cT_-\; .\qedhere
    \end{equation*}
\end{proof}

Note that the function $\zeta_1=\ind_V$ also solves the inhomogeneous equation. To see this, let $v=(x,z),w=(x',z')\in\zz^{2n}_2$ and recall from Eq.~(\ref{eq: Weyl operators}) that $\chi_w(v)=(-1)^{\omega(v,w)}=(-1)^{x'\cdot z-x\cdot z'}$ as well as $\tau(v)=(-1)^{x\cdot z}$ (where $x\cdot z=\sum_{k=1}^nx_kz_k$ denotes the inner product in $\zz^n_2$). From this, we obtain
\begin{align}\label{eq: inhomogenous solution}
    \widehat{f_{\zeta_1}}(w)
    =\frac{1}{2^n}\sum_{v\in V}\chi_w(v)\tau(v)
    =\frac{1}{2^n}\sum_{x,z\in\zz^n_2}(-1)^{x'\cdot z+x\cdot (z-z')}
    =(-1)^{x'\cdot z'}
    =\tau(w)\; ,
\end{align}
where we used character orthogonality (in $\widehat{\zz^n_2}$) in the third step. In terms of the generating set of Thm.~\ref{thm: all weighted quadratic invariants}, it is given by the relation (using the above), $\sum_{a\in\cT_-}\xi_a=\frac{1}{2}(4^n\zeta_0-2^n\zeta_1)$.

We close the present section by determining the full operator form of the $M(\zeta)$. Since by Thm.~\ref{thm: all weighted quadratic invariants}, $M(\zeta)$ is a linear combination of the operators $M(\zeta_1)=2^n\SWAP$ and $M^\pm_a:=M(\zeta^\pm_a)$ for $\zeta^\pm_a:=\frac{1}{2}(1\pm\xi_a)$, it is sufficient to determine the form of the latter. Moreover, we know from Lm.~\ref{lm: symmetric decomposition} that for pure-state invariants, $M(\zeta)=\lambda_\zeta \Pi_\sym\oplus M(\zeta)|_\Lambda$, hence, we are left to compute $M(\xi_a)|_\Lambda$.

\begin{lemma}\label{lm: anti-symmetric eigenvalues}
    For every $a\in\cT_-$, we have $M^\pm_a=\frac{1}{2}(2^n\SWAP\pm 4^n\Pi_a)$ for $\Pi_a:=\frac{1}{2^n}\kett{a}\braa{a}$, explicitly
    \begin{align}\label{eq: M-spectrum}
        M^\pm_a\kett{w}
        =\ \begin{cases}
            2^{n-1}\kett{w} &\mathrm{if}\ w\in\cT_+\\
            -2^{n-1}\kett{w} &\mathrm{if}\ a\neq w\in\cT_-\\
            -2^{n-1}(1\mp 2^n)\kett{w} &\mathrm{if}\ a=w\in\cT_-
        \end{cases}\; .
    \end{align}
\end{lemma}

\begin{proof}
    The symmetric case of Eq.~(\ref{eq: M-spectrum}) follows from Lm.~\ref{lm: symmetric decomposition}. For $w\in\cT_-$, we similarly compute
    \begin{align*}
        M^\pm_a\kett{w}
        &=\sum_{v\in V}\zeta^\pm_a(v)(W_v\otimes W_v)\kett{w}\\
        &=\frac{1}{2}\sum_{v\in V}\left(1\pm\chi_a(v)\tau(v)\right)\chi_w(v)\tau(v)\kett{w}\\
        &=\frac{1}{2}\left(\sum_{v\in V}\chi_w(v)\tau(v)\pm\sum_{v\in V}\chi_{a+w}(v)\right)\kett{w}
        =\frac{1}{2}(2^n\tau(w)\pm 4^n\delta_{a,w})\kett{w}\; ,
    \end{align*}
    where we used Eq.~(\ref{eq: inhomogenous solution}) together with character orthogonality in the last step. The assertion then follows since when evaluated on the Bell basis, $\SWAP\kett{w}=\tau(w)\kett{w}$ and $\Pi_a\kett{w}=\delta_{a,w}\kett{w}$.
\end{proof}

Note that $M^+_a+M^-_a=2^n\SWAP$. On the other hand, using that $\sum_{a\in\cT_-}\Pi_a=\Pi_\Lambda$, one finds
\begin{align*}
    \sum_{a\in\cT_-}M(\xi_a)
    =\sum_{a\in\cT_-}(M^+_a-M^-_a)
    =4^n\Pi_\Lambda\; .
\end{align*}

\subsection{Boolean pure-state invariants}\label{sec: Boolean pure-state invariants}

By Thm.~\ref{thm: all weighted quadratic invariants}, $\zeta\in\R^V$ defines a pure-state invariant if and only if, up to a constant term $\zeta_0$ (or $\zeta_1$), it is a linear combination of quadratic sign functions $\xi_a$ for $a\in\cT_-=\{v\in V\mid\tau(v)=-1\}$. In this section, we consider Boolean pure-state invariants, that is, those with Boolean-valued weights, $\zeta\in\zz^V_2$. In other words, we ask which pure-state invariants correspond with indicator functions $\zeta=\ind_X$ for $X\subset\aP_n$.

\begin{corollary}\label{cor: Boolean quadratic invariants}
    For every $a\in\cT_-=\{w\in V\mid\tau(w)=-1\}$, $\zeta^\pm_a$ defines a Boolean pure-state invariant.
\end{corollary}

\begin{proof}
    Clearly, every $\zeta^\pm_a$ is Boolean-valued and is a pure-state invariant by Thm.~\ref{thm: all weighted quadratic invariants}.
\end{proof}

In the next section, we will provide an explicit representation for the pure-state invariants in Cor.~\ref{cor: Boolean quadratic invariants}, and collect various of their properties. However, as we will see shortly, not every Boolean pure-state invariant is of this form. Algebraically, this may be seen as a consequence of the fact that while the characters $\{\chi_a\}_{a\in\cT_-}$ are linearly independent, the space $\cT_-\subset V$ clearly is not.\footnote{Note that given a linearly independent set $\{w_i\}_{i=1}^k\in V$, the map $v\mapsto(\chi_{w_i}(v))_{i=1}^k\in(\pm 1)^k$ has rank $k$. Consequently, it is surjective, that is, for every Boolean vector $b\in\zz^k_2$ there exists $v\in V$ such that $\chi_{w_i}(v)=(-1)^{b_i}$ for all $i\in[k]$. Now, if $\cT_-$ was linearly independent, we could apply this argument to conclude that $\beta_a=0$ for all but at most a single $a\in\cT_-$ in the decomposition $\zz^V_2\ni\zeta=\sum_{a\in\cT_-}\beta_a\xi_a=\sum_{a\in\cT_-}\beta_a\chi_a\tau$ in Thm.~\ref{thm: all weighted quadratic invariants}.
However, since $|\cT_-|=2^{n-1}(2^n-1)>2n=\dim(V)$, $\cT_-$ is clearly not linearly independent for $n\geq 2$.}

From a more geometric perspective, we are intersecting a linear space of quadratic sign functions on $V$ with the nonlinear Boolean hypercube $\{0,1\}^V\subset\R^V$. Besides the Boolean invariants $\zeta_a^\pm=\frac{1}{2}(1\pm\xi_a)$ arising from a single quadratic sign function, there may therefore exist Boolean-valued combinations involving more than one element of $\{\xi_a\}_{a\in\cT_-}$. To see that this is indeed the case, we first translate the problem into one in the incidence geometry of symplectic polar spaces \cite{SanigaPlanat2007,SanigaPlanat2008}.

Recall that a subspace $U\subset V$ of a symplectic vector space $(V,\omega)$ is called \emph{isotropic} if $\omega|_U=\omega_{U\times U}=0$. A maximal isotropic subspace is called \emph{Lagrangian}. We will write $\Lag(V)$ for the space of Lagrangian subspaces of $V$. The space of one-dimensional isotropic subspaces is the \emph{projective symplectic space $\PG(V,2)$}. Note that since $\zz_2$ contains a single non-zero element only, projectivisation simply removes the zero point, and we will thus identify $v$ with $\langle v\rangle$ for every $0\neq v\in V$. In particular, $\PG(V,2)\cong V^\circ=V\backslash\{0\}$. More generally, we will write $U^\circ:=U\backslash\{0\}$ and in slight abuse of terminology refer to both $U<V$ as well as $U^\circ\subset\PG(V,2)$ as isotropic subspaces.

\begin{definition}\label{def: symplectic polar space}
    Let $(V=\zz^{2n}_2,\omega)$ be a symplectic vector space. Then the \emph{symplectic polar space $W(2n-1,2)$} is the incidence geometry of the projectivised isotropic subspaces of $V$.
\end{definition}

Viewed in terms of Pauli operators, $W(2n-1,2)$ corresponds with the space of all Abelian subgroups $A\subset\aP_n$ (up to phase), with the identity removed. In particular, $\aP^\circ_n\cong\PG(V,2)$ (as sets) and Lagrangian subspaces correspond with maximal Abelian subgroups. Symplectic polar spaces are an active area of research \cite{BambergKellyLawPenttila2007,CeriaPavese2022,DeBeuleMannaertSmaldore2024,Frembs2026_unextendibility,SanigaPlanat2007,Taylor1992,Thas2001}. We list some structures of interest to us.

\begin{definition}\label{def: spreads}
    A \emph{partial spread of $PG(V,2)$} is a collection $\{L^\circ_k\}_{k=1}^K$ with $L_k\in\Lag(V)$ such that $L^\circ_k\cap L^\circ_l=\emptyset$ whenever $k\neq l$. A \emph{(complete) spread} is one with $\bigcupdot_{k=1}^K L^\circ_k=\PG(V,2)$ and thus $K=2^n+1$.
\end{definition}

Equivalently, a spread partitions $\aP^\circ_n$ into complementary maximal Abelian subgroups (with the identity removed), where complementarity here means that the common eigenbases of the Pauli operators in the respective subgroups are mutually unbiased \cite{Ivonovic1981,WoottersFields1989,LawrenceBruknerZeilinger2002}. This property makes (partial) spreads natural structures for the study of (entropic) uncertainty relations \cite{WehnerWinter2010,BallesterWehner2007,KettererGuehne2020,FrembsHoehnNataleWever2026b}. We expand on the relation between complementarity and entropic uncertainty relations in Sec.~\ref{sec: uncertainty relations}.

\begin{definition}\label{def: ovoids}
    A subset $O\subset\PG(V,2)$ is called an \emph{$m$-ovoid} if $|O\cap L^\circ|=m$ for every Lagrangian subspace $L\in\Lag(V)$.
\end{definition}

Expressed in terms of Pauli operators, an $m$-ovoid is the union of the intersections of $\aP^\circ_n$ with all Lagrangians whenever all these intersections have the same cardinality. It turns out that this purely incidence-geometric property characterises pure-state invariants with Boolean coefficients.

\begin{theorem}\label{thm: all Boolean pure-state invariants}
    Let $X\subset V^\circ$. Then $I_\psi(X)=m$ for all $\psi\in\CP^{2^n-1}$ if and only if $X$ is an $m$-ovoid.
\end{theorem}

\begin{proof}
    As above, we will identify $V^\circ\cong\PG(V,2)$. Assume first that $I_\psi(X)$ is a pure-state invariant, and let $\psi_A\in\CP^{2^n-1}$ be a stabiliser state with stabiliser group $A<\aP_n$, which corresponds to a Lagrangian subspace $L_A\in\Lag(V)$ (via the relation $\tr[p_{\psi_A}W_v]^2=1$ if and only if $v\in L_A$). Then \begin{align*}
        I_{\psi_A}(X)
        =\sum_{v\in X}\tr[p_{\psi_A}W_v]^2
        =|X\cap L^\circ_A|\; ,
    \end{align*}
    and since this holds for all stabiliser states, we find $I_\psi(X)=|X\cap L^\circ|=m$ for all Lagrangian subspaces $L\in\Lag(V)$, hence, $X$ is an $m$-ovoid by Def.~\ref{def: ovoids}.

    Conversely, let $X\subset V^\circ$ be an $m$-ovoid. Define the operator $M(\ind_X):=\sum_{v\in X}W_v\otimes W_v$. By Thm.~\ref{thm: all weighted quadratic invariants}, $X$ is a pure-state invariant if and only if $M(\ind_X)$ acts by a scalar on the symmetric subspace, which in the Bell basis (see Sec.~\ref{sec: characterisation}) reads  $M(\ind_X)\kett{s}=\lambda_s\kett{s}$ with $\lambda_s=\lambda$ for all $s\in\cT_+$ (that is, all symmetric Pauli operators, $W^T_s=W_s\in\aP_n$). For every such operator consider the quadratic sign function $\xi_s=\chi_s\tau$ with $s\in\cT_+$, set $\xi_s=:(-1)^{q_s}$ and let $J_s:=\{v\in V^\circ\mid q_s(v)=0\}$. Then the coefficients in Lm.~\ref{lm: symmetric eigenvalues} read
    \begin{align}\label{eq: lambda_s}
        \lambda_s
        =\sum_{v\in X}\xi_s(v)
        =|X\cap J_s|-|X\backslash J_s|
        =2|X\cap J_s|-|X|\; .
    \end{align}
    Clearly, we have $|X|=m(2^n+1)$, since every complete spread partitions $V^\circ\cong\PG(V,2)$ into $2^n+1$ disjoint elements, and, by definition of an $m$-ovoid, $X$ intersects every Lagrangian in this spread in $m$ elements. We thus need to determine the first term in the last expression.
    
    To do so, we will need the cardinality of $J_s$. Recall that $\tau(v)=(-1)^{q_W(v)}=(-1)^{a\cdot b}=(-1)^{\sum_{i=1}^n a_ib_i}$ for $V\ni v=\sum_{i=1}^na_ie_i+b_if_i$ expressed with respect to a symplectic basis $(e_i,f_j)_{i,j=1}^n$, defined by the relations $\omega(e_i,e_j)=0=\omega(f_i,f_j)$ and $\omega(e_i,f_j)=\delta_{ij}$ (e.g. Pauli-$Z$ and Pauli-$X$ operators in the Weyl representation in Eq.~(\ref{eq: Weyl operators})). Now, note that $\xi_s=\tau\circ t_s$ for $t_s(v)=v+\omega(v,s)s$. Indeed,
    \begin{align*}
        \tau\circ t_s(v)
        =\tau(v+\omega(v,s)s)
        =\tau(v)\chi_s(v)\tau(s)
        =\tau(v)\chi_s(v)
        =\xi_s(v)\; ,
    \end{align*}
    where we used $\tau(s)=1$ for $s\in\cT_+$ and $\tau(u+v)=(-1)^{\omega(u,v)}\tau(u)\tau(v)$. Moreover, since $t_s$ is linear and invertible (note that $t^2_s=\mathrm{id}$), it induces a basis change, and it is straightforward to check that $e'_i=t_s(e_i)$ and $f'_j=t_s(f_j)$ defines a symplectic basis. Consequently, $\xi_s(v)=(-1)^{a'\cdot b'}=(-1)^{\sum_{i=1}^na'_ib'_i}$ for $v=\sum_{i=1}^na'_ie'_i+b'_if'_i$ expressed in the transformed symplectic basis $\{e'_i,f'_j\}_{i,j=1}^n$. Using the expression of $\xi_s$ in this basis, one easily counts the size of the set $J_s$: for $a'=0$, every vector $b'\in\zz^n_2$ satisfies $a'\cdot b'=0$, while for every $a'\neq 0$ there are $2^{n-1}$ vectors $b'$ such that the sum vanishes, hence, $|J_s|=(2^n-1)+(2^n-1)2^{n-1}=(2^{n-1}+1)(2^n-1)$ (since $0\notin J_s$).
    
    Now, for the first term in Eq.~(\ref{eq: lambda_s}), let $\Lag(J_s)=\{L\in\Lag(V)\mid L^\circ\subset J_s\}$ and define the incidence sets:
    \begin{align*}
        \Omega(J_s,\Lag(J_s))
        &:=\{(v,L)\mid v\in J_s,\ L\in\Lag(J_s)\} \\
        \Omega(X\cap J_s,\Lag(J_s))&:=\{(v,L)\mid v\in X\cap J_s,\ L\in\Lag(J_s)\}\; .
    \end{align*}
    We can count their respective sizes in two ways. First, summing over Lagrangian subspaces yields,
    \begin{align*}
        |\Omega(J_s,\Lag(J_s))|
        =\sum_{L\in\Lag(J_s)}|L^\circ|=(2^n-1)|\Lag(J_s)|\; ,
    \end{align*}
    while summation over points results in the relation,
    \begin{align*}
        |\Omega(J_s,\Lag(J_s))|
        =\sum_{v\in J_s}c_s
        =c_s|J_s|\; ,
    \end{align*}
    since, by symmetry of $J_s$, every point in $J_s$ lies in the same number $c_s$ of Lagrangian subspaces in $\Lag(J_s)$. To see this, for every $v\in J_s$, consider the quotient $v^\perp/\langle v\rangle$, where $v^\perp=\{u\in V\mid \omega(u,v)=0\}$ is the symplectic orthogonal complement of $v$ in $V$. Then for $v'=u+v\in v^\perp$ (that is, $u\in v^\perp$),
    \begin{align*}
        \xi_s(v')
        =\tau(v')\chi_s(v')
        =\tau(u+v)\chi_s(u+v)
        =\tau(u)\tau(v)(-1)^{\omega(u,v)+\omega(u,s)+\omega(v,s)}
        =\tau(u)\chi_s(u)
        =\xi_s(u)\; ,
    \end{align*}
    since $\omega(u,v)=0$ and $\xi_s(v)=1$ as $v\in J_s$. Hence, $\xi_s$ is constant on the cosets of $\langle v\rangle$ in $v^\perp$, and the restriction $\xi^{v^\perp}_s:v^\perp/\langle v\rangle\ra\{\pm 1\}$ thus well-defined. It follows that the number of Lagrangian subspaces in $\Lag(J_s)$ incident with $v$ equals the number of maximal totally singular subspaces of the quadratic space $(v^\perp/\langle v\rangle,q_s)$. Finally, since the symplectic form $\omega$ is nondegenerate, $(v^\perp)^\perp=\langle v\rangle$ which implies that its restriction to $v^\perp/\langle v\rangle$ is also non-degenerate, and we find that without loss of generality for $e_1=v$, $\xi^{v^\perp}_s(\tv)=(-1)^{\ta'\cdot\tb'}=(-1)^{\sum_{i=2}^n a'_ib'_i}$ for the symplectic basis (from above). Clearly, this implies that the number of Lagrangian subspaces incident with a point $v\in J_s$ is in fact independent of $v$.
    
    Using this symmetry for the second incidence set, we similarly obtain the relation
    \begin{align*}
        |\Omega(X\cap J_s,\Lag(J_s))|
        =\sum_{v\in X\cap J_s}c_s
        =c_s|X\cap J_s|\; .
    \end{align*}
    Moreover, summing instead over Lagrangian subspaces, and using the defining property of $m$-ovoids,
    \begin{align*}
        |\Omega(X\cap J_s,\Lag(J_s))|
        =\sum_{L\in\Lag(J_s)}|X\cap L^\circ|
        =m|\Lag(J_s)|\; .
    \end{align*}
    Combining the above incidence counts with the cardinality of $J_s$, we arrive at
    \begin{align*}
        |X\cap J_s|
        =m\frac{|\Lag(J_s)|}{c_s}
        =m\frac{|J_s|}{2^n-1}
        =m(2^{n-1}+1)\; ,
    \end{align*}
    and inserting into Eq.~(\ref{eq: lambda_s}) thus yields the constant value $\lambda_s=2m(2^{n-1}+1)-m(2^n+1)=m$.
\end{proof}

Thm.~\ref{thm: all Boolean pure-state invariants} asserts that the problem of characterising all subsets $X\subset V^\circ$ such that $I_\psi(X)=\mathrm{const}$ for all $\psi\in\CP^{2^n-1}$ is equivalent to the problem of characterising all $m$-ovoids, and thus highlights the physical significance of this finite geometry problem. Unfortunately, to date, there is no complete characterisation of $m$-ovoids; nor will we make progress on this problem here. Instead, in the next section we will single out the specific cases in Cor.~\ref{cor: Boolean quadratic invariants} from a physically motivated assumption, and focus on those in more detail below. However, before doing so, we first review a few known results about $m$-ovoids and translate them into the setting of $n$-qubit Pauli operators and pure-state invariants (see also Ref.~\cite{SanigaPlanat2007,SanigaPlanat2008}).

First, note that, by Cor.~\ref{cor: Boolean quadratic invariants}, every $a\in\cT_-$ gives rise to an $2^{n-1}$-ovoid for $\zeta^-_a$ and a $(2^{n-1}-1)$-ovoid for $\zeta^+_a$. For $n=2$ every $m$-ovoid (other than $V^\circ$ and up to complements) is of this form, however, already for $n=3$, there are $13,384$ distinct $3$-ovoids, falling into three isomorphism classes \cite{BambergKellyLawPenttila2007,CeriaPavese2022}. Of these, only $28$ belong to the family in Cor.~\ref{cor: Boolean quadratic invariants}. From the perspective of the reconstruction in Ref.~\cite{HoehnWever2017}, the two-qubit case is thus highly exceptional!  (We leave the extension of the reconstruction argument under the additional assumption of update closure to future work.) Second, the fact that for the three-qubit case every $m$-ovoid has $m=2^{n-1}-1=3$  (up to complements), yet this does not generalise to arbitrary $n$ either: for instance, $m$-ovoids with $m=2^{n/2}-1$ exist for every even $n\geq 2$ \cite{Thas1981,Kelly2007}. Third, from the perspective of complementarity and entropic uncertainty relations, it is natural to collect the Pauli observables of an ovoid into (observables built from) mutually commuting ones \cite{FrembsHoehnNataleWever2026b}. In particular, given any (complete) spread $\{L^\circ_i\}_{i=1}^{2^n+1}$, the defining property of an $m$-ovoid $X\subset V^\circ$ implies $|X\cap L^\circ|=m$. In order to interpret these as sets of compatible observables, via $O=\sum_{v\in L}\chi(v)W_v$ for $\chi:L\ra\{\pm 1\}$ a character, it is natural to demand that the sets $X\cap L_i$ are themselves isotropic subspaces for some choice of spread $\{L^\circ_i\}_i$. Such $m$-ovoids correspond with $h$-systems.

\begin{definition}\label{def: h-systems}
    An \emph{h-system of $W(2n-1,2)$} for $0\le h\leq n-1$ is a collection $\{U^\circ_i\}_{i=1}^{2^n+1}$ of isotropic subspaces $U_i<V$ of dimension $h+1$ that are pairwise opposite, that is, $U_i\cap U^\perp_j=\{0\}$ for all $i\neq j$.\footnote{Here, $h$ is the projective dimension of the projective subspaces $U^\circ_i$.}
\end{definition}

By Ref.~\cite{ShultThas1994}, the point set covered by an $h$-system of a finite classical polar space is an $m$-ovoid with $m=(q^{h+1}-1)/(q-1)$; for $q=2$ this gives $m=2^{h+1}-1$. In particular, $h=0$ recovers ovoids with $m=1$ and $h=n-1$ (complete) spreads with $m=2^n-1$. $h$-systems thus form a sort of hybrid between Def.~\ref{def: spreads} and Def.~\ref{def: ovoids}: an $m$-ovoid is a purely incidence-geometric object, it demands no additional algebraic properties with respect to the ambient linear structure; by contrast, (similar to $V$) an $h$-system can be partitioned into isotropic subspaces \cite{ShultThas1994,Thas1989,BambergKellyLawPenttila2007}. While every pure-state invariant $\zeta_a^+$ in Cor.~\ref{cor: Boolean quadratic invariants} is an $h$-system, the converse is not true. In particular, this applies to $m$-ovoids arising from picking pairwise opposite isotropic subspaces (of dimension $\log_2(m+1)$) in every Lagrangian subspace of a complete spread. As with $m$-ovoids, the characterisation of $h$-systems remains open. Despite the combinatorial complexity involved in determining $h$-systems and $m$-ovoids in general (equivalently, Boolean pure-state invariants by Thm.~\ref{thm: all Boolean pure-state invariants}), in the next section we will fully characterise the subfamily of pure-state invariants that are not merely kinematical, but are also compatible with a notion of dynamical state update.

\subsection{Update-closed pure-state invariants}\label{sec: update-closed Boolean invariants}

By Thm.~\ref{thm: all Boolean pure-state invariants}, Boolean pure-state invariants correspond with $m$-ovoids, whose characterisation remains in large part uncharted territory. In this section, we study Boolean pure-state invariants under an additional, physically motivated constraint for which a full characterisation is possible.

Since Boolean pure-state invariants correspond with subsets of Hermitian Pauli observables $X\subset\aP_n$, one may - on a purely formal level - require additional structure of the Pauli group to be preserved. From a mathematical perspective, the choice of which structure to preserve is somewhat arbitrary. Instead, we seek a physical motivation, for which we take inspiration from Ref.~\cite{Hoehn2017,HoehnWever2017}, where the two-qubit pure-state invariants are interpreted as conserved `informational charges'. This interpretation suggests to enrich the purely kinematical constraints of Def.~\ref{def: I} by an additional dynamical compatibility condition: a (discrete) notion of \emph{information update}.\footnote{In analogy with the reconstruction in Ref.~\cite{Hoehn2017,HoehnWever2017}, one can also demand a notion of compatibility with respect to continuous symmetry transformations. We will defer this to future work.}

Conceptually, the idea is the following: assume we are given a stabiliser state $\psi_A$ corresponding to some Abelian stabiliser subgroup $A<\aP_n$ (up to phase). With respect to a pure-state invariant $I_\psi(X)$, this state then provides information about the observables contained in $X\subset\aP_n$, namely maximal information on the observables in $A\cap X$ and no information on all others. Now, assume we perform a measurement corresponding to a (nondegenerate) observable whose eigenspaces are the joint stabiliser eigenstates of the Pauli observables of another maximal Abelian subgroup. Such a measurement has two effects: it provides maximal information about the outcomes of the measured observable, in particular, those that also lie in $X$, and, in turn, it must erase part of the information previously held about $X$. The reason is simply that, since $I_\psi(X)$ is a pure-state invariant, the overall information about the total system remains constant, resulting in a necessary trade-off between the information held about generally complementary observables in $X$. It is thus natural to demand pure-state invariants to be compatible with this process of `information update'. We now formalise this idea in three different, but ultimately closely related ways.\\

\textbf{Inference closure.} We begin with an inference-closure condition.

\begin{definition}\label{def: inference-closure}
    A subset of observables $X\subset\cO\subset\LHsa$ is \emph{inference-closed relative to $\cO$} if for any two compatible observables $O,O'\in X$ also $f(O,O')\in X$ for any function $f$ with $f(O,O')\in\cO$.
\end{definition}

Def.~\ref{def: inference-closure} is a subset-relative closure condition under the spectral calculus, which we will apply to $\aP_n$-relative subsets. The idea is that functions of commuting, that is, simultaneously measurable observables may be understood in terms of post-processing on the side of the experimenter, which allows her to infer information about observables that are algebraically dependent on given ones.\footnote{Note that without demanding the subset-relativity in Def.~\ref{def: inference-closure}, inference-closed sets contain all relabellings of the spectra of the observables in $X$. Since, from the perspective of pure-state invariants and complementarity, such labels are (up to degeneracy) merely gauge, the restriction to $\cO$ may be interpreted to be a type of gauge-fixing.}\\

\textbf{Update-closure (under L\"uders measurement).} The dynamical aspect in Def.~\ref{def: inference-closure} remains still somewhat implicit. To bring it to the forefront, consider the following scenario. Assume we measure $P\in X$, and obtain the outcome $s=\pm 1$.\footnote{More generally, one can consider measurements of operators with eigenspaces corresponding to the joint eigenspaces of the Pauli observables of an abelian subgroup $A<\cP_n$, that is, whose spectral projectors are of the form $\Pi_{A,\chi}=\frac{1}{|A|}\sum_{P\in A}\chi(P)P$ for $\chi:A\ra\{\pm 1\}$ a character on $A$.} Then for any $Q\in\aP_n$ with $[Q,P]=0$, $Q$ and $PQ$ become dependent random variables with respect to the L\"uders post-measurement state of $P$,
\begin{align}\label{eq: Lueders update}
    \rho\mapsto\rho_{P,s}
    :=\frac{\Pi_{P,s}\rho\Pi_{P,s}}{\tr[\Pi_{P,s}\rho]}\; ,
\end{align}
where $\Pi_{P,s}=\frac{1}{2}(\one+sP)$, and $\tr[\Pi_{P,s}\rho]\neq 0$.\footnote{Of course, if the probability of obtaining the outcome $s$ in the state $\rho$ is zero, that is, $\tr[\Pi_{P,s}\rho]=0$, there is no associated post-measurement state.} Now, note that the information measure $I$ in Def.~\ref{def: I} of any two compatible observables $Q,P$ in $\rho_{P,s}$ obeys the relation,
\begin{align}\label{eq: I-Lüders-multipliciativity}
    I_{\rho_{P,s}}(PQ)
    =I_{\rho_{P,s}}(P)I_{\rho_{P,s}}(Q)\; .
\end{align}
Yet, this relation does not hold for a generic Boolean pure-state invariant, as it may distinguish between observables $Q$ and $PQ=s(P)Q$ whenever $Q,P\in X$, yet $QP\notin X$. In this case, $I_\psi(X)$ does not contain the information about observables, e.g. $QP$ that can be operationally inferred from information about observables it does contain, e.g. $P$ and $Q$. This suggests the following condition.

\begin{definition}\label{def: update-closure}
    A subset $\one\in X\subset\aP_n$ is called \emph{(L\"uders-)update-closed} if and only if for every $P\in X$ and $Q\in\aP_n$ with $[Q,P]=0$, and (non-zero L\"uders post-measurement) state $\rho_{P,s}$ of $P$ in Eq.~(\ref{eq: Lueders update}),
    \begin{align}\label{eq: update-closure}
        I_{\rho_{P,s}}(\{PQ\}\cap X)
        =I_{\rho_{P,s}}(\{P\}\cap X)I_{\rho_{P,s}}(\{Q\}\cap X)\; .
    \end{align}
\end{definition}

The dependence on post-measurement states in Def.~\ref{def: update-closure} is necessary, since the notion of dependent observables is. In particular, $I$ is not multiplicative on commuting Pauli operators for generic states. Note that Def.~\ref{def: update-closure} does not require $X$ to be a pure-state invariant. Moreover, Def.~\ref{def: update-closure} is readily extended to, and in fact equivalent to, the analogous closure condition under state-update with respect to observables associated with the Pauli group, that is, with respect to the projectors
$\Pi_{A,\chi}=\frac{1}{|A|}\sum_{P\in A}\chi(P)P$ for $\chi:A\ra\{\pm 1\}$ a character as above (cf. Ref.~\cite{FrembsHoehnNataleWever2026b}).\\

\textbf{Residue-closure.} Def.~\ref{def: inference-closure} and Def.~\ref{def: update-closure} are ultimately not specific to Pauli operators (up to phase). As we will see below, under this restriction, the conditions above simplify to the following condition:
\begin{align}\label{eq: symplectic inference}
    v,w\in X,\ v\neq w,\ \omega(v,w)=0
    \quad\Longrightarrow\quad
    v+w\in X\; .
\end{align}
In fact, it is immediate to see that Eq.~(\ref{eq: symplectic inference}) is equivalent to Def.~\ref{def: inference-closure} under the restriction to Pauli labels. Moreover, we can also define an explicit notion of (compatibility with) information update purely in terms of its underlying symplectic structure. To make this precise, we first define a notion of updating (the information contained in) subsets $X\subset V\cong\aP_n$ with respect to (learning information about) an isotropic subspace $U\subset V$.

\begin{definition}\label{def: update set}
    Let $X\subset V=\zz^{2n}_2$ and let $U\subset V$ be an isotropic subspace. Define its \emph{residue} by
    \begin{align}\label{eq: residue}
        \Res(U):=U^\perp/U\; ,
    \end{align}
    and the \emph{update of $X$ by $U$} with respect to the map $\pi_U:U^\perp\ra\Res(U)$ by
    \begin{align}\label{eq: updated set}
        X_U
        :=\pi_U((X\cap U^\perp)\backslash U)\; .
    \end{align}
\end{definition}

In words, the update of $X$ by $U$ is the intersection of $X$ with all elements that commute with $U$ modulo $U$ and $U$ excluded. We record the following straightforward result.

\begin{lemma}\label{lm: residue spaces}
    Let $(V=\zz^{2n}_2,\omega)$ be a symplectic space, and $U\subset V$ an isotropic subspace of dimension $k$. Then $(\Res(U),\omega|_{\Res(U)})$ is a symplectic space of dimension $2(n-k)$ with corresponding polar space $\Res^\circ(U)=\Res(U)\backslash\{U\}\cong W(2(n-k)-1,2)$.
\end{lemma}

\begin{proof}
    Consider the restriction of the symplectic form $\omega$ on $V=\zz^{2n}_2$ to the coset space $\Res(U)=U^\perp/U$ defined by $\omega|_{\Res(U)}(v+U,w+U):=\omega(v,w)$ for all $v,w\in U^\perp$. This is well-defined since
    \begin{align*}
        \omega(v+u,w+u')
        =\omega(v,w)+\omega(v,u')+\omega(u,w)+\omega(u,u')
        =\omega(v,w)\; ,
    \end{align*}
    as $v,w\in U^\perp$ and $U$ is isotropic. If the symplectic bracket of $v+U$ and every other coset $w+U\in\Res(U)$ vanishes, then $v\in (U^\perp)^\perp$. But since $\omega$ is nondegenerate, we have $(U^\perp)^\perp=U$, hence, $v+U$ is the zero coset. It follows that $\omega|_{\Res(U)}$ is nondegenerate, and $(\Res(U),\omega|_{\Res(U)})$ a symplectic vector space of $\dim(\Res(U))=2(n-k)$ with corresponding polar space $\Res(U)\backslash\{U\}\cong W(2(n-k)-1,2)$.
\end{proof}

Now, if $X$ is a pure-state invariant, equivalently an $m$-ovoid by Thm.~\ref{thm: all Boolean pure-state invariants}, then learning information about $U$ splits the information contained in $I_\psi(X)$ into the information about $U$ and the information in its residue. The latter should again be an ovoid whose size encodes the information of the residue, which is the information about $U$ subtracted from the total information about $X$. This suggests the following incidence-geometric compatibility condition with respect to the update in Def.~\ref{def: update set}.

\begin{definition}\label{def: residue-closure}
    Let $X\subset V^\circ$ be an $m$-ovoid. Then $X$ is called \emph{residue-closed} if and only if $X_U$ is an $\frac{m-(|S_U|-1)}{|S_U|}$-ovoid in $\Res^\circ(U)$ for every isotropic subspace $U\subset V$ with $S_U=(X\cap U)\cup\{0\}$.
\end{definition}

As shown in Lm.~\ref{lm: characterising Jordan} below, it is sufficient to impose residue closure for one-dimensional subspaces $U=\langle u\rangle$ with $u\in X^\circ$. More importantly, our next lemma relates the above three compatibility notions with the algebraic closure under products of commuting elements, which in the Pauli group becomes closure under the Jordan product or anti-commutator. As above, we identify $V\cong\aP_n\cong\zz^{2n}_2$.

\begin{lemma}\label{lm: characterising Jordan}
    Let $0\in X\subset V\cong\aP_n$ such that $X^\circ$ is an $m$-ovoid. Then the following are equivalent:
    \begin{itemize}
        \item[(i)] $X$ is inference-closed (as by Def.~\ref{def: inference-closure}),
        \item[(ii)] $X$ is update-closed (as by Def.~\ref{def: update-closure}),
        \item[(iii)] $X^\circ$ is residue-closed (as by Def.~\ref{def: residue-closure}),
        \item[(iv)] $X$ is Jordan-closed, that is closed under the (normalised) anti-commutator.
    \end{itemize}
\end{lemma}

\begin{proof}
    (i)$\Leftrightarrow$(iv): Let $v,w\in X$ with $\omega(v,w)=0$. If $X$ is Jordan-closed, then $v+w\in X$, and since this is the only (deterministic) function for which $0,v,w\neq f(v,w)\in V$, $X$ is inference-closed by Def.~\ref{def: inference-closure}. Conversely, if $X$ is inference-closed relative to the Hermitian Pauli representatives $\aP_n\cong V$, then $f(v,w)=v+w\in X$ whenever $v,w\in X$ and $\omega(v,w)=0$, hence, $X$ is Jordan-closed.

    (ii)$\Leftrightarrow$(iv): Let $Q,P\in X$ with $[Q,P]=0$, then with respect to a post-measurement state $\rho_{P,s}$ of $P$ that is a joint $Q,P$-eigenstate, we have $I_{\rho_{P,s}}(\{P\}\cap X)=1$ and $I_{\rho_{P,s}}(\{Q\}\cap X)=1$, hence, Eq.~(\ref{eq: update-closure}) requires that also $I_{\rho_{P,s}}(\{QP\}\cap X)=1$ which is only possible if $QP\in X$.
    
    Conversely, if $X$ is Jordan-closed, then for every $Q,P\in X$ with $[Q,P]=0$ also $QP\in X$. Moreover, under the restriction to the support of the post-measurement state $\rho_{P,s}$ with respect to $P$, we have $P=s\one$ for $s\in\{\pm 1\}$ and thus $QP=sQ$ from which $I_{\rho_{P,s}}(QP)=I_{\rho_{P,s}}(sQ)=I_{\rho_{P,s}}(Q)$, which implies Eq.~(\ref{eq: update-closure}) since $I_{\rho_{P,s}}(P)=1$. If instead $Q\notin X$, then by Jordan-closure also $PQ\notin X$ and Eq.~(\ref{eq: update-closure}) is again satisfied since $I_{\rho_{P,s}}(\{Q\}\cap X)=0=I_{\rho_{P,s}}(\{QP\}\cap X)$.

    (iii)$\Leftrightarrow$(iv): Let $U\leq V$ be isotropic and define $S_U:=\bigl(X\cap U\bigr)$. If $X$ is Jordan-closed, then $S_U\leq U$ is a linear subspace. Moreover, for every $x\in (X\cap U^\perp)\backslash U$ and $u\in U$,
    \begin{align*}
        x+u\in X
        \quad\Longleftrightarrow\quad
        u\in S_U\; .
    \end{align*}
    Indeed, the implication from right to left follows from Jordan closure, as $\omega(x,u)=0$. Conversely, if $x,x+u\in X$, then $\omega(x,x+u)=0$, and Jordan closure gives $u=x+(x+u)\in X$. Hence, every (non-zero) $U$-coset contained in $X_U$ contains exactly $|S_U|$ elements of $X$.
    
    Now let $L\supset U$ be a Lagrangian subspace. Since $X^\circ$ is an $m$-ovoid, $|X^\circ\cap L|=m$. Of these points, precisely $|S_U|-1$ lie in $U^\circ$, while every remaining point of $X_U\cap L/U$ is a coset containing $|S_U|$ points in $X\cap L^\circ$. Consequently,
    \begin{align*}
        |X_U\cap L/U|
        =\frac{m-|S_U|+1}{|S_U|}
        =:m_U\; .
    \end{align*}
    Since this expression is independent of $L$, $X_U$ is an $m_U$-ovoid in $\Res(U)=U^\perp/U$.

    Conversely, it is sufficient to impose residue closure for one-dimensional isotropic subspaces $U=\langle u\rangle$ with $u\in X^\circ$. Indeed, let $L'$ be a maximal isotropic subspace in $\Res(\langle u\rangle)\cong W(2(n-1)-1,2)$, and let $L\in\Lag(V)$ be its preimage under $\pi_U$ such that $U\subset L$ and $L'=L/U$. Since $X^\circ$ is an $m$-ovoid, we have $|X^\circ\cap L|=m$, and since, by assumption, $X_{\langle u\rangle}$ is an $\frac{m-1}{2}$-ovoid in $\Res(\langle u\rangle)$, also $|X_{\langle u\rangle}\cap L/\langle u\rangle|=\frac{m-1}{2}$. The latter intersection counts every pair $\{x,x+u\}$ for which either $x\in X$ or $x+u\in X$ (with $u\neq x\in X^\circ$ and $\omega(x,u)=0$), and it thus follows by cardinality that if $x\in X$, then also $x+u\in X$. Since every $x$ is contained in some Lagrangian subspace $L'\in\Res(\langle u\rangle)$ for some $u$, this implies that $X$ is Jordan-closed.
\end{proof}

We will call an $m$-ovoid $X\subset V^\circ$ \emph{update-closed} if $X\cup\{0\}$ is update-closed. Correspondingly, we call a Boolean pure-state invariant $I_\psi(X)$ \emph{update-closed} if its support $X$ is update-closed.

The constraint on an $m$-ovoid to be Jordan-closed is highly restrictive. In fact, we will show that other than the trivial cases, $m=0$ and $m=2^n-1$ only a single other case occurs. In the following, we will call a Boolean pure-state invariant $I_\psi(X)$ for $X\subset\aP_n$ \emph{proper} if $X\neq\emptyset,\{\one\},\aP^\circ_n,\aP_n$.

\begin{lemma}\label{lm: Jordan-closed ovoids}
    Let $X\subset V^\circ$ be a non-empty, proper, update-closed $m$-ovoid. Then $m=2^{n-1}-1$.
\end{lemma}

\begin{proof}
    By definition, $|X\cap L|=m$ for every Lagrangian $L\in\Lag(V)$, and since $X$ is Jordan-closed (by Lm.~\ref{lm: characterising Jordan}), non-empty and proper, we have $m=2^k-1$ for $0<k<n$. If $k=1$, then $X$ itself is a $1$-ovoid, hence, $n\leq 2$ \cite{Thas2001}, and since $k<n$, this gives $n=2$ and $k=n-1$. We may therefore assume $k\geq 2$. Let $U^\circ\subset X$ be a projective subspace of dimension $k-2$ (which exists since $X$ is Jordan-closed), then $\Res(U)=U^\perp/U$ is a symplectic space of dimension $2(n-(k-1))$ by Lm.~\ref{lm: residue spaces}, and the update $X_U$ of $X$ by $U$ in Def.~\ref{def: update set} intersects every maximal isotropic subspace in $\Res(U)\backslash U$ in a single coset. Indeed, let $L\supset U$ be a Lagrangian subspace, then $X\cap L$ corresponds to a $k$-dimensional isotropic subspace, hence, $X_U\cap L/U$ is a singleton. It follows that $X_U$ is a $1$-ovoid of $W(2(n-k+1)-1,2)$. However, $1$-ovoids exist only for $n-k+1\leq 2$ \cite{Thas2001}, which together with $k<n$ implies $k=n-1$.\footnote{In terms of Pauli operators, a $1$-ovoid would define a set of pairwise anti-commuting operators. Such sets have cardinality at most $2n+1$ (see e.g. Lm.~8 and Cor.~4 in Ref.~\cite{SarkarVanDenBerg2021}), hence, do not exist beyond the two-qubit case.}
\end{proof}

By the previous lemma, every update-closed $m$-ovoid $X$ has $m=2^{n-1}-1$. In particular, given a spread $\{L_i\}_{i=1}^{2^n+1}$, $X$ defines an $h$-system $\{\tL_i\}_{i=1}^{2^n+1}$ (see Def.~\ref{def: h-systems}). Equivalently, in terms of Pauli operators, viewed with respect to a partition into $2^n+1$ maximal Abelian subgroups, $X$ thus corresponds to a set $\{\tA^\circ_i\}_{i=1}^{2^n+1}$ of (complementary) Abelian subgroups $\tA_i$ of dimension $n-1$, where complementarity means that $C(\tA_i)\cap\tA_j=\{\one\}=C(\tA_j)\cap\tA_i$ for all $i\neq j$ (see Thm.~1 in Ref.~\cite{FrembsHoehnNataleWever2026b}). Next, we relate Jordan closure of $X$ to commutator-closure of its complement.

\begin{lemma}\label{lm: Jordan-closed ovoids correspond with Pauli-Lie algebras}
    Let $\one\in J\subset\aP_n$ be a proper, update-closed pure-state invariant, equivalently, by Thm.~\ref{thm: all Boolean pure-state invariants} and Lm.~\ref{lm: Jordan-closed ovoids}, $J^\circ$ is an update-closed $(2^{n-1}-1)$-ovoid. Then $L:=\overline{J}=\aP_n\backslash J$ is commutator-closed.
\end{lemma}

\begin{proof}
    Consider a collection of pairwise disjoint maximal Abelian subgroups $\{A_i\}_{i=1}^{2^n+1}$ of dimension $n$, equivalently (after removing the identity element) a complete spread of $PG(V,2)$ (with $V\cong\aP_n$). Since $J$ is Jordan-closed (by Lm.~\ref{lm: characterising Jordan}) and a pure-state invariant, we have $|J\cap A_i|=2^{n-1}$ for all $i$ by Lm.~\ref{lm: Jordan-closed ovoids}, that is, $J$ defines a set $\{\tA_i\}_{i=1}^{2^n+1}$ of Abelian subgroups $\tA_i<A_i$ of dimension $n-1$. Note that for any $1\leq i\neq j\leq 2^n+1$ there exist elements $g^j_i\in A_i$, $g^i_j\in A_j$ with $\langle\tA_i,g^j_i\rangle=A_i$, $\langle\tA_j,g^i_j\rangle=A_j$ (that is, $g^j_i\notin\tA_i$ and $g^i_j\notin\tA_j$) and such that $[g^j_i,\tA_j]=0=[g^i_j,\tA_i]$, and thus also $[g^j_i,g^i_j]\neq 0$. To see this, note first that, for every $1\leq i\leq 2^n+1$, there exists an element $g^j_i\in A_i$ such that $[g^j_i,\tA_j]=0$. Indeed, w.l.o.g. that is, up to conjugation by a Clifford unitary (cf. Lm.~1 in Ref.~\cite{FrembsHoehnNataleWever2026a}), let $A_i=\langle Z_k\rangle_{k=1}^n$, $\tA_i=\langle Z_k\rangle_{k=2}^n$ and $A_j=\langle X_k\rangle_{k=1}^n$, $\tA_j=\langle X_k\rangle_{k=2}^n$ then $g^j_i=Z_1$. But since $J$ is a proper Jordan-closed pure-state invariant, we must have $g^j_i\notin J$, hence, $g^j_i\notin\tA_i$ (similarly for $i\leftrightarrow j$). Moreover, $[g^j_i,g^i_j]\neq 0$ since otherwise $\langle g^j_i,g^i_j,\tA_j\rangle$ is an Abelian subgroup of dimension greater than $n$, which is impossible.

    Now, let $\oa,\oa'\in L=\overline{J}$, then from above we can write $\oa=g^j_i\ta_i$ and $\oa'=g^i_j\ta_j$ for $\ta_i\in\tA_i$, $\ta_j\in\tA_j$ and $1\leq i,j\leq 2^n+1$. If $i=j$, the commutator vanishes, while for $i\neq j$, the commutator reads
    \begin{equation}\label{eq: commutator from anti-commutator closure}
    \begin{aligned}
        [\oa,\oa']
        =\frac{1}{2}(g^j_i\ta_ig^i_j\ta_j-g^i_j\ta_jg^j_i\ta_i) 
        =\frac{1}{2}(g^j_ig^i_j\ta_i\ta_j-g^i_jg^j_i\ta_j\ta_i) 
        =\frac{1}{2}(g^j_ig^i_j\ta_i\ta_j+g^j_ig^i_j\ta_j\ta_i)
        =g^j_ig^i_j\{\ta_i,\ta_j\}\; ,
    \end{aligned}
    \end{equation}
    where we used $[g^j_i,\tA_j]=0=[g^i_j,\tA_i]$ and $[g^j_i,g^i_j]\neq 0$ from above. Next, observe that $g^j_ig^i_j\notin J$. Indeed, $[g^j_ig^i_j,\tA_i]=0$ and $g^j_ig^i_j\notin\tA_i$, hence, if $g^j_ig^i_j\in J$ then for the maximal Abelian subgroup $A=\langle\tA_i,g^j_ig^i_j\rangle$ we find $|J\cap A|=2^n$, in contradiction with $I_\psi(J)=2^{n-1}$ for all pure states. Finally, $[g^j_ig^i_j,\ta_i\ta_j]=0$ and since $J$ is Jordan-closed, $\{\ta_i,\ta_j\}\in J$. Together this implies that $L$ is closed under commutators: indeed, by Eq.~(\ref{eq: commutator from anti-commutator closure}), $[\oa,\oa']\sim hb$ with $h=g^j_ig^i_j\in L$, $b=\{\ta_i,\ta_j\}\in J$ and $[h,b]=0$. If $b=\{\tilde a_i,\tilde a_j\}=0$, then $[a,a']=0$ and there is nothing to prove. Otherwise, if $[\oa,\oa']\in J$, then Jordan-closure of $J$ would imply $h=(hb)b\in J$, contradicting that $h\in L$.
\end{proof}

Finally, we uniquely identify $L$ with a generating set of a copy of the symplectic (Pauli-)Lie algebra $\spf(2^{n-1})\subset\su(2^n)$ (see Def.~\ref{def: Pauli Lie algebra} in Sec.~\ref{sec: Cartan decomposition}).

\begin{theorem}\label{thm: all update-closed pure-state invariants}
    Let $n\geq 2$, and let $\one\in J\subset\aP_n$ such that $J^\circ$ is a proper, update-closed $(2^{n-1}-1)$-ovoid. Then there exists $\lf\cong\spf(2^{n-1})\subset\su(2^n)$ such that  $L:=\overline{J}\cong\lf|_\aP$. Conversely, every such copy defines a proper, update-closed $(2^{n-1}-1)$-ovoid.
\end{theorem}

\begin{proof}
    By Lm.~\ref{lm: Jordan-closed ovoids correspond with Pauli-Lie algebras}, $L=\overline{J}$ generates a (Pauli-)Lie algebra (see Def.~\ref{def: Pauli Lie algebra} in Sec.~\ref{sec: Cartan decomposition}). Now, $L$ contains $4^n-((2^n+1)(2^{n-1}-1)+1)=2^{n-1}(2^n+1)$ non-zero elements. We now invoke the full classification of Pauli Lie algebras in Thm.~2 of Ref.~\cite{AguilarEtAl2024}. Besides the simple families, this classification contains only their direct-sum and controlled extensions. We show that the only Pauli-Lie algebra in this classification that matches this dimension is $\spf(2^{n-1})\subset\su(2^n)$. (Recall that $\mathfrak{sp}(k)$ has real dimension $k(2k+1)$.) Indeed, for $n\geq 3$ the other simple Pauli-Lie algebras in Ref.~\cite{AguilarEtAl2024} have dimension $\dim(\su(2^n))=4^n-1$, $\dim(\sof(2^n))=2^{n-1}(2^n-1)$, and $\dim(\sof(m))\leq n(2n+1)$ for $m\leq 2n+1$, while direct sums and controlled cases are necessarily of smaller dimension: for direct sum decompositions, one has
    \begin{align*}
        \sum_j(4^{n_j}-1)
        \leq 4^{n-1}+2
        <2^{n-1}(2^n+1)\; ,
    \end{align*}
    while for controlled Pauli-Lie algebras with $c\geq 1$ and $m=n-c$,
    \begin{align*}
        2^c(4^m-1)
        \leq 2(4^{n-1}-1)
        =\frac{1}{2}4^n-2
        <\frac{1}{2}4^n+2^{n-1}
        =2^{n-1}(2^n+1)\; .
    \end{align*}
    Consequently, if the sets $L=\overline{J}\subset\aP_n$ with $J$ an update-closed pure-state invariant exist, then they arise from restrictions $L=\lf|_{\aP_n}$ of symplectic Pauli Lie algebras $\lf\cong\spf(2^{n-1})\subset\su(2^n)$. For $n=2$, the same argument does not rule out the free-fermionic Lie algebra $\sof(2n+1)$; indeed, $\sof(5)\cong\spf(2)$, hence, this case also fits the claim.
    
    To prove existence, we show that $J$ is indeed Jordan-closed and $J^\circ$ an $2^{n-1}-1$-ovoid. First, we show that $J=\oL$ is Jordan-closed, that is, closed under anti-commutators.\footnote{We mention in passing that this property also holds for the special orthogonal Pauli-Lie algebra $\soP(2^n)$.} To see this, note that
    \begin{align}\label{eq: symplectic Paulis}
    \begin{split}
        L&=\{S\in\aP_n\mid S^TQ+QS=0, \mathrm{\ where\ } Q^T=-Q\in\aP_n\} \\
        &=\{S\in\aP_n\mid QSQ=-S^T, \mathrm{\ where\ } Q^T=-Q\in\aP_n\}\; ,\\
        \Longleftrightarrow \quad J=\oL&=\{\oS\in\aP_n\mid Q\oS Q=\oS^T, \mathrm{\ where\ } Q^T=-Q\in\aP_n\}\; ,
    \end{split}
    \end{align}
    from which it is easy to prove the claim: let $\oS,\oS'\in\oL$, then
    \begin{align*}
        Q\{\oS,\oS'\}Q
        =\{Q\oS Q,Q\oS'Q\}
        =\{\oS^T,\oS'^T\}
        =\{\oS,\oS'\}^T\; ,
    \end{align*}
    hence, $\{\oS,\oS'\}\in J$. It follows that $J$ can be partitioned into Abelian subgroups, and since $|J^\circ|=(2^n+1)(2^{n-1}-1)$, this implies that $J^\circ$ is an $2^{n-1}-1$-ovoid if $J$ contains no Abelian subgroup of dimension $n$. To this end, note first that since $\aP_n$ is a basis for $\mathrm{Mat}_{2^n}(\C)$ there are $\frac{2^n(2^n-1)}{2}=2^{n-1}(2^n-1)$ anti-symmetric matrices $Q\in\aP_n$, and thus an equal number of (Clifford-equivalent) representations of $L$ and $J$ in Eq.~(\ref{eq: symplectic Paulis}). Now, assume that $A\subset J$ is an Abelian subgroup of dimension $n$. Since the group of symplectic transformations $\Sp(2n,\zz_2)$ acts transitively on Lagrangian subspaces, since every element $C\in \Sp(2n,\zz_2)$ is induced by a Clifford unitary via the surjection $\Cl_n\ra\Sp(2n,\zz_2)$ and since Clifford unitaries preserve Jordan closure, we may assume that $A=A_X=\langle X_1,\cdots,X_n\rangle$. Next, observe that the anti-symmetric Pauli operator $Q$ in Eq.~(\ref{eq: symplectic Paulis}) is mapped to an anti-symmetric Pauli under the transformation $Q\mapsto CQC^T$. Indeed, $CQC^T$ is anti-symmetric, $(CQC^T)^T=CQ^TC^T=-CQC^T$, and a Pauli operator since $CQC^T=(CQC^\dagger)(CC^T)$ and $CC^T=C\overline{C}^\dagger$ is in the kernel of $\Cl_n\ra\Sp(2n,\zz_2)$ since $C$ and $\overline{C}$ induce the same symplectic transformation. Consequently, $CQC^T$ corresponds with a Pauli operator containing an odd number of $Y$ tensor factors. However, for any such operator $Q$ there exists an element $P\in A_X$ which does not commute with $Q$, hence, we arrive at the contradiction $J\supset A_X\ni P\notin J$. Finally, since every Lagrangian is included in some complete spread, summing the contributions $|J\cap L^\circ|\leq 2^{n-1}-1$ of a complete spread $\{L^\circ_i\}_{i=1}^{2^n+1}$ and comparing with $|J^\circ|=(2^n+1)(2^{n-1}-1)$ implies that $J^\circ$ is an update-closed $(2^{n-1}-1)$-ovoid, equivalently (by Thm.~\ref{thm: all Boolean pure-state invariants} and Lm.~\ref{lm: Jordan-closed ovoids}) $J$ is a proper, update-closed Boolean pure-state invariant.
\end{proof}

Thm.~\ref{thm: all update-closed pure-state invariants} shows that symplectic Pauli-Lie algebras are Boolean pure-state invariants, and that every proper, update-closed invariant is a complement thereof. In particular, the six pure-state invariants for two qubits (`pentagon identities') in Ref.~\cite{HoehnWever2017} are of this type. More generally, by Thm.~2 in Ref.~\cite{FrembsHoehnNataleWever2026a}, there exist $2^{n-1}(2^n-1)$ distinct, but Clifford-conjugate subsets $L\subset\aP_n$ with $L\cong\lf|_{\aP_n}$. Indeed, the space of anti-symmetric matrices in $\mathrm{Mat}_{2^n}(\C)$ (of which $\aP_n$ is a generating set) has dimension $\frac{2^n(2^n-1)}{2}$, hence, there are that many choices for the anti-symmetric matrix $Q$ in Eq.~(\ref{eq: symplectic Paulis}). We stress that the additional assumption on pure-state invariants cannot be dropped from Thm.~\ref{thm: all update-closed pure-state invariants}: there exist many more pure-state invariants that are not update-closed (see Sec.~\ref{sec: Boolean pure-state invariants}).

In the following sections, we will focus on update-closed Boolean pure-state invariants. The attentive reader will have noticed that the proofs of Thm.~\ref{thm: all weighted quadratic invariants} (and Cor.~\ref{cor: Boolean quadratic invariants}), as well as that of Thm.~\ref{thm: all Boolean pure-state invariants} make heavy use of the quantities $\xi_r$. In the next section, we will study these in more detail.

\section{Structure of update-closed Boolean pure-state invariants}\label{sec: structure}

In this section, we highlight various perspectives on update-closed Boolean pure-state invariants - moving from concrete to abstract. We begin with a parametrisation in terms of the Weyl representation of a canonical pair of complete Lagrangian subspaces. In particular, we identify update-closed pure-state invariants with the family of invariants in Ref.~\cite{FrembsHoehnNataleWever2026a}. In Sec.~\ref{sec: Cartan decomposition}, we focus on the fact that they correspond with generating sets of Pauli Lie algebras, and their recent classification in Ref.~\cite{AguilarEtAl2024,GargiuloHerringerZeier2026}. In Sec.~\ref{sec: quadratic refinement}, we identify these with quadratic refinements of the symplectic structure with Arf invariant $1$ (see Eq.~(\ref{eq: Arf invariant}) below), equivalently of Witt index $n-1$. In Sec.~\ref{sec: quaternionic structure}, we distinguish these two cases from the perspective of the threefold way \cite{Dyson1962,Baez2012}, specifically of anti-unitary involutions $\Theta$ which equip the complex Hilbert space either with a real - for $\Theta^2=+1$ - or quaternionic - for $\Theta^2=-1$ - structure. We discuss in detail the difference between these two cases, relating pure-state invariants with the quaternionic case. Finally, based on this, in Sec.~\ref{sec: invariant set stabiliser} we analyse the stabiliser symmetries of pure-state invariants.

The mathematics in this section has been long known in the literature on symplectic polar spaces \cite{Arf1941,Dye1978,Taylor1992,BrownHumphries1986,Cuypers1994,Munemasa1996,Grove2002}, yet only very recently has it entered a wider audience within the quantum information community. In particular, some of the results we obtain overlap with two independent and concurrent works that appeared while this manuscript was in preparation: Ref.~\cite{Cuypers2026} studies Pauli Lie algebras via the quadratic space $\Pi_n/\{\pm 1\}$ and obtains the orthogonal/symplectic dichotomy (see Sec.~\ref{sec: quaternionic structure} below) in terms of the type of a geometric hyperplane; and Ref.~\cite{GargiuloHerringerZeier2026} classifies invariant bilinear forms for Pauli strings and relates them to binary quadratic forms. The results below were obtained independently, building on Ref.~\cite{FrembsHoehnNataleWever2026a,FrembsHoehnNataleWever2026b}; we state them in full for self-containedness and indicate the correspondence with the formulations of Ref.~\cite{Cuypers2026,GargiuloHerringerZeier2026} as we go.\\

\textbf{Weyl form and symplectic Pauli Lie algebras.} We begin by providing an explicit representation for the pure-state invariants in Thm.~\ref{thm: all update-closed pure-state invariants}. Clearly, $\zeta_0$ corresponds with the identity component, which yields a trivial invariant for all states by normalisation. Moreover, for pure states $I_\psi(\zeta_1)=I_\psi(V)=2^n$. We show that up to these ambiguities the family of pure-state invariants in Ref.~\cite{FrembsHoehnNataleWever2026a} is complete. For convenience, we reprint Def.~1 from Ref.~\cite{FrembsHoehnNataleWever2026a} here.

\begin{definition}\label{def: family of pure-state invariants}
    For $n\geq 1$, let $A_X,A_Z\leq \aP_n$ be two maximal Abelian and complementary subgroups, equivalently, two non-overlapping Lagrangian subspaces, that is, $A_X\cap A_Z=\{\one\}$, and thus $A_X\oplus A_Z\cong \aP_n$. For any subgroups $\tA_X\subset A_X$, $\tA_Z\subset A_Z$ of dimension $n-1$ and elements $g^Z_X\in A_X\backslash\tA_X$, $g^X_Z\in A_Z\backslash\tA_Z$ such that $[g^Z_X,\tA_Z]=0=[g^X_Z,\tA_X]$ (and thus $[g^Z_X,g^X_Z]\neq 0$), define\footnote{As shown in Lm.~1 in Ref.~\cite{FrembsHoehnNataleWever2026a}, given an Abelian subgroup $\tA_Z\subset A_Z$ and an element $g^X_Z\in A_Z\backslash\tA_Z$, there exists a unique subgroup $\tA_X\leq A_X$ and element $g^Z_X\in A_X\backslash\tA_X$ satisfying these properties.}
    \begin{align}\label{eq: family of pure-state invariants}
    \begin{split}
        J(\tA_Z,g^X_Z)
        &:=\{\{\ta_X,\ta_Z\}\mid \ta_X\in\tA_X,\ta_Z\in\tA_Z\}\ \cupdot\ 
        \{\{g^Z_X\ta_X,g^X_Z\ta_Z\}\mid \ta_X\in\tA_X,\ta_Z\in\tA_Z\}\\
        &\hspace{.5cm}\cupdot \{[\ta_X,g^X_Z\ta_Z]\mid \ta_X\in\tA_X,\ta_Z\in\tA_Z
        \}\ \cupdot\
        \{[g^Z_X\ta_X,\ta_Z]\mid \ta_X\in\tA_X,\ta_Z\in\tA_Z\}\; ,
    \end{split}
    \end{align}
    where, according to our convention in Sec.~\ref{sec: Pauli observables}, vanishing brackets are omitted.
\end{definition}

\begin{corollary}\label{cor: completeness}
     Proper, update-closed pure-state invariants correspond with the family in Def.~\ref{def: family of pure-state invariants}, equivalently with a single Boolean weight $\zeta^+_a=\frac{1}{2}(1+\xi_a)$ for $a\in\cT_-$ in Thm.~\ref{thm: all weighted quadratic invariants} and Cor.~\ref{cor: Boolean quadratic invariants}.
\end{corollary}

\begin{proof}
    Without loss of generality (see Lm.~1 in Ref.~\cite{FrembsHoehnNataleWever2026a}), we may take $A_Z=\langle Z_1,\cdots,Z_n\rangle$, $A_X=\langle X_1,\cdots,X_n\rangle$, as well as $\tA_Z=\langle Z_1,\cdots,Z_{n-1}\rangle$ and $g^X_Z=Z_n$. Eq.~(\ref{eq: family of pure-state invariants}) can thus be re-written in terms of the character $\delta_X:A_X\times A_Z\ra\zz_2$ given by $\delta_X(a_Xa_Z)=0$ if $a_X\in\tA_X$ (similarly for $\delta_Z$),
    \begin{equation}\label{eq: J from P1}
    \begin{aligned}
        J(\tA_Z,g^X_Z)
        &=\{v=(v_X,v_Z)\in V\mid\omega(v_X,v_Z)=\delta_X(v_X)+\delta_Z(v_Z)\mod 2\}\\
        &=\{v=(v_X,v_Z)\in V\mid\tau(v)=(-1)^{\omega(a,v_X)+\omega(a,v_Z)}\}\\
        &=\{v\in V\mid\xi_a(v)=\chi_a(v)\tau(v)=1\}
        =:J_a\; ,
    \end{aligned}
    \end{equation}
    where we defined $a:=(v^X_Z,v^Z_X)$, used that $(-1)^{\omega(v_X,v_Z)}=\tau(v)$ for $v=(v_X,v_Z)$ (by which $a\in\cT_-$), and $\delta_X(v_X)=\omega(a,v_X)\mod 2$ as well as $\delta_Z(v_Z)=\omega(a,v_Z)\mod 2$. It follows that every element in Def.~\ref{def: family of pure-state invariants} corresponds (up to the trivial invariant $\zeta_0$) with the Boolean pure-state invariants corresponding to a single Boolean weight $\zeta^+_a$ in Cor.~\ref{cor: Boolean quadratic invariants} (see also Thm.~\ref{thm: all weighted quadratic invariants}). The complementary condition $\xi_a(v)=-1$ is precisely the defining relation of the symplectic Pauli Lie algebra in Eq.~(\ref{eq: symplectic Paulis}), hence, Def.~\ref{def: family of pure-state invariants} also corresponds with the update-closed Boolean pure-state invariants by Thm.~\ref{thm: all update-closed pure-state invariants}.
\end{proof}

As in Eq.~(\ref{eq: family of pure-state invariants}) and Eq.~(\ref{eq: J from P1}), we will write $L$ and $J$ both for subsets of the symplectic vector space $(V,\omega)$, as well as for subsets of their phase-free Pauli representatives in $\aP_n$.

Def.~\ref{def: family of pure-state invariants} thus provides an explicit parametrisation of update-closed pure-state invariants as well as of maximal, Jordan-closed, complementary sets of abelian subgroups \cite{FrembsHoehnNataleWever2026b}, and identifies them with those in Cor.~\ref{cor: Boolean quadratic invariants}. Building on the analysis in Ref.~\cite{FrembsHoehnNataleWever2026a}, in the following sections we further extract the structural properties of these pure-state invariants in more detail.

\subsection{Cartan decompositions and Pauli Lie algebras}\label{sec: Cartan decomposition}

Lm.~2 in Ref.~\cite{FrembsHoehnNataleWever2026a} shows that the sets in Def.~\ref{def: family of pure-state invariants} satisfy strong algebraic closure relations. To state these explicitly, it is best to interpret (phase-free) Pauli operators as generators of a Lie subalgebra of the special unitary Lie algebra. We recall the following more general definition from Ref.~\cite{AguilarEtAl2024}.

\begin{definition}\label{def: Pauli Lie algebra}
    A \emph{Pauli Lie algebra $\fg\subset\su(2^n)$} is any Lie subalgebra (closed under commutators), generated from $n$-qubit Pauli strings under commutators and real-linear combinations. Equivalently, $\fg\subset\su(2^n)$ is a Pauli Lie algebra if its restriction $\fg|_{\pm i\aP_n}$ is closed under the commutator (on $\cP_n$).

    A \emph{Pauli Jordan algebra $\cJ\subset\LHsa$} is any Jordan subalgebra (closed under anti-commutators), generated from $n$-qubit Pauli strings under anti-commutators and real-linear combinations.

    A generating set $L\subset\pm i\aP_n$ of a Pauli Lie, respectively $J\subset\pm\aP_n$ of a Pauli Jordan algebra is called \emph{minimal} if its proper subsets generate smaller Lie, respectively Jordan algebras.
\end{definition}

The structure of Def.~\ref{def: Pauli Lie algebra} naturally differentiates $L$ from $J$ when considering phases explicitly, as $L\subset\pm i\aP_n$ and $J\subset\pm\aP_n$. In the next section, we will find that the decompositions can indeed be encoded by choosing phases in this manner. Nevertheless, we will often identify $L$ and $J$ with their phase-free Pauli representatives in $\aP_n$, upon which $J=\oL=\aP_n\backslash L$ (as in Thm.~\ref{thm: all update-closed pure-state invariants}).

The real generating sets $L$ in Def.~\ref{def: family of pure-state invariants} are closed under commutators and thus define a Pauli Lie algebra. In fact, they satisfy even stronger relations, as captured by the following definition.

\begin{definition}\label{def: PLJ algebra}
    A subset $L\subset\pm i\aP_n$ is called a generating set of a Pauli Lie Jordan algebra if
    \begin{equation}\label{eq: algebraic closure relations}
    \begin{aligned}
        \{L,L\}&=:J &\hspace{2cm} \{L,J\}&=L &\hspace{2cm} \{J,J\}&=J \\
        [L,L]&=L &\hspace{2cm} [L,J]&=J &\hspace{2cm}  [J,J]&=L\; .
    \end{aligned}
    \end{equation}
\end{definition}

Here, we use the convention from Sec.~\ref{sec: Pauli observables} that set-valued commutators and anti-commutators are understood projectively: vanishing brackets are discarded, while every nonzero bracket is identified with its unique Hermitian Pauli representative up to central phase. Our terminology in Def.~\ref{def: PLJ algebra} follows that in Def.~\ref{def: Pauli Lie algebra}. Decompositions of this type, namely Cartan decompositions of $\su(2^n)$ with Pauli-adapted generating sets — have since been studied systematically in Ref.~\cite{Cuypers2026,GargiuloHerringerZeier2026}; what concerns us here is which of them arise as pure-state invariants. To put Def.~\ref{def: PLJ algebra} in a broader context, recall that a \emph{Cartan decomposition} of a Lie algebra $\fg=\lf\oplus\pf$ is defined by the relations
\begin{align}\label{eq: Cartan decomposition}
    [\lf,\lf]&\subseteq\lf &
    [\lf,\pf]&\subseteq\pf &
    [\pf,\pf]&\subseteq\lf\; .
\end{align}
Cartan decompositions are the infinitesimal shadow of symmetric spaces \cite{Helgason2001}. In particular, differentiating the compact symmetric spaces $\SU(d)/\SO(d)$ and $\SU(d)/\Sp(d/2)$, for $n$-qubit systems with $d=2^n$ yields the Cartan decompositions $\su(2^n)=\lf\oplus\pf$ with either $\lf\cong\sof(2^n)$ or $\lf\cong\spf(2^{n-1})$. In geometric terms, $\lf=T_eK$ and $\pf\cong T_{eK}(G/K)$, where the second identification is induced by the differential $G\ra G/K$, and $(G,K)=(\SU(2^n),\SO(2^n))$ and $(G,K)=(\SU(2^n),\Sp(2^{n-1}))$, respectively. In Sec.~\ref{sec: quaternionic structure} we will see that these symmetric spaces obtain a natural interpretation as real structures on the complex Hilbert space $\C^{2^n}$.

Given the restrictive conditions in Def.~\ref{def: PLJ algebra}, one may wonder whether these are characteristic of update-closed Boolean pure-state invariants, yet this is not the case. The classification of Ref.~\cite{AguilarEtAl2024} shows that, up to Clifford equivalence, Pauli Lie algebras are built from the free-fermionic family $\sof(m)$ with $m\leq 2n+1$, and the exponentially large families $\sof(2^n)$, $\spf(2^{n-1})$, and $\su(2^n)$, together with direct-sum and controlled extensions. For the dimension argument below, it is sufficient to compare these families. Notably, the only case that is not exponential in dimension is the first, which corresponds with the free-fermion Pauli Lie algebra (for details, see Ref.~\cite{AguilarEtAl2024,GargiuloHerringerZeier2026}). The last case is simply the unconstrained one, which from the perspective of pure-state invariants corresponds with the scalar constraint $\tr[p^2_\psi]=1$. It is not difficult to see that neither of them are Pauli Lie Jordan algebras (see Def.~\ref{def: PLJ algebra}). In contrast, as stated above, both the orthogonal and symplectic Pauli Lie algebras satisfy the closure relations in Def.~\ref{def: PLJ algebra},\footnote{This is easily checked, e.g. from their generating sets in Eq.~(\ref{eq: so and sp}) below.} yet (by Thm.~\ref{thm: all update-closed pure-state invariants}) only the symplectic case gives rise to pure-state invariants. We thus seek an additional structural property satisfied by pure-state invariants that distinguishes between the two. In Lm.~5 of Ref.~\cite{FrembsHoehnNataleWever2026a}, we identify such a property in the (constant) size of commutants $C_{L/J}(Q):=\{P\in L/J\mid [P,Q]=0\}$ within the decompositions $\aP_n\cong V=L_a\cupdot J_a$, where $J_a$ as in Eq.~(\ref{eq: J from P1}), $L_a=V\backslash J_a$ and $W_a^T=-W_a\in\aP_n$ is any anti-symmetric Pauli operator which takes on the role of a symplectic form, e.g. $iW_a=iY_1=\begin{pmatrix}
    0 & \one \\ -\one & 0
\end{pmatrix}$. Before determining these cardinalities in Lm.~\ref{lm: constant commutants} below, we first identify the reason for this with an additional invariant of quadratic refinements of the symplectic form in characteristic $2$.

\subsection{Quadratic refinements, Arf invariant and Witt index}\label{sec: quadratic refinement}

In this and the following sections, we will develop some of the material on quadratic refinements. Quadratic refinements of a symplectic form in characteristic $2$ are standard in the mathematical literature; we recall the relevant facts to keep the exposition self-contained, referring to Ref.~\cite{Arf1941,Dye1978,BrownHumphries1986,Taylor1992,Cuypers1994,Munemasa1996,Grove2002} for details. In the Pauli setting they have recently been brought to bear on the classification of Pauli Lie algebras \cite{Cuypers2026,GargiuloHerringerZeier2026}; our concern here is different, namely to characterise which of the associated Pauli subsets correspond with pure-state invariants. Our main result of this section, Thm.~\ref{thm: pure-state invariants as quadratic refinements}, proves that update-closed pure-state invariants are precisely the zero sets of quadratic refinements with Arf invariant $1$. Given its relative novelty, and to keep the exposition self-contained, we provide basic proofs of the results relevant for our purposes, and refer the interested reader to above references for further details.

Let $V$ be a vector space over a field $K$, and let $q:V\ra K$ be a homogeneous quadratic function, that is, $q(kv)=k^2q(v)$ for all $k\in K$ and $v\in V$. Its polar form $b_q:V\times V\ra K$ given by $b_q(u,v)=q(u+v)-q(u)-q(v)$ defines a symmetric bilinear form, and is uniquely determined by $q$. Conversely, given a symmetric bilinear form $b:V\times V\ra K$, evaluation on the diagonal yields a homogeneous quadratic by setting $q(v)=\frac{1}{2}b(v,v)$ whenever $K$ has characteristic $\ch(K)\neq 2$. Hence, $q$ is also uniquely determined by $b$ and there is a bijective correspondence between these two structures.\footnote{More generally, this correspondence is the vector space isomorphism between homogeneous polynomials and symmetric forms, and extends beyond the quadratic case (over fields of characteristic zero).}

The situation changes if $\ch(K)=2$, as in this case $b_q(v,v)=q(2v)-2q(v)=0$ for all $v\in V$. This implies that $b_q$ is alternating, however, otherwise forces no constraints on the diagonal values of $b_q$. As a consequence, the relation between homogeneous quadratic functions and alternating bilinear forms $b$ is no longer one to one. In order to characterise the ambiguity in the correspondence $b\leftrightarrow q$, call any quadratic function $q:V\ra K$ with $q(kv)=k^2q(v)$ and polarisation $b_q=b$ a \emph{quadratic refinement of $b$}. A given $b$ generally admits different quadratic refinements. A simple example is given by the zero form $b=0$ on $V=\zz_2$ which has quadratic refinements $q,q':\zz_2\ra\zz_2$ given by $q(0)=0$, $q(1)=0$ and $q'(0)=0$, $q'(1)=1$. Moreover, the relation $(-1)^{q_a(v)}:=\xi_a(v)=\chi_a(v)\tau(v)$ for $v\in V$ defines a quadratic refinement of the symplectic form $\omega$ on $V$ for all $a\in\cT_-$, since
\begin{align}\label{eq: quadratic refinement from pure-state invariant}
    (-1)^{q_a(u+v)}
    =\chi_a(u+v)\tau(u+v)
    =\chi_a(u)\chi_a(v)\tau(u)\tau(v)(-1)^{\omega(u,v)}
    =(-1)^{q_a(u)+q_a(v)+\omega(u,v)}\; ,
\end{align}
where we used multiplicativity of $\chi_a$ and $\tau(u+v)=(-1)^{\omega(u,v)}\tau(u)\tau(v)$.

Which types of quadratic refinements correspond to the $q_a$ in Eq.~(\ref{eq: quadratic refinement from pure-state invariant})? To answer this question, note first that in the present setting $K=\zz_2$, $V=\zz^{2n}_2$ and $b:V\times V\ra K$ is non-degenerate. Since a polar form is necessarily alternating, $b$ thus coincides with the symplectic form $\omega$ on $V$. Moreover, quadratic refinements are characterised by their \emph{Arf invariant}: given a symplectic basis $\{e_i,f_j\}_{i,j=1}^n$, that is, such that $\omega(e_i,e_j)=0=\omega(f_k,f_l)$ and $\omega(e_i,f_j)=\delta_{ij}$ for all $i,j,k,l\in[n]$, it is defined by
\begin{align}\label{eq: Arf invariant}
    \Arf(q)
    :=\sum_{i=1}^n q(e_i)q(f_i)\mod 2\; .
\end{align}
For $K=\zz_2$, it is not too difficult to see that this definition is independent of the choice of symplectic basis (see e.g. Lm.~4 in Ref.~\cite{FrembsHoehnNataleWever2026a}), and thus defines an invariant under $\Sp(2n,\zz_2)$.\footnote{More generally, $\Arf(q;e,f)$ is independent of the symplectic basis in the quotient $K/p(K)$ for $K$ any field of characteristic $2$, and where $p(K)=\{k^2+k\mid k\in K\}$. Clearly, for $K=\zz_2$, $p(K)=\{0\}$.}

There is an alternative characterisation. Let $V=V_0(q)\cupdot V_1(q)$ with $V_s=\{v\in V\mid q(v)=s\}$ be the partition into level sets under $q$, that is, into singular and nonsingular vectors. An alternative characterisation is then given by the \emph{Witt index $i_W(q)$ of $q$}, defined by
\begin{align}\label{eq: Witt index}
    i_W(q)
    :=\max\{\dim(U)\mid U\subset V_0\ \mathrm{isotropic}\}\; .
\end{align}
In words, $i_W(q)$ specifies the dimension of the largest isotropic subspace in the level set $V_0=\{v\in V\mid q(v)=0\}$ of $q$. It follows from Lm.~\ref{lm: Jordan-closed ovoids}, Thm.~\ref{thm: all update-closed pure-state invariants} and Cor.~\ref{cor: completeness} that the quadratic refinements $(-1)^{q_a}=\xi_a$ have Witt index $i_W(q_a)=n-1$. By contrast, the Witt index of quadratic refinements with Arf invariant $0$ is $n$ \cite{Munemasa1996,Taylor1992}. We have the following canonical representation.

\begin{lemma}\label{lm: quadratic refinement standard form}
    Let $(V=\zz^{2n}_2,\omega)$ be a symplectic vector space. Up to a change of symplectic basis, any quadratic refinement $q$ of $\omega$ is of one of the following canonical forms
    \begin{equation}\label{eq: quadratic refinement standard form}
    \begin{aligned}
        q_+(v)&=\sum_{i=1}^n\alpha_i\beta_i
        & &\longrightarrow & \Arf(q_+)&=0 &
        &\Longleftrightarrow & i_W(q_+)&=n \\
        q_-(v)&=\sum_{i=1}^n\alpha_i\beta_i+\alpha_1+\beta_1
        & &\longrightarrow & \Arf(q_-)&=1 &
        &\Longleftrightarrow & i_W(q_-)&=n-1
    \end{aligned}\; .
    \end{equation}
    Moreover, $\Sp(2n,\zz_2)$ acts transitively on quadratic refinements of the same Arf invariant.
\end{lemma}

\begin{proof}
    For $n\geq 2$, pick $0\neq e\in V$ with $q(e)=0$ (if $e_1,e_2$ are nonsingular, then $q(e_1+e_2)=1+1+0=0$) and let $f\in V$ with $\omega(e,f)=1$. If $q(f)=1$, replace $f$ by $f'=f+e$, for which $\omega(e,f')=\omega(e,f+e)=\omega(e,f)$ and $q(f')=q(f+e)=q(f)+q(e)+\omega(e,f)=q(f)+1=0$. Consequently, on $W=\langle e,f\rangle$, $q$ has the canonical form $q(\alpha e+\beta f)=\alpha\beta$. Decompose $V=W\oplus W^\perp$ (with respect to the bilinear form $\omega$) and note that by Eq.~(\ref{eq: Arf invariant}), $\Arf(q)=\Arf(q|_{W^\perp})$. Repeating the same construction $n-1$ times (and applying the same basis change argument for a symplectic pair $e,f$ with $\omega(e,f)=1$ as above), we are left with two cases for a symplectic basis on the remaining two-dimensional subspace:
    \begin{align}\label{eq: base case}
        &q(e)=q(f)=0 & 
        &q(e)=q(f)=1\; .
    \end{align}
    Clearly, with respect to this basis $q$ is of either one of the forms in Eq.~(\ref{eq: quadratic refinement standard form}), with $\Arf(q)=0$ in the first and $\Arf(q)=1$ in the second case. Moreover, this argument also establishes Eq.~(\ref{eq: quadratic refinement standard form}) for $n=1$.

    Finally, let $q,q'$ be two quadratic refinements of $\omega$ with the same Arf invariant. By the first part, there exist symplectic bases $\{e_i,f_j\}_{i,j=1}^n$ and $\{e'_i,f'_j\}_{i,j=1}^n$ on which both $q$ and $q'$ take either one of the standard forms in Eq.~(\ref{eq: quadratic refinement standard form}). Let $S\in\mathrm{GL}(V)$ be defined by $S(e_i)=e'_i$ and $S(f_i)=f'_i$ for all $i\in[n]$. Then $S\in\Sp(2n,\zz_2)$ and $q'(v)=(S\cdot q)(v):=q(S^{-1}(v))$ for all $v\in V$.
\end{proof}

What type of information does the Arf invariant of a quadratic refinement of $\omega$, equivalently its Witt index encode? For one thing, it distinguishes quadratic refinements by their number of zeros, (see e.g. Lm.~5 in Ref.~\cite{FrembsHoehnNataleWever2026a}).

\begin{lemma}
    Let $(V=\zz^{2n}_2,\omega)$ be a symplectic vector space and $q:V\ra\zz_2$ a quadratic refinement of $\omega$ with Arf invariant $\Arf(q)$. Then
    \begin{align}\label{eq: zero count}
        N(q):=
        |\{v\in V\mid q(v)=0\}|
        =\frac{1}{2}4^n+(-1)^{\Arf(q)}2^{n-1}\; .
    \end{align}
\end{lemma}

\begin{proof}
    This follows easily by induction on $n$. For $n=1$ one has $N(q)=3$ or $1$ according whether $\Arf(q)=0$ or $1$. For $n>1$ write $V=V_1\perp V_2$ with $\dim V_1=2$. Since $q(v_1+v_2)=q(v_1)+q(v_2)$ for $v_i\in V_i$ and $\Arf(q)=\Arf(q|_{V_1})+\Arf(q|_{V_2})$, counting the two ways of writing $0=q(v_1)+q(v_2)$ gives $N(q)=N_1N_2+(4-N_1)(2^{2n-2}-N_2)$, which one checks yields the stated formula.
\end{proof}

Comparing with the cardinality of the family of pure-state invariants in Def.~\ref{def: family of pure-state invariants}, which correspond with the zero-level sets of the quadratic form $(-1)^{q_a}=\xi_a$ by Cor.~\ref{cor: completeness}, we find $\Arf(q_a)=1$ (for $a\in\cT_-$). In other words, pure-state invariants define quadratic refinements of Arf invariant $1$. Conversely, we will show that every quadratic refinement $q$ of the symplectic form $\omega$ on $V=\zz^{2n}_2$ with $\Arf(q)=1$ gives rise to an (update-closed) pure-state invariant.

\begin{lemma}\label{lm: Arf reps}
    Let $(V=\zz^{2n}_2,\omega)$ be a symplectic vector space, and let $q,q'$ be two quadratic refinements of $\omega$. Then there exists a unique $v'\in V$ such that $q'(v)=q(v)+\omega(v,v')$ for all $v\in V$, and
    \begin{align}\label{eq: shifted Arf invariant}
        \Arf(q')=\Arf(q)+q(v')\; .
    \end{align}
\end{lemma}

\begin{proof}
    If two quadratic functions $q,q':V\ra K$ have the same polar form, then their difference $l=q-q'$ has vanishing polarisation, which implies that $l(v+w)=l(v)+l(w)$ is additive. In the case $K=\zz_2$ that we are interested in, $l:V\ra\zz_2$ is a linear function. Consequently, the space of quadratic refinements of $\omega$ is the affine space $\{q+l\mid l:V\ra\zz_2\ \mathrm{linear}\}$ for $q$ any quadratic refinement of $\omega$. Since $\omega$ is non-degenerate, any linear function can be written as $v\ra\omega(v,v')$ for some $v'\in V$.

    To prove Eq.~(\ref{eq: shifted Arf invariant}), pick a symplectic basis $\{e_i,f_j\}_{i,j=1}^n$ and let $v'=\sum_{i=1}^n\alpha_ie_i+\beta_if_i$. By the first part, $q'(e_i)=q(e_i)+\omega(v',e_i)=q(e_i)+\beta_i$ and $q'(f_j)=q(f_j)+\omega(v',f_j)=q(f_j)+\alpha_j$, hence,
    \begin{align*}
        \Arf(q')
        =\sum_{i=1}^nq'(e_i)q'(f_i)
        =\Arf(q)+\sum_{i=1}^n(\alpha_iq(e_i)+\beta_iq(f_i)+\alpha_i\beta_i)
        =\Arf(q)+q(v')\; ,
    \end{align*}
    where we used that $\Arf$ is independent of the choice of symplectic basis, as well as the identity $q(u+v)=q(u)+q(v)+\omega(u,v)$ with the relations $\omega(e_i,e_i)=\omega(f_j,f_j)=0$ and $\omega(e_i,f_j)=\delta_{ij}$.
\end{proof}

Note that $q_+(v)=a\cdot b$ in Lm.~\ref{lm: quadratic refinement standard form} is the (standard) quadratic form in the Weyl representation of Eq.~(\ref{eq: Weyl operators}), that is, $q_+=q_W$. By Lm.~\ref{lm: Arf reps}, any other quadratic refinement is of the form $q_{v'}:=q_W+\omega(\cdot,v')$ for $v'\in V$, in particular, $q_a(v)=q_W(v)+\omega(v,a)$, equivalently $(-1)^{q_a}=\chi_a\tau=\xi_a$ for all $a\in\cT_-$.

As a consequence, we obtain another characterisation of update-closed pure-state invariants.

\begin{theorem}\label{thm: pure-state invariants as quadratic refinements}
    $I_\psi(X)$ for $\{0\}\subsetneq X\subsetneq V$ is an update-closed pure-state invariant if and only if $X=\{v\in V\mid q(v)=0\}$ for $q$ a quadratic refinement of the symplectic form $\omega$ on $V\cong\aP_n$ with $\Arf(q)=1$.
\end{theorem}

\begin{proof}
    One direction follows from Thm.~\ref{thm: all update-closed pure-state invariants}, Cor.~\ref{cor: completeness} and the preceding paragraphs; the converse direction follows from Lm.~\ref{lm: Arf reps}, in particular, $\Arf(q_a)=\Arf(q_W)+q_W(a)=1$ for all $a\in\cT_-$.
\end{proof}

Equivalently, $X\subset\aP_n$ defines an update-closed pure-state invariant if and only if $X=\{v\in V\mid q(v)=0\}$ for $q$ a quadratic refinement of Witt index $n-1$. In Ref.~\cite{FrembsHoehnNataleWever2026b}, we show that the latter characterisation  is directly related to the complementarity structure in the $n$-qubit Pauli group, which in turn gives rise to strong entropic uncertainty relations (see also Sec.~\ref{sec: uncertainty relations}).

In particular, quadratic refinements with non-zero Arf invariant are parametrised by vectors $r\in V$ with $q_W(r)=1$. Clearly, this picks out all Pauli operators with an odd number of Pauli-Y components, equivalently, all anti-symmetric Pauli labels $r=a\in\cT_-$. In the next section, we relate the Arf invariant with the two types of Cartan decompositions of Pauli Lie Jordan algebras, which correspond with the (exponentially-sized) orthogonal and symplectic Pauli Lie algebras.

\subsection{Orthogonal vs symplectic Lie algebra}\label{sec: quaternionic structure}

In Sec.~\ref{sec: quadratic refinement}, we have seen that the pure-state invariants in Thm.~\ref{thm: all update-closed pure-state invariants} correspond with quadratic refinements of the symplectic form $\omega$ with $\Arf$ invariant $1$. In this section, we take a more abstract view on this characterisation, by identifying the $\Arf$ invariant with the choice of a real as opposed to a quaternionic structure on the complex Hilbert space $\cH=\C^{2^n}$. Such a structure arises from an anti-unitary operator $\Theta:\cH\ra\cH$ with either $\Theta^2=\one$ or $\Theta^2=-\one$. Recall that an anti-unitary operator $\Theta:\cH\ra\cH$ is anti-linear, that is, $\Theta(\alpha\psi)=\overline{\alpha}\Theta\psi$ for all $\alpha\in\C$ and $\psi\in\cH$, and decomposes as
\begin{align}\label{eq: anti-unitary decomposition}
    \Theta=\Theta_R=RK\; ,
\end{align}
where $K$ denotes complex conjugation and $R\in\UH$ is unitary. Now, given such an anti-unitary structure $\Theta$, the unitary subgroups preserving it are respectively
\begin{align}\label{eq: J-preserving unitaries}
\begin{split}
    \OG_\Theta(2^n)&=\{U\in\UH\mid \Theta U\Theta^{-1}=U,\ \Theta^2=\one\} \\
    \Sp_\Theta(2^{n-1})
    &=\{U\in\UH\mid \Theta U\Theta^{-1}=U,\ \Theta^2=-\one\}\; .
\end{split}
\end{align}
In the orthogonal case, we can find a unitary such that $\Theta=K$ (namely $R=\one$), in the compact symplectic case $\Sp(n)=\Sp(2n,\C)\cap \U(2n)$ a canonical choice for $\Theta=RK$ is given by the standard symplectic form, $R=\begin{pmatrix}
    \mathbbm{0} & \mathbbm{1} \\
    -\mathbbm{1} & \mathbbm{0}
\end{pmatrix}=iY_1$. The respective Lie algebras are thus isomorphic to
\begin{align}\label{eq: so and sp}
\begin{split}
    \sof_\Theta(2^n)&\cong\{u\in\su(2^n)\mid uR+Ru^T=0,\ R^T=R\} \\
    \spf_\Theta(2^{n-1})
    &\cong\{u\in\su(2^n)\mid uR+Ru^T=0,\ R^T=-R\}
\end{split}\; .
\end{align}
By Thm.~\ref{thm: all update-closed pure-state invariants}, pure-state invariants arise from copies of the symplectic Pauli Lie algebra $\spf(2^{n-1})$ with Pauli generating sets. More generally, for Pauli operators $R=W_r\in\aP_n$ for $r\in V$,
\begin{align}\label{eq: symplectic generating sets}
    L_r
    =\{P\in\aP_n\mid PR+RP^T=0\}
    =\{P\in\aP_n\mid PW_r+W_rP^T=0\}\; ,
\end{align}
is a generating set of a (copy of a) symplectic Pauli Lie algebra for $R^T=-R\in\aP_n$ anti-symmetric, and a generating set of a (copy of an) orthogonal Pauli Lie algebra for $R^T=R\in\aP_n$ symmetric. Building on Thm.~\ref{thm: pure-state invariants as quadratic refinements} in Sec.~\ref{sec: quadratic refinement}, we find that the respective generating sets correspond with quadratic refinements, and are distinguished from one another by their Arf invariant.

\begin{proposition}\label{prop: Arf-PLJ correspondence}
    Let $(V=\zz^{2n}_2,\omega)$ be a symplectic vector space, let $q_r$ for $r\in V$ be a quadratic refinement of $\omega$, and let $L_r$ be defined as in Eq.~(\ref{eq: symplectic generating sets}). Then (with Eq.~(\ref{eq: anti-unitary decomposition}))
    \begin{equation}\label{eq: Arf distinguishes real structures}
    \begin{aligned}
        \langle L_r\rangle&\cong\sof(2^n) & &\Longleftrightarrow & \Theta^2_{W_r}&=\one & &\Longleftrightarrow & \Arf(q_r)&=0 \\
        \langle L_r\rangle&\cong\spf(2^{n-1}) & &\Longleftrightarrow & \Theta^2_{W_r}&=-\one & &\Longleftrightarrow & \Arf(q_r)&=1\; ,
    \end{aligned}
    \end{equation}
    where $r\in\cT_+$ in the first and $r\in\cT_-$ in the second line. Moreover, the correspondence between generating sets of Pauli Lie algebras and quadratic refinements is bijective in both cases.
\end{proposition}

\begin{proof}
    The first correspondences follow from the paragraphs preceding the theorem. For the second, let $R=W_r\in\aP_n$ for $r\in V$, then
    \begin{align}\label{eq: real Pauli-structures}
        \Theta^2_{W_r}
        =(W_rK)^2
        =W_rKW_rK
        =W_r\overline{W_r}
        =(-1)^{q_W(r)}\one
        =(-1)^{\Arf(q_r)}\one\; ,
    \end{align}
    where $q_W(v)=v_X\cdot v_Z$ is the quadratic refinement in the Weyl representation in Eq.~(\ref{eq: Weyl operators}) with $\Arf(q_W)=0$, and we used that $\Arf(q_r)=\Arf(q_W)+q_W(r)=q_W(r)$ by Lm.~\ref{lm: Arf reps}. Hence, $\Theta_{W_r}$ defines a real structure for symmetric Pauli operators, equivalently for $r\in\cT_+$, and a quaternionic structure for anti-symmetric ones, equivalently for $r\in\cT_-$. Finally, bijectivity follows from
    \begin{align*}
        L_r
        =\{W_v\in\aP_n\mid W_vW_r+W_rW_v^T=0\}
        =\{W_v\in\aP_n\mid W_rW^T_vW_r=-W_v\}
        =\{W_v\in\aP_n\mid q_r(v)=1\}\; ,
    \end{align*}
    where $(-1)^{q_r}:=\xi_r=\chi_r\tau$ is the quadratic form associated with these sets. By Thm.~\ref{thm: all update-closed pure-state invariants} and Thm.~\ref{thm: pure-state invariants as quadratic refinements}, $r\in\cT_-$ corresponds with the symplectic, and $r\in\cT_+$ with the orthogonal case.
\end{proof}

In particular, combining Eq.~(\ref{eq: J from P1}) and Eq.~(\ref{eq: symplectic generating sets}) (with the notation in Sec.~\ref{sec: Cartan decomposition}), and using again the identification $\aP_n\cong V$ as sets, we thus set
\begin{equation}\label{eq: J and L}
\begin{aligned}
    J_r
    &:=\{v\in V\mid \xi_r(v)=\chi_r(v)\tau(v)=1\}
    =\{v\in V\mid q_r=0\}\\
    L_r
    &:=\{v\in V\mid \xi_r(v)=\chi_r(v)\tau(v)=-1\}
    =\{v\in V\mid q_r=1\}
\end{aligned}\; .
\end{equation}

The Arf invariant, equivalently, the Witt index of a quadratic refinement of the symplectic form $\omega$ on $V=\zz^{2n}_2$ thus distinguish between the two Pauli Lie Jordan algebras in Def.~\ref{def: PLJ algebra}.

The correspondence between Pauli generating sets and quadratic forms also appears, in different formulations, in Ref.~\cite{Cuypers2026,GargiuloHerringerZeier2026}. We end this section by pointing out that the above analysis sheds more light on the classification of Pauli Lie algebras. First, as a consequence of Thm.~\ref{thm: all update-closed pure-state invariants} and Cor.~\ref{cor: completeness}, Def.~\ref{def: family of pure-state invariants} characterises all embeddings of the symplectic Pauli Lie algebras $\spf(2^{n-1})\subset\su(2^n)$, and provides an explicit representation of the symplectic structure of the Pauli sets (see App.~\ref{app: symplectic representation}). Second, our characterisation of update-closed pure-state invariants as zero sets of quadratic refinements with Arf invariant $1$ adds another physical layer to the classification of bilinear forms in Ref.~\cite{AguilarEtAl2024}.

In the last two sections, we have predominantly focussed on the symplectic Lie algebra structure underlying (with generating sets $L_a$ for $a\in\cT_-$ of) update-closed Boolean pure-state invariants. Yet, since they are indeed Pauli Lie Jordan algebras (see Def.~\ref{def: PLJ algebra}), this covers only one side of the story. In App.~\ref{app: cubic purity relations}, we consider an alternative characterisation of purity in terms of cubic relations, which turns out to be more naturally related with the closure relations under the anti-commutator. More generally, it would be interesting to mirror the characterisation of Pauli Lie algebras in Ref.~\cite{AguilarEtAl2024} by one for Pauli Jordan algebras. We will not do this here, but point out that Thm.~\ref{thm: all update-closed pure-state invariants} singles out those that are also pure-state invariants, equivalently, that define $m$-ovoids (see also Sec.~\ref{sec: positive Pauli projection maps}).

\subsection{Stabiliser groups of pure-state invariants}\label{sec: invariant set stabiliser}

In preparation for Sec.~\ref{sec: applications}, we study the symmetries of the sets $J_r$ for $r\in V$, specifically those for $r\in\cT_-$, which correspond with pure-state invariants by Thm.~\ref{thm: all update-closed pure-state invariants}. More precisely, we compute the projective stabiliser group $\PStab(J_r)$ within the projective unitary group $\mathrm{PU}(2^n)$,
\begin{align}\label{eq: stabiliser group}
    \PStab(J_r)
    =\{U\in\mathrm{PU}(2^n)\mid UJ_rU^\dagger\sim J_r\}\; ,
\end{align}
as well as the stabiliser group of its complement in the Pauli group $L_r:=\overline{J_r}=\aP_n\backslash J_r$,
\begin{align}\label{eq: complement stabiliser group}
    \PStab(L_r)
    =\{U\in\mathrm{PU}(2^n)\mid UL_rU^\dagger\sim L_r\}\; .
\end{align}
Note that these are different from the stabiliser of $R=W_r$, which corresponds with the orthogonal group for $R^T=R\in\aP_n$ or the symplectic group $R^T=-R\in\aP_n$ (cf. Eq.~(\ref{eq: so and sp})). Still, in Cor.~\ref{cor: Clifford intersections}, we will see that these subgroups are closely related. Already by Prop.~\ref{prop: Arf-PLJ correspondence}, we expect $\PStab(J_r)$ and $\PStab(L_r)$ to be closely related to the group of linear isometries of quadratic refinements $q$,
\begin{align}\label{eq: isomety group}
    \OG(q)
    :=\{T\in\mathrm{GL}(V)\mid q(Tv)=q(v)\ \forall v\in V\}\; .
\end{align}
By Lm.~\ref{lm: quadratic refinement standard form}, $\Sp(2n,\zz_2)$ acts transitively on quadratic refinements of either Arf invariant via $S\cdot q(v):=q(S^{-1}(v))$ for $S\in\Sp(2n,\zz_2)$ a symplectic transformation and $v\in V$. It follows that the isometry groups of different quadratic refinements of the symplectic form $\omega$ on $V=\zz^{2n}_2$ with the same Arf invariant are isomorphic to the same abstract group denoted by $\OG^\pm(2n,\zz_2)$ for $q_\pm$ with $\Arf(q_+)=0$ and $\Arf(q_-)=1$. In particular, the isometry groups $\OG(q_a)$ of the quadratic refinements $(-1)^{q_a}=\xi_a$ for $a\in\cT_-$ (corresponding to pure-state invariants by Thm.~\ref{thm: pure-state invariants as quadratic refinements}) are isomorphic to $\OG^-(2n,\zz_2)$.

To relate these with the stabiliser subgroups in Eq.~(\ref{eq: stabiliser group}), recall that the Clifford group $\Cl_n$ induces an action on Pauli labels by symplectic transformations. More precisely, let $\pCl_n:=\Cl_n/\U(1)$ be the projective Clifford group with corresponding projection $\pi:\pCl_n\ra\Sp(2n,\zz_2)$, and write $C_S\in\pi^{-1}(S)$ for any chosen representative of the pre-image of $S\in\Sp(2n,\zz_2)$ under $\pi$. Now, in Ref.~\cite{FrembsHoehnNataleWever2026a} it was shown that the stabiliser group of the family of all pure-state invariants in Def.~\ref{def: family of pure-state invariants} (which parametrise all update-closed pure-state invariants by Thm.~\ref{thm: all update-closed pure-state invariants} and Cor.~\ref{cor: completeness}), is the full Clifford group. Nevertheless, $\PStab(J_r)$ and $\PStab(L_r)$ may a priori contain both Clifford and non-Clifford unitaries. The next lemma proves that this is not the case, and identifies the relevant subgroups in the Clifford group with the corresponding isometry group in Eq.~(\ref{eq: isomety group}) up to phases.

\begin{lemma}\label{lm: transitivity}
    Let $r\in V=\zz^{2n}_2$, and $J_r$, $L_r$ as in Eq.~(\ref{eq: J and L}). Then for all $n\geq 2$,
    \begin{align*}
        \PStab(J_r)
        &=\PStab(L_r)
        =\pi^{-1}(\OG(q_r)) 
        \subset\pCl_n\; ,
    \end{align*}
    whereas for $n=1$, $\PStab(J_r)=\pi^{-1}(\OG(q_r))$ and $\PStab(L_r)\cong \OG(2)$ for $r\in\cT_+$ as well as $\PStab(J_r)=\mathrm{PU}(2)$ and $\PStab(L_r)=\pi^{-1}(\OG(q_r))=\pCl_1$ for $r\in\cT_-$.
\end{lemma}

\begin{proof}
    By Eq.~(\ref{eq: J and L}), the Pauli subsets $J_r$ and $L_r$ correspond with the level sets of $q_r$,
    \begin{align}\label{eq: gen sets as level sets}
        J_r&=\{v\in V\mid q_r(v)=0\} &
        L_r
        &=\{v\in V\mid q_r(v)=1\}\; .
    \end{align}
    It is thus sufficient to prove that these sets span the Pauli group, as $UJ_rU^\dagger\sim J_r$ implies
    \begin{align*}
        U\cP_n U^\dagger
        =U\langle J_r\rangle U^\dagger
        =\langle UJ_rU^\dagger\rangle
        =\langle J_r\rangle
        =\cP_n\; ,
    \end{align*}
    hence, $U=C\in\Cl_n$ (and similarly for $L_r$). Since Clifford unitaries $C\in\Cl_n$ induce symplectic transformations, $CJ_rC^\dagger=J_r$ with Eq.~(\ref{eq: J and L}) is then equivalent to
    \begin{align*}
        q_r(S_C(v))=1
        \quad\Longleftrightarrow\quad
        q_r(v)=1\quad\quad\quad\forall v\in V\; ,
    \end{align*}
    that is, $C\in\PStab(J_r)$ if and only if $S_C\in \OG(q_r)$ is a linear isometry of $q_r$ (and similarly for $L_r$).
    
    To show that the respective sets span $\aP_n$, equivalently, that the corresponding level sets in Eq.~(\ref{eq: gen sets as level sets}) span $V$, we consider the two canonical cases $q_\pm$ in turn. By transitivity of $\Sp(2n,\zz_2)$ on quadratic refinements with fixed Arf invariant in Lm.~\ref{lm: quadratic refinement standard form}, this is without loss of generality.

    First, note that $q^{-1}_+(0)$ contains the standard symplectic basis $\{e_i,f_j\}_{i,j=1}^n$ for all $n\geq 1$ since $q_+(e_i)=q_+(f_i)=0$ for all $i\in[n]$, hence, $\langle q^{-1}_+(0)\rangle=V$. For $q^{-1}_+(1)$ and $n\geq 2$, consider the elements $h_i=e_i+f_i$ and note that $q_+(h_i)=1=q_+(h_i+e_j)$ for every $i\neq j$. Now, since $e_j=h_i+(h_i+e_j)$ and $f_j=h_j+e_j$, it follows again that $q^{-1}_+(1)$ spans $V$. For $n=1$, the argument fails and one instead finds $q^{-1}_+(1)=(1,1)$, corresponding to Pauli-$Y$ (up to phase), hence, $\PStab(L_r)=\OG(2)$ in this case.

    Second, for $q^{-1}_-(0)$ there exists a symplectic basis with $q_-(e_1)=q_-(f_1)=1$ and $q_-(e_i)=q_-(f_i)=0$ for all $i\geq 2$ by Lm.~\ref{lm: quadratic refinement standard form}. For $n\geq 2$, $q^{-1}_-(0)$ thus contains $e_i,f_i$ for $i\geq 2$ as well as the vectors
    \begin{align*}
        u&=e_1+e_2+f_2 & v&=f_1+e_2+f_2 & w&=e_1+e_2+f_1+f_2\; .
    \end{align*}
    Since $e_1=v+w$ and $f_1=u+w$, together these vectors again span $V$, and since $q_-(u)=q_-(v)=q_-(w)=0$ so does $q^{-1}_-(0)$. For $n=1$, the argument fails and one instead finds $q^{-1}_-(0)=0$ (where $q_-=q_a$ for $a=(1,1)$ with $W_a=Y$), hence, $\OG^-(2,\zz_2)=\mathrm{GL}(2,\zz_2)=\mathrm{SL}(2,\zz_2)=\Sp(2,\zz_2)$ (the first equality holds since $q_-(v)=1$ for every $0\neq v\in V$, the second since $\zz^\times_2=1$, and the third since $A^TJA=\det(A)J$ for $2\times 2$-matrices and $J=iY$ the standard symplectic matrix), and thus $\PStab(J_a)=\mathrm{PU}(2)$. For $q^{-1}_-(1)$, note that $e_1,f_1$ together with $e_i+f_i$, $e_1+f_i$ and $e_i+f_1$ for $i\geq 2$ lie in the $1$-level set of $q_-$, hence, so do $e_i=f_1+(e_i+f_1)$ and $f_i=e_1+(e_1+f_i)$ from which $\langle q^{-1}_-(1)\rangle =V$.
\end{proof}

The relation between Pauli Lie algebras and Clifford subgroups generated by transvections has been developed systematically and more generally in Ref.~\cite{GargiuloHerringerZeier2026}; Lm.~\ref{lm: transitivity} is the special case relevant from the perspective of pure-state invariants.

For completeness, we record the cardinalities of the isometry groups (see e.g. Ref.~\cite{Taylor1992}),
\begin{align*}
    |\OG^\pm(2n,\zz_2)|
    =2^{n(n-1)+1}(2^n\mp 1)\prod_{i=1}^{n-1}(2^{2i}-1)\; .
\end{align*}
For $n\geq 2$, the cardinalities of the projective stabiliser groups are a multiple thereof, with multiplication factor given by $4^n=|\pP_n|$, since $\pP_n$ is the kernel of the projective map $\pi:\pCl_n\ra\Sp(2n,\zz_2)$.

As an immediate consequence of Lm.~\ref{lm: transitivity}, we obtain a description of these stabiliser subgroups as intersections with the stabiliser groups of the real structures in Eq.~(\ref{eq: so and sp}).

\begin{corollary}\label{cor: Clifford intersections}
    For every $n\geq 2$, $s\in\cT_+$ and $a\in\cT_-$, and $\mathrm{P}:\U(2^n)\ra\mathrm{PU}(2^n)$ it holds
    \begin{equation*}
    \begin{aligned}
        \PStab(J_s)
        =\PStab(L_s)
        &=\mathrm{P}(\OG_{W_sK}(2^n)\cap\Cl_n) \\
        \PStab(J_a)
        =\PStab(L_a)
        &=\mathrm{P}(\Sp_{W_aK}(2^{n-1})\cap\Cl_n)
    \end{aligned}\; .
    \end{equation*}
\end{corollary}

\begin{proof}
    The first equality is Lm.~\ref{lm: transitivity}. For the second equality, let first $URU^T=R$ for $R=W_r$ with $r=s,a$. For $P\in L_r$, we have $PR+RP^T=0$ and thus 
    \begin{align*}
        (UPU^\dagger)R+R(UPU^\dagger)^T
        =(UPU^\dagger)(URU^T)+(URU^T)(UPU^\dagger)^T
        =U(PR+RP^T)U^T
        =0\; .
    \end{align*}
    Moreover, since $U$ is Clifford, $UPU^\dagger$ is again a Pauli operator (up to phase), which proves that $\mathrm{P}(\OG_{W_sK}(2^n)\cap\Cl_n)\subset\PStab(L_s)$ and $\mathrm{P}(\Sp_{W_aK}(2^{n-1})\cap\Cl_n)\subset\PStab(L_a)$.
    
    Conversely, let $U\in\PStab(L_r)$. By the proof of Lm.~\ref{lm: transitivity}, $L_r$ generates the Pauli Lie algebra $\fg_R=\{x\in\mathfrak{su}(2^n)\mid xR+Rx^T=0\}$ for $R=W_r$ (that is, $\fg_{W_s}=\sof_{W_sK}(2^n)$ for $W^T_s=W_s$ and $\fg_{W_a}=\spf_{W_aK}(2^{n-1})$ for $W^T_a=-W_a$), hence, $U\fg_R U^\dagger=\fg_R$. Next, observe that $R':=U^\dagger R\overline{U}$ is another $\fg_R$-invariant nondegenerate bilinear form on the irreducible defining vector space $\C^{2^n}$. Indeed, from $UxU^\dagger\in\fg_R$ for every $x\in\fg_R$, multiplication of $(UxU^\dagger)R+R(UxU^\dagger)^T=0$ from the left by $U^\dagger$ and from the right by $\overline{U}$ yields $xR'+R'x^T=0$. Comparing this with $xR+Rx^T=0$ gives $xR'R^{-1}=R'R^{-1}x$ for all $x\in\fg_R$. Since the defining representation is irreducible (and since $R,R'$ are unitary), Schur's lemma implies $R'=U^\dagger R\overline{U}=\lambda R$ with $\lambda\in U(1)$. Consequently, after re-phasing $U'=\sqrt{\lambda}U$, this gives $U'RU'^T=R$, and thus projectively, and together with $U\in\pCl_n$ by Lm.~\ref{lm: transitivity}, that $U\in\mathrm{P}(\OG_{RK}(2^n)\cap\Cl_n)$, respectively $U\in\mathrm{P}(\Sp_{RK}(2^{n-1})\cap\Cl_n)$.
\end{proof}

In particular, the isometry group of the quadratic refinement $q_+=q_W=a\cdot b$ in the Weyl representation (see Eq.~(\ref{eq: Weyl operators})) corresponds with the real Clifford group $\Cl^\R_n:=\OG(2^n)\cap\Cl_n=\{C\in\Cl_n\mid C_{ij}\in\R\}$, via $\OG^+(2n,\zz_2)\cong\pi(\Cl^\R_n)$. A generating set is given by Paulis, Hadamard and CNOT gates (phase gates are evidently not real). In particular, $q_+=q_W=q_0$ with $\OG(2^n)=\OG_\one(2^n)$ (see Eq.~(\ref{eq: so and sp})).

Note also that by Lm.~\ref{lm: transitivity}, the stabiliser groups $\PStab(q)$ and $\PStab(q')$ for quadratic refinements $q,q'$ with $\Arf(q')=\Arf(q)$ are related by conjugation with a Clifford unitary $C_S\in\Cl_n$, via
\begin{align*}
    \pi^{-1}(\OG(q'))
    =C_S\pi^{-1}(\OG(q))C^{-1}_S\; ,
\end{align*}
whose symplectic projection $S=\pi(C)\in\Sp(2n,\zz_2)$ satisfies $q'=S\cdot q$. In particular, $\OG(q_r)$ for $r\in\cT_+$ is Clifford-conjugate to $\OG(q_W)=\pi(\Cl^\R_n)$, yet, unlike $\Cl^\R_n$, is generally not real.

For an explicit Clifford representation of $\PStab(L_r)$ with $r\in\cT_-$, define the \emph{transvections}
\begin{align}\label{eq: transvection}
    t_u(v):=v+\omega(u,v)u
\end{align}
for every $u,v\in V$. Then for $q_r(u)=1$ (and only in this case) one finds $t_u\in\OG(q_r)$ since,
\begin{align}\label{eq: action by transvection}
    q_r(t_u(v))
    =q_r(v+\omega(u,v)u)
    =q_r(v)+\omega(u,v)(q_r(u)+1)
    =q_r(v)\; ,
\end{align}
where we used the defining property of a quadratic refinement in the second step. In terms of (projective) Clifford unitaries, $t_u$ lifts to conjugation by $C_{t_u}:=e^{-\frac{i\pi}{4}W_u}$ (up to phase). Indeed,
\begin{align}\label{eq: Clifford transvection}
    T_u(W_v):=\Ad_{C_{t_u}}(W_v)
    =C_{t_u}W_vC^\dagger_{t_u}
    \sim\begin{cases}
        W_v&\mathrm{if}\ \omega(u,v)=0\\
        W_uW_v\sim W_{u+v}&\mathrm{if}\ \omega(u,v)=1
    \end{cases}\; ,
\end{align}
and thus $T_u(W_v)\sim W_{t_u(v)}$. Conversely, since $\OG^-(2n,\zz_2)$ is generated by transvections (see, e.g. Cor.~11.42 in Ref.~\cite{Taylor1992}),\footnote{Over $\zz_2$, the transvections $t_u(v)=v+\omega(u,v)u$ with $q(u)=1$ are the orthogonal reflections in the terminology of Ref.~\cite{Taylor1992}. By Cor.~11.42 therein, these transformations generate $O(q)$ (except for $O^+(4,2)$, where they generate an index-$2$ subgroup instead). This is the characteristic-two analogue of the Cartan–Dieudonn\'e theorem for fields $\mathrm{char}(K)\neq 2$, which asserts that any orthogonal transformation on a finite-dimensional vector space $V$ with non-degenerate symmetric bilinear form over $K$ is generated by at most $\mathrm{dim}(V)$ reflections. In fact, the result extends to all fields $K\neq\zz_2$ (see Thm.~11.39 of Ref.~\cite{Taylor1992}).\label{fn: Cartan-Dieudonne}}
it follows that all Clifford unitaries corresponding to $\OG^-(2n,\zz_2)$ are $\frac{\pi}{2}$-rotations along the Pauli directions in $L_r$, that is, for $r\in\cT_-$,
\begin{align*}
    \pi^{-1}(\OG(q_r))
    =\langle\mathrm{P}(C_{t_u})
    \mid u\in L_r\rangle\subset\pCl_n\; .
\end{align*}
In fact, transvections generate the isometry groups $\OG(q)$ for any quadratic refinement (with the exception of $q_+$ for $n=2$).$^{\ref{fn: Cartan-Dieudonne}}$ Using this, we obtain an alternative proof for Lm.~5 in Ref.~\cite{FrembsHoehnNataleWever2026a}. To this end, recall that the commutant of $Q\in\aP_n$ in $X\subset\aP_n$ is defined by $C_X(Q)=\{P\in X\mid [P,Q]=0\}$.

\begin{lemma}\label{lm: constant commutants}
    Let $(V=\zz^{2n}_2,\omega)$ be a symplectic space, and $q_r$, $r\in V$ be a quadratic refinement of $\omega$ with $0$- and $1$-level sets denoted by $J_r$ and $L_r$, respectively (see Eq.~(\ref{eq: J and L})). Then the cardinalities of the commutants $C_X(Q)$ for $X=J_r,L_r$ depend only on $q_r(Q)$, and are given by,
    \begin{equation}\label{eq: constant commutants}
    \begin{aligned}
        |C_{J_r}(Q\in L_r)|
        &=4^{n-1} &\hspace{1cm}
        |C_{J_r}(Q\in J^\circ_r)|
        &=4^{n-1}+(-1)^{\Arf(q_r)}2^{n-1} \\
        |C_{L_r}(Q\in L_r)|
        &=4^{n-1} &\hspace{1cm}
        |C_{L_r}(Q\in J^\circ_r)|
        &=4^{n-1}-(-1)^{\Arf(q_r)}2^{n-1}\; .
    \end{aligned}
    \end{equation}
\end{lemma}

\begin{proof}
    For every $u\in L_r=\{v\in V\mid q_r(v)=1\}$, define the bijective map $T_u:\aP_n\ra\aP_n$ by
    \begin{align}\label{eq: transvections}
        T_u(W_v)
        =\Ad_{C_{t_u}}(W_v)
        \sim\begin{cases}
            W_v&\mathrm{if}\ \omega(v,u)=0\\
            W_uW_v&\mathrm{if}\ \omega(v,u)=1
        \end{cases}\; .
    \end{align}
    $T_u$ corresponds with the Clifford representation of the transvections in Eq.~(\ref{eq: transvection}), that is, with (a conjugation by) the Clifford unitary $C_{t_u}=e^{-\frac{i\pi}{4}W_u}$ which induces a $\frac{\pi}{2}$-rotation around the Pauli-$W_u$ axis in Eq.~(\ref{eq: Clifford transvection}). From the relations in Eq.~(\ref{eq: algebraic closure relations}), $T_u(J_r)=J_r$ and $T_u(L_r)=L_r$. Consequently,
    \begin{equation}\label{eq: transvected commutants}
    \begin{aligned}
        T_u(C_{J_r}(Q))
        &=\{T_u(P)\in J_r\mid[P,Q]=0\}\\
        &=\{P'\in J_r\mid[T_u(P'),Q]=0\}\\
        &=\{P'\in J_r\mid[P',T_u(Q)]=0\}
        =C_{J_r}(T_u(Q))\; ,
    \end{aligned}
    \end{equation}
    where we used $T^2_u=\mathrm{id}$ in the second, and $[T_u(P),T_u(Q)]=[P,Q]$ for all $P,Q\in\aP_n$ in the third step, both of which are readily verified from Eq.~(\ref{eq: transvections}). The case $T_u(C_{L_r}(Q))=C_{L_r}(T_u(Q))$ is analogous. Now, by Cor.~11.42 in Ref.~\cite{Taylor1992},$^{\ref{fn: Cartan-Dieudonne}}$ the isometry groups $\OG(q)<\Sp(2n,\zz_2)$ are generated by the transvections $t_u$ for $u\in L_r$, that is, $\OG(q_r)=\langle t_u\mid q_r(u)=1\rangle$ (for all but $\OG^+(4,\zz_2)$, that is, $n=2$ and $\Arf(q)=0$, in which case one verifies correctness of Eq.~(\ref{eq: constant commutants}) by direct computation). Moreover, $\OG(q_r)$ acts transitively on the level sets of $q_r$ in $V^\circ$ (see e.g. Thm.~7.4 in Ref.~\cite{Taylor1992}). Together with Eq.~(\ref{eq: transvected commutants}) this implies that $|C_{J_r}(Q)|$ and $|C_{L_r}(Q)|$ depend only on $q_r(w)$ for $Q=W_w$.
    
    To compute these cardinalities, let $Q=W_w$. Since $[W_v,Q]=0$ if and only if $v\in w^\perp$, we have
    \begin{align*}
        C_{J_r}(Q)&=J_r\cap w^\perp
        &
        C_{L_r}(Q)&=L_r\cap w^\perp\; .
    \end{align*}
    Now, if $q_r(w)=0$, then for $v\in w^\perp$ we have $q_r(v+w)=q_r(v)+q_r(w)+\omega(v,w)=q_r(v)$, that is, $q_r$ is constant on the cosets of $\langle w\rangle$ and descends to a quadratic refinement $\overline{q}_r$ of the induced nondegenerate symplectic form on $w^\perp/\langle w\rangle\cong\zz_2^{2(n-1)}$ with $\Arf(\overline{q}_r)=\Arf(q_r)$. By Eq.~(\ref{eq: zero count}), $|\overline{q}_r^{-1}(0)|=\frac{1}{2}4^{n-1}+(-1)^{\Arf(q_r)}2^{n-2}$, and every such coset lifts to two elements in $w^\perp$, which yields the two equations on the right of Eq.~(\ref{eq: constant commutants}). If instead $q_r(w)=1$, then $q_r(v+w)=q_r(v)+q_r(w)+\omega(v,w)=q_r(v)+1$, hence, $v\mapsto v+w$ is a fixed-point-free involution of $w^\perp$ interchanging $J_r\cap w^\perp$ and $L_r\cap w^\perp$. Clearly, this implies that exactly half of the $|w^\perp|=|C_{J_r}(Q)|+|C_{L_r}(Q)|=2^{2n-1}$ elements are singular.
\end{proof}

Note that the cardinalities in each column of Eq.~(\ref{eq: constant commutants}) are not independent but sum to $\frac{1}{2}4^n$, expressing that every non-identity Pauli operator (anti-)commutes with half the elements in $\aP_n$.

Notably, Lm.~\ref{lm: constant commutants} applies to both the orthogonal and symplectic Pauli Lie algebras. Comparing with the derivation of pure-state invariants in Ref.~\cite{FrembsHoehnNataleWever2026a} this suggests both cases should lead to pure-state invariants. However, this is not the case: while the relation in Thm.~1 of Ref.~\cite{FrembsHoehnNataleWever2026a},
\begin{align}\label{eq: P1 relation}
    I_\psi(J_r)
    &=\frac{1-2^{n-1}}{2^{n-1}}|J_r|+\frac{1}{2^{n-1}}\sum_{\one\neq P\in\aP_n} |C_{J_r}(P)|\alpha^2_P \nonumber \\
    \Longleftrightarrow\quad \left(1-\frac{|C_{J_r}(Q\in J^\circ_r)|-|C_{J_r}(Q\in L_r)|}{2^{n-1}}\right)I_\psi(J_r)
    &=\frac{1-2^{n-1}}{2^{n-1}}|J_r|+\frac{(2^n|C_{J_r}(Q\in L_r)|-|C_{J_r}(Q\in J^\circ_r)|)}{2^{n-1}}\; ,
\end{align}
is valid, both for the generating sets of the orthogonal and symplectic Lie algebras in Eq.~(\ref{eq: so and sp}), only in the latter case does it give rise to pure-state invariants; indeed, by Lm.~\ref{lm: constant commutants},
\begin{align}\label{eq: pre-factor}
    1-\frac{|C_{J_r}(Q\in J^\circ_r)|-|C_{J_r}(Q\in L_r)|}{2^{n-1}}
    =1-\frac{(-1)^{\Arf(q_r)}2^{n-1}}{2^{n-1}}
    =1-(-1)^{\Arf(q_r)}\; ,
\end{align}
hence, this pre-factor in Eq.~(\ref{eq: P1 relation}) vanishes for $\Arf(q_r)=0$, that is, for $r\in\cT_+$, rendering Eq.~(\ref{eq: P1 relation}) tautological (the RHS of Eq.~(\ref{eq: P1 relation}) for $R^T=R$, becomes $2^{n-1}(1-(-1)^{\Arf(q_r)})$). Eq.~(\ref{eq: pre-factor}) therefore distinguishes the symplectic from the orthogonal Pauli Lie Jordan algebras (see Def.~\ref{def: PLJ algebra}), and explains why only the former give rise to pure-state invariants.

In addition to highlighting this distinction, the comparison of the isometry groups of quadratic refinements will allow us to prove a counterpart to Thm.~4 of Ref.~\cite{HashagenFlammiaGrossWallman2018}, which establishes the real Clifford group $\Cl^\R_n$ with $\OG^+(2n,\zz_2)\cong\pi(\Cl^\R_n)$ as an orthogonal $2$-design. By contrast, the group $\OG^-(2n,\zz_2)$ seems to have received much less attention than its $+$-type counterpart (see, however, Ref.~\cite{CalderbankRainsShorSloane1996,DehaeneDeMoor2003} and the recent Ref.~\cite{GargiuloHerringerZeier2026}). In Sec.~\ref{sec: 2-designs}, we prove that similar to its more popular cousin, $\OG^-(2n,\zz_2)$ gives rise to complex projective state $2$-designs. In addition, the next section collects a number of other applications of the structure of pure-state invariants that promotes their usefulness beyond its immediate geometric description of $n$-qubit pure states.

\section{Applications}\label{sec: applications}

Thm.~\ref{thm: pure-state invariants as quadratic refinements} identifies update-closed Boolean pure-state invariants with the level sets of quadratic refinements with Arf invariant $1$. Unlike the prevalently discussed refinement with Arf invariant $0$ in the Weyl representation (see Eq.~(\ref{eq: Weyl operators})), we have seen that it gives rise to interesting properties. In this section, we collect some further corollaries of Thm.~\ref{thm: all weighted quadratic invariants} and Thm.~\ref{thm: all update-closed pure-state invariants}, thus highlighting the broader applicability of our results.

First, in Sec.~\ref{sec: uncertainty relations}, we show that Thm.~\ref{thm: all update-closed pure-state invariants} is equivalent to a characterisation of all weighted Pauli subsets whose total variance is constant, which in turn give rise to uncertainty relations in variance.  This is another indication that underlying the uncertainty relations in Cor.~\ref{cor: all variance uncertainty relations} are certain maximal complementarity structures, which we analyse in more detail in Ref.~\cite{FrembsHoehnNataleWever2026b}. In Sec.~\ref{sec: Pauli channels} and Sec.~\ref{sec: positive Pauli projection maps}, we further relate these results to benchmarking parameters for Pauli channels, as well as to positivity of Pauli projections. The remaining sections exploit the symmetry properties of update-closed pure-state invariants. Sec.~\ref{sec: 2-designs} shows that the stabiliser groups of these pure-state invariants give rise to complex projective $2$-designs. Finally, in Sec.~\ref{sec: reconstruction} we make contact with the reconstruction programme laid out in Ref.~\cite{Hoehn2017}, which uniquely identifies $\PSU(4)$ as the subgroup within all Bloch symmetries $\SO(15)$ that preserve the structure underlying pure-state invariants in the two-qubit case. In light of the present results, this motivates a generalisation of the reconstruction to arbitrary $n$, which we leave to future work.

\subsection{Complementarity and variance-based uncertainty relations}\label{sec: uncertainty relations}

So far, we have mostly considered pure states. Yet, the quadratic measure $I$ in Def.~\ref{def: I} is defined for all quantum states. For applications, it is important to understand its behaviour away from the purity boundary. To this end, we point out that pure-state invariants satisfy the following important bound, justifying in part its interpretation as an information measure in Ref.~\cite{BruknerZeilinger1999,BruknerZeilinger2001,Hoehn2017,HoehnWever2017}.

\begin{lemma}\label{lm: mixed state upper bound}
    For any pure-state invariant $\zeta\in\R^V_{\geq 0}$ and mixed state $\rho$, it holds that
    \begin{equation*}
        I_\rho(\zeta)\leq I_\psi(\zeta)\; .
    \end{equation*}
\end{lemma}

\begin{proof}
    Note that $I_\rho(\zeta)$ is a quadratic function in the Bloch components of $\rho$ whose coefficients are positive by assumption. It follows that $I_\rho(\zeta)$ is convex, that is, $I_{p\rho+(1-p)\rho'}(\zeta)\leq pI_\rho(\zeta)+(1-p)I_{\rho'}(\zeta)$ for all $\rho,\rho'\in\cD(\C^{2^n})$ and $p\in[0,1]$. Now, let $\rho=\sum_ip_ip_{\psi_i}$ be the decomposition of $\rho$ into orthogonal rank-1 eigenspaces (possibly including degeneracies). Then
    \begin{align*}
        I_\rho(\zeta)
        =I_{\sum_ip_ip_{\psi_i}}(\zeta)
        \leq \sum_ip_iI_{\psi_i}(\zeta)
        =I_\psi(\zeta)\; ,
    \end{align*}
    where we used that $I_\psi(\zeta)$ is a pure-state invariant in the last step.
\end{proof}

Ref.~\cite{BruknerZeilinger1999,BruknerZeilinger2001,Hoehn2017,HoehnWever2017} argue that a suitable information measure for quantum systems needs to take into account that measurements are generally complementary. The classical Shannon entropy is expressed for a single measurement only and thus inadequate from this vantage point. By taking into account more than one measurement, an experimenter needs to also account for the uncertainty between measurements. The classical way to express this is in terms of variances.\footnote{In more modern terms, uncertainty is measured using entropic quantities. In Ref.~\cite{FrembsHoehnNataleWever2026b}, we relate pure-state invariants with complementarity structures and find that these give rise to entropic uncertainty relations, too.}

Recall that the variance of an observable $O\in\LHsa$ in the state $\rho\in\Den$ is given by,
\begin{align}\label{eq: variance}
    \Var_\rho(O)
    =\mathbb{E}_\rho(O^2)-\mathbb{E}^2_\rho(O)
    =\tr[O^2\rho]-\tr[O\rho]^2\; .
\end{align}
We will use the shorthand $\Var_\psi:=\Var_{p_\psi}$. Now, consider a scenario where an observable $O_k$ from a set $\{O_k\}_k$ is measured with probability weight $\mu(O_k)$. Then the averaged preparation uncertainty is the weighted sum of the variances of the individual observables. In the case that the observables are taken from the $n$-qubit Pauli group, Thm.~\ref{thm: all weighted quadratic invariants} implies a characterisation of all tight variance-based uncertainty relations, whose bound is constant on all pure states.

\begin{corollary}\label{cor: all variance uncertainty relations}
    Let $\mu$ be a probability distribution over $\aP_n$. Then any tight uncertainty relation such that $\sum_{Q\in\aP_n}\mu(Q)\Var_\psi(Q)=\mathrm{const}$ for all $\psi\in\CP^{2^n-1}$  corresponds with a pure-state invariant.
\end{corollary}

\begin{proof}
    Inserting Eq.~(\ref{eq: variance}) and using $P^2=\one$ for every Hermitian Pauli $P\in\aP_n$, one finds
    \begin{equation}\label{eq: uncertainty relation}
        \sum_{Q\in\aP_n}\mu(Q)\Var_\psi(Q)
        =\sum_{Q\in\aP_n}\mu(Q)(\tr[Q^2p_\psi]-\tr[Qp_\psi]^2)
        =1-I_\psi(\mu)\; ,
    \end{equation}
    which is constant for all $\psi\in\CP^{2^n-1}$ if and only if $I_\psi(\mu)=\mathrm{const}$, hence, Thm.~\ref{thm: all weighted quadratic invariants} applies. Moreover, $1-I_\psi(\mu)$ is a tight lower bound for the average uncertainty by Lm.~\ref{lm: mixed state upper bound}.
\end{proof}

Cor.~\ref{cor: all variance uncertainty relations} considers the (individual) variances of (dichotomic) Pauli observables which are highly degenerate (for $n>>1$). On the other extreme lie uncertainty relations for non-degenerate observables. For instance, one often considers non-degenerate observables whose eigenvectors correspond with sets of mutually unbiased bases \cite{Ivonovic1981,WoottersFields1989,SanchezRuiz1994,Azarchs2004,BallesterWehner2007}. Such bases are readily obtained from common eigenbases of stabiliser states of maximal sets of mutually commuting Pauli operators. The corresponding non-degenerate observables are thus associated with maximal Abelian stabiliser subgroups of the Pauli group $\{A_k\}_{k=0}^{K}$ for $A_k<\cP_n$, that is, they are of the form $O_k=\sum_{\chi\in A^*_k}o_\chi\Pi_{\chi}$ with outcomes $o_\chi\in\R$ such that $o_\chi\neq o_{\chi'}$ for $\chi\neq\chi'$ and where $A^*_k$ denotes the set of characters of $A_k$.

Uncertainty relations for sets of non-degenerate observables associated with Pauli operators in this way have been studied before. In particular, for two as well as maximal sets of $2^n+1$ such observables one obtains entropic uncertainty relations with strong bounds \cite{BallesterWehner2007,WehnerWinter2010}. As with uncertainty relations in general, these arise from the complementary nature of quantum observables \cite{Bohr1928,BuschHeinosaariSchultzStevens2020,StarkeEtAl2026}. In the setting of non-degenerate observables associated to the Pauli group, this connection is particularly strong: the complementary nature of observables is expressed in the mutual unbiasedness of their eigenspaces. Put in information-theoretic terms: \emph{two observables are complementary if certainty about (the outcome of) one implies complete uncertainty about (the outcomes of) the other and vice versa}. A set of observables is complementary if any two of its members are.

In Ref.~\cite{FrembsHoehnNataleWever2026b}, we study this notion of complementarity more generally for degenerate observables associated with the $n$-qubit Pauli group, that is, observables whose eigenspaces correspond with the joint eigenspaces of Pauli operators in not necessarily maximal Abelian subgroups $A<\aP_n$. In particular, we relate maximal sets (of cardinality $2^n+1$) of such complementary observables with pure-state invariants and prove that similar to the non-degenerate case, these sets give rise to (entropic) uncertainty relations with strong bounds.\footnote{The characterisation of non-maximal complementary Abelian subgroups in the Pauli and the specification of the bounds on uncertainty relations they give rise to remain open problems.} As alluded to in Sec.~\ref{sec: update-closed Boolean invariants}, maximal sets of complementary observables associated with the Pauli group correspond with $h$-systems (see Def.~\ref{def: h-systems}), for which a classification is still missing. Still, with Thm.~\ref{thm: all update-closed pure-state invariants} we obtain the following sub-classification. We again identify $V\cong\aP_n$, in particular, Abelian subgroups $A<\aP_n$ up to phase with isotropic subspaces $U\subset V$, $\omega_U=0$.

\begin{corollary}\label{cor: all Jordan-closed complementary sets}
    Proper maximal, Jordan-closed complementary sets of Abelian stabiliser subgroups ($-\one\notin\tA<\cP_n$), correspond with partitions of update-closed pure-state invariants,
    \begin{align}
        J_a^\circ
        =\bigcupdot_{i=1}^{2^n+1}\tA_i^\circ\; ,
    \end{align}
    into subgroups of dimension $n-1$ for some $a\in T_-$ (see Def.~\ref{def: family of pure-state invariants}).
\end{corollary}

\begin{proof}
    Let $\{\tA_i\}_{i=1}^{2^n+1}$ be a maximal set of complementary Abelian Pauli subgroups of dimension smaller than $n$, and set $X:=\bigcupdot_{i=1}^{2^n+1}\tA_i^\circ$. By Thm.~3 in Ref.~\cite{FrembsHoehnNataleWever2026b}, $X$ defines a Boolean pure-state invariant. Now, if the complementary set is Jordan-closed, then by Thm.~\ref{thm: all update-closed pure-state invariants} and Cor.~\ref{cor: completeness}, there exists $a\in T_-$ such that $X=J_a^\circ$. Moreover, by Lm.~\ref{lm: Jordan-closed ovoids}, every proper Jordan-closed $m$-ovoid has $m=2^{n-1}-1$. Consequently, $|\tA_i|=2^{n-1}$, and therefore $\dim(\tA_i)=n-1$.

    Conversely, suppose that $J_a^\circ=\bigcupdot_{i=1}^{2^n+1}\tA_i^\circ$ with $\dim(\tA_i)=n-1$ for all $i$ is an update-closed pure-state invariant. By Thm.~\ref{thm: pure-state invariants as quadratic refinements}, $J_a=\{v\in V\mid q_a(v)=0\}$, where $q_a$ is a quadratic refinement of $\omega$ with $\Arf(q_a)=1$, equivalently Witt index $i_W(q_a)=n-1$. We claim that the isotropic subspaces $\tA_i$ are complementary, that is, for all $i\neq j$,
    \begin{align*}
        \tA_i^\circ\cap\tA_j^\perp
        =\emptyset\; .
    \end{align*}
    Indeed, assume that for some $i\neq j$ there exists $0\neq u\in\tA_i\cap\tA_j^\perp$. Since the above union is disjoint, $u\notin\tA_j$. Moreover, $q_a(u)=0$, $q_a|_{\tA_j}=0$ and $\omega(u,\tA_j)=0$ from which it follows that $\langle \tA_j,u\rangle$ is an isotropic subspace of dimension $n$, contradicting $i_W(q_a)=n-1$. Hence, $\{\tA_i\}_{i=1}^{2^n+1}$ is a maximal set of complementary Abelian subgroups and thus defines an $h$-system (see Def.~\ref{def: h-systems}). Finally, since $J_a$ is Jordan-closed by Thm.~\ref{thm: all update-closed pure-state invariants}, the resulting complementary set is Jordan-closed.
\end{proof}

\subsection{Pauli channels with state-independent fidelity and output purity}\label{sec: Pauli channels}

Pauli channels are a key tool in quantum information, specifically, they are central to (randomised and unitarity) benchmarking \cite{WallmanEtAl2015,HelsenEtAl2022}, represent standard error models and thus are prevalent in quantum error correction \cite{FlammiaWallman2020}, and more generally in Pauli sampling and simulation \cite{HuangKuengPreskill2020,HuangKuengPreskill2021}. In this section, we show that pure-state invariants characterise Pauli channels whose input fidelity and output purity are isotropic over pure inputs. Consequently, the usual Haar-averaged contraction and unitarity parameters can be recovered from an arbitrary single pure input.

Define the (weighted) Pauli channel $\Lambda_\mu(\rho)=\sum_{Q\in\aP_n}\mu(Q)Q\rho Q^\dagger$ for $\rho\in\cD(\C^{2^n})$ and with $\mu$ a probability distribution over $\aP_n$, that is, $\sum_{Q\in\aP_n}\mu(Q)=1$. As a consequence, note that $\Lambda_\mu$ is completely positive, unital and trace-preserving. For $\mu_X:=\frac{1}{|X|}\ind_X$ the uniform mixture of the subset $X\subset\aP_n$, we also write $\Lambda_X(\rho):=\frac{1}{|X|}\sum_{Q\in X}Q\rho Q^\dagger$.

Clearly, a Pauli channel is diagonal in the Pauli basis, explicitly for $\rho=\frac{1}{2^n}\sum_{Q\in\aP_n}\alpha_QQ$,
\begin{align}\label{eq: Pauli channel}
    \Lambda_\mu(\rho)
    =\frac{1}{2^n}\one+\frac{1}{2^n}\sum_{\one\neq Q\in\aP_n}\alpha_Q\left(\sum_{Q'\in\aP_n}(-1)^{\omega(Q,Q')}\mu(Q')\right)Q
    =:\frac{1}{2^n}\sum_{v\in V}\gamma_\mu(v)\alpha_vW_v\; ,
\end{align}
with eigenvalue function $\gamma_\mu(v):=\sum_{w\in V}\mu(w)(-1)^{\omega(v,w)}$ such that $\Lambda_\mu(W_v)=\gamma_\mu(v)W_v$. The corresponding Pauli transfer matrix $M(\Lambda_\mu)=\mathrm{diag}(\gamma_\mu(v))_{0\neq v\in V}\in\mathrm{Mat}_{4^n-1}(\R)$ describes the action of $\Lambda_\mu$ on Bloch vectors $\alpha\in\R^{4^n-1}$ by $\alpha'=M(\Lambda_\mu)\alpha$. (Note that since $\Lambda_\mu$ is unital, the identity term is fixed, in particular, $\gamma_\mu(0)=\sum_{v\in V}\mu(W_v)=1$.)

For (randomised and unitarity) benchmarking, one is interested in Haar-random averages of the moments of $M(\Lambda_\mu)$. Write $\mu_\mathrm{FS}$ for the Fubini-Study measure on $\CP^{2^n-1}$, then
\begin{equation}\label{eq: benchmarking parameters}
\begin{aligned}
    \mathbb{E}_{\mu_\mathrm{FS}}[F_\psi(\Lambda_\mu)]
    &:=\mathbb{E}_{\mu_\mathrm{FS}}\left[\tr[p_\psi\Lambda_\mu(p_\psi)]\right]
    =\frac{1}{2^n}+\frac{\tr[M(\Lambda_\mu)]}{2^n(2^n+1)}
    =\frac{1+(2^n-1)p(\Lambda_\mu)}{2^n}
    \\
    \mathbb{E}_{\mu_\mathrm{FS}}[P_\psi(\Lambda_\mu)]
    &:=\mathbb{E}_{\mu_\mathrm{FS}}\left[\tr[(\Lambda_\mu(p_\psi))^2]\right]
    =\frac{1}{2^n}+\frac{||M(\Lambda_\mu)||^2_2}{2^n(2^n+1)}
    =\frac{1+(2^n-1)u(\Lambda_\mu)}{2^n}\; ,
\end{aligned}
\end{equation}
where $||M||^2_2:=\tr[M^TM]$ is the Frobenius norm, and the relevant benchmarking parameters are the first and second moments of the Pauli transfer matrix, which are called the contraction $p(\Lambda_\mu)$ and unitarity $u(\Lambda_\mu)$ \cite{WallmanEtAl2015}, and are readily computed to be given by the following expressions,\footnote{Clearly, $\Lambda_\mu$ is unitary if and only if $\mu(Q)=\delta_{QQ_0}$ for some $Q_0\in\aP_n$, in which case $u(\Lambda_\mu)=1$.}
\begin{align*}
    p(\Lambda_\mu)
    &=\frac{\tr[M(\Lambda_\mu)]}{4^n-1}
    =\frac{4^n\mu(\one)-1}{4^n-1} &
    u(\Lambda_\mu)
    &=\frac{||M(\Lambda_\mu)||^2_2}{4^n-1}
    =\frac{4^n\sum_{Q\in\aP_n}\mu^2(Q)-1}{4^n-1}\; .
\end{align*}

Evidently, the contraction parameter bears close resemblance with pure-state invariants.

\begin{corollary}\label{cor: state-independent contraction}
    All Pauli channels $\Lambda_\mu$ for which $F_\psi(\Lambda_\mu)=\tr[p_\psi\Lambda_\mu(p_\psi)]$ is state-independent correspond with pure-state invariants in Thm.~\ref{thm: all weighted quadratic invariants}. In this case, the contraction parameter reads
    \begin{align*}
        p(\Lambda_\mu)
        =\frac{2^nI_\psi(\mu)-1}{2^n-1}\; .
    \end{align*}
\end{corollary}

\begin{proof}
    This follows by simply re-writing the expression in Def.~\ref{def: I},
    \begin{align*}
        I_\psi(\mu)
        =\sum_{Q\in\aP_n}\mu(Q)\tr[p_\psi Q]^2
        =\sum_{Q\in\aP_n}\mu(Q)\tr[(p_\psi Q)^2]
        =\sum_{Q\in\aP_n}\mu(Q)\tr[p_\psi (Qp_\psi Q^\dagger)]
        =\tr[p_\psi\Lambda_\mu(p_\psi)]\; ,
    \end{align*}
    since $Q^\dagger=Q$ for all Paulis $Q\in\aP_n$. The explicit value of $p(\Lambda_\mu)$ is readily computed from Eq.~(\ref{eq: benchmarking parameters}).
\end{proof}

What is more, for the uniform Pauli channels associated with update-closed Boolean pure-state invariants, the output purity provides a state-independent single-input estimator of its unitarity.

\begin{corollary}\label{cor: state-independent unitarity}
    Let $\Lambda_\mu$ be a Pauli channel as in Eq.~(\ref{eq: Pauli channel}). Then
    \begin{align*}
        P_\psi(\Lambda_\mu)
        :=\tr\!\left[(\Lambda_\mu(p_\psi))^2\right]
        =\frac1{2^n}I_\psi(\gamma_\mu^2).
    \end{align*}
    Consequently, $P_\psi(\Lambda_\mu)$ is independent of the pure input state if and only if $\gamma_\mu^2$ is a pure-state invariant in the sense of Def.~\ref{def: I}. In this case, $u(\Lambda_\mu)=\frac{I_\psi(\gamma_\mu^2)-1}{2^n-1}$, in particular, for every $a\in T_-$, the uniform channels associated with $J_a$ and $L_a$ 
    satisfy
    \begin{align*}
        u(\Lambda_{L_a})
        =\frac1{(2^n+1)^2}\; ,\qquad
        u(\Lambda_{J_a})
        =\frac1{(2^n-1)^2}.
    \end{align*}
\end{corollary}

\begin{proof}
    Using Eq.~(\ref{eq: Pauli channel}) and Pauli orthogonality yields, one immediately finds
    \begin{align*}
        \tr\left[(\Lambda_\mu(p_\psi))^2\right]
        =\frac{1}{2^n}\sum_{v\in V}\gamma_\mu(v)^2\alpha_v^2
        =\frac{1}{2^n}I_\psi(\gamma_\mu^2)\; ,
    \end{align*}
    which proves the first statement, and since $\gamma_\mu(0)=1$, Eq.~(\ref{eq: benchmarking parameters}) gives $u(\Lambda_\mu)=\frac{I_\psi(\gamma_\mu^2)-1}{2^n-1}$.
    
    Now, for $q_a=q_W+\omega(a,\cdot)$ with $a\in\cT_-$ (see Lm.~\ref{lm: Arf reps}), the Gauss sum from Eq.~(\ref{eq: inhomogenous solution}) becomes
    \begin{align*}
        \sum_{w\in V}(-1)^{q_a(w)+\omega(v,w)}
        =-2^n(-1)^{q_a(v)}\; .
    \end{align*}
    Together with character orthogonality, this yields for all $0\neq v\in V$
    \begin{align}\label{eq: gamma}
    \begin{split}
        \gamma_{L_a}(v)
        &=\frac{1}{2|L_a|}\sum_{v\in V}(1-(-1)^{q_a(w)})(-1)^{\omega(v,w)}
        =\frac{(-1)^{q_a(v)}}{2^n+1}\; ,\\
        \gamma_{J_a}(v)
        &=\frac{1}{2|J_a|}\sum_{v\in V}(1+(-1)^{q_a(w)})(-1)^{\omega(v,w)}
        =-\frac{(-1)^{q_a(v)}}{2^n-1}\; ,
    \end{split}
    \end{align}
    and $\gamma_{L_a}(0)=\gamma_{J_a}(0)=1$. The stated unitarities now follow immediately from
    \begin{equation*}
        u(\Lambda_\mu)
        =\frac{||M(\Lambda_\mu)||^2_2}{4^n-1}
        =\frac{1}{4^n-1}\sum_{0\neq v\in V}\gamma^2_\mu(v)\; .\qedhere
    \end{equation*}
\end{proof}

More generally, Eq.~(\ref{eq: gamma}) implies the following property for every density operator $\rho$,
\begin{align*}
    \tr[(\Lambda_{L_a}(\rho))^2]
    =\frac{1}{2^n}+
    \frac{\tr[\rho^2]-\frac{1}{2^n}}{(2^n+1)^2}\; ,
    \qquad
    \tr[(\Lambda_{J_a}(\rho))^2]
    =\frac{1}{2^n}+
    \frac{\tr[\rho^2]-\frac{1}{2^n}}{(2^n-1)^2}\; .
\end{align*}
Consequently, the output purity depends only on the input purity, and not on the orientation of the input state in Bloch space. In other words, the Pauli channels $\Lambda_{L_a}$ and $\Lambda_{J_a}$ enact isotropic contractions of the traceless Bloch vector in Hilbert-Schmidt norm (without generally being depolarising). Note, however, that this property is not special to pure-state invariants.\footnote{For uniform subset channels with $\mu_X=\frac{1}{|X|}\mathbf{1}_X$ for $X\subset V$, constancy of $|\gamma_\mu(v)|=\frac{1}{|X|}\sum_{x\in X}(-1)^{\omega(v,x)}$ on $V^\circ$ is equivalent to $X$ being the support of a bent (Boolean) function (defined by its discrete Fourier (or Walsh) transform  having constant magnitude), equivalently, for $X$ to be a Hadamard difference set. The functions $\mathbf{1}_{L_a}=q_a$ and $\mathbf{1}_{J_a}=q_a+1$ are precisely the quadratic bent functions from this family, equivalently $L_a$ and $J_a$ are Hadamard difference sets in the elementary Abelian group $V\cong\zz^{2n}_2$ \cite{Dillon1974,Carlet2021}. For $n\geq 2$, Boolean bent functions $f:\zz^{2n}_2\ra\zz_2$ exist for all degrees $2\leq d\leq n$ \cite{CarletVilla2025}; for $n=1$, bent functions are quadratic.}

Finally, we point out that the overlap of the Pauli channels corresponding to the different sets in the partition $\aP_n=L_a\cupdot J_a$ for $a\in\cT_-$ is state-independent. Indeed,
\begin{align*}
    4^{2n}\tr[(\Lambda_{\aP_n}(p_\psi))^2]
    &=\tr[\left(|J_a|\Lambda_{J_a}(p_\psi)+|L_a|\Lambda_{L_a}(p_\psi)\right)\left(|J_a|\Lambda_{J_a}(p_\psi)+|L_a|\Lambda_{L_a}(p_\psi)\right)]\\
    &=|J_a|^2\tr[(\Lambda_{J_a}(p_\psi))^2]+|L_a|^2\tr[(\Lambda_{L_a}(p_\psi))^2]+2|J_a|\cdot|L_a|\tr[\Lambda_{J_a}(p_\psi)\Lambda_{L_a}(p_\psi)]\; ,
\end{align*}
which with $\Lambda_{\aP_n}(p_\psi)=\frac{1}{2^n}\one$ and the above yields $\tr[\Lambda_{L_a}(p_\psi)\Lambda_{J_a}(p_\psi)]=\frac{1}{2^n+1}$ for all $\psi\in\CP^{2^n-1}$. In particular, $I_\psi(L_a),I_\psi(J_a)$ are pure-state invariants if and only if (any two out of) the quantities $\tr[(\Lambda_{L_a}(p_\psi))^2],\tr[(\Lambda_{J_a}(p_\psi))^2]$ and $\tr[\Lambda_{J_a}(p_\psi)\Lambda_{L_a}(p_\psi)]$ are state-independent.

\subsection{Positive Pauli projection maps}\label{sec: positive Pauli projection maps}

In this section, we provide an alternative characterisation of the update-closure condition for pure-state invariants in Sec.~\ref{sec: update-closed Boolean invariants}. For any $\one\in X\subset\pm\aP_n$ with at most one element in every projective class, that is, such that $P\in X\Rightarrow -P\notin X$, define the Pauli-diagonal projection,
\begin{align}\label{eq: purity Pauli channel}
    \epsilon_X(\rho)
    :=\frac{1}{2^n}\sum_{P\in X}\tr[\rho P]P\; .
\end{align}
In fact, this map is simply the projection of (the Bloch representation of $\rho$) onto the span of $X$.

\begin{lemma}\label{lm: Jordan projection}
    Let $\one\in X\subset\pm\aP_n$ with $P\in X\Rightarrow -P\notin X$. Then $X$ is Jordan-closed if and only if $\epsilon_X$ is positive.
\end{lemma}

\begin{proof}
    Suppose $X$ is not Jordan-closed, then there exist $Q,P\in X$ with $[Q,P]=0$, yet $\{Q,P\}=QP\notin X$. Consider the joint projector $\Pi=\frac{1}{4}(\one+Q+P+QP)$, which in particular is positive. Yet,
    \begin{align}
        \epsilon_X(\Pi)
        =\frac{1}{4}(\one+Q+P)\; ,
    \end{align}
    has a negative eigenvalue $-\frac{1}{4}$, hence, $\epsilon_X$ is not positive.

    To show that Jordan-closure implies positivity for $\epsilon_X$, define the Jordan algebra $\cJ_X=\mathrm{span}_\R\{X\}$. Let $A\geq 0$ and assume that there exists $\cJ_X\ni B:=\epsilon_X(A)<0$. Since $\cJ_X$ is a finite-dimensional Euclidean (equivalently, formally real) Jordan algebra, $B$ has a spectral decomposition \cite{Hanche-OlsenStormer1984},
    \begin{align*}
        B=\sum_{i\in\cI}\lambda_ie_i\; ,
    \end{align*}
    where $\lambda_i\in\R$ for every $i\in\cI$, and $\{e_i\}_{i\in\cI}$ is a set of pairwise orthogonal idempotents. $B\not\geq 0$ thus implies that there exists a projector $e_i$ with $\tr[e_iB]=\lambda_i\tr[e_i]<0$. Yet, this contradicts that
    \begin{equation*}
        \tr[e_iB]
        =\tr[e_i\epsilon_X(A)]
        =\tr[\epsilon_X(e_i)A]
        =\tr[e_iA]\geq 0\; .\qedhere
    \end{equation*}
\end{proof}

\begin{corollary}
    A pure-state invariant $\one\in X\subset\aP_n$ is update-closed if and only if $\epsilon_X$ is positive.
\end{corollary}

\begin{proof}
    This follows directly from Lm.~\ref{lm: Jordan projection} and Lm.~\ref{lm: characterising Jordan}.
\end{proof}

In other words, the family of pure-state invariants in Def.~\ref{def: family of pure-state invariants} can also be interpreted as the shadow of positive projections on the full quantum state space.

\subsection{Symplectic $2$-designs and complex projective state $2$-designs}\label{sec: 2-designs}

In this section, we show that symmetries of the pure-state invariants give rise to complex projective state $2$-designs. To begin with, we recall the closely related definition of a unitary $k$-design.\footnote{For other equivalent definitions, see e.g. Ref.~\cite{Mele2024}.}

\begin{definition}\label{def: unitary k-design}
    A pair $(\cU,\mu)$ for $\mu$ a probability distribution with support restricted to the subset of unitaries $\cU\subset\UH$ (for $\dim(\cH)<\infty$), is called a \emph{unitary $k$-design} if and only if
    \begin{align*}
        \mathbb{E}_\mu[U^{\otimes k}O(U^\dagger)^{\otimes k}]
        =\mathbb{E}_{\mu_{\UH}}[U^{\otimes k}O(U^\dagger)^{\otimes k}]\; ,
    \end{align*}
    for all linear operators $O\in(\cL(\cH))^{\otimes k}\cong\cL(\cH^{\otimes k})$, where $\mathbb{E}_\mu(f(U)):=\int_{\UH} d\mu(U) f(U)$, and $\mu_{\UH}$ denotes the unique normalised Haar measure on $\UH$.
\end{definition}

$k$-designs for other groups are defined similarly, e.g. symplectic and orthogonal $k$-designs for the symplectic and orthogonal subgroups of the unitary group.

\begin{definition}\label{def: state k-design}
    A pair $(\cS,\mu)$ for $\mu$ a probability distribution with support restricted to a subset of pure quantum states $\cS\subset\CP^{\dim(\cH)-1}$, is called a \emph{(complex projective) state $k$-design} if and only if
    \begin{align*}
        \mathbb{E}_\mu[p_\psi^{\otimes k}]
        =\mathbb{E}_{\mu_\mathrm{FS}}[p_\psi^{\otimes k}]
        =\mathbb{E}_{\mu_{\UH}}[(Up_\psi U^\dagger)^{\otimes k}]\; ,
    \end{align*}
    where $\mu_\mathrm{FS}$ denotes the normalised (unitarily-invariant) Fubini-Study measure on $\CP^{2^n-1}$ (equivalently, the push-forward of the normalised Haar measure $\mu_{\UH}$ on $\UH$).
\end{definition}

State $k$-designs arise naturally from unitary $k$-design $(\cU,\mu)$ via $\{\mu(U),U\psi\}_{U\in\cU}$ with $\psi\in\CP^{\dim(\cH)-1}$ some fiducial state. We will be interested in the case where $\cU=G<\UH$ is a subgroup of the unitary group. For instance, the $n$-qubit Clifford group $G=\Cl_n$ is well known to be a unitary $3$-design \cite{Webb2016,Zhu2017}, which implies that $\{\mu(U),U\psi\}_{U\in\Cl_n}$ is a complex projective $3$-design; in particular, this is true for the orbit of stabiliser states \cite{KuengGross2015}. However, note that in order for $\{\mu(U),U\psi\}_{U\in G}$ to be a state $k$-design it is not necessary that $G$ is a unitary $k$-design itself.

To see this, define the $k$-th moment operator $\cM^{(k)}_\mu$ of $\mu$ for every $O\in\cL(\cH^{\otimes k})$ by
\begin{align}\label{eq: moment operator}
    \cM^{(k)}_\mu(O)
    :=\mathbb{E}_\mu[U^{\otimes k}O(U^\dagger)^{\otimes k}]\; .
\end{align}
Clearly, $\{\mu(U),U\}_{U\in\cU}$ is a unitary $k$-design if and only if $\cM^{(k)}_\mu=\cM^{(k)}_{\mu_{\UH}}$. To test this, recall (e.g. from Ref.~\cite{Mele2024}) that for any compact subgroup $G\leq\UH$ with Haar measure $\mu_G$,\footnote{In particular, for finite groups $G$, $\mu_G$ is the uniform measure, that is, $\mu_G(g)=\frac{1}{|G|}$.\label{fn: uniform measure}}
\begin{align}\label{eq: moment operator as projector}
    \cM^{(k)}_{\mu_G}
    =\Pi_{\mathrm{Comm}(G,k)}\; ,
\end{align}
where $\Pi_{\mathrm{Comm}(G,k)}$ denotes the orthogonal projector (with respect to the Hilbert Schmidt inner product) onto the commutant $\mathrm{Comm}(G,k):=\{O\in\cL(\cH^{\otimes k})\mid[O,U^{\otimes k}]=0\ \forall U\in G\}$. Now, by Schur-Weyl duality (applied to the unitary group $\UH$) the commutant $\mathrm{Comm}(\UH,k)$ is spanned by permutations on tensor factors \cite{Mele2024}. For $k=2$, this implies that $\mathrm{Comm}(\UH,2)$ is spanned by $\one\otimes\one$ and $\SWAP$, from which one readily computes the second moment operator,
\begin{align}\label{eq: second moment}
    \cM^{(2)}_{\mu_{\UH}}(O)&=\frac{\tr[O(\one+\SWAP)]}{2^n(2^n+1)}\Pi_\sym+\frac{\tr[O(\one-\SWAP)]}{2^n(2^n-1)}\Pi_\Lambda\; .
\end{align}

In contrast, for a state $k$-design in Def.~\ref{def: state k-design}, it is (necessary and) sufficient that the moment operators agree on the symmetric subspace $\Sym^{(k)}(\C^{2^n})$ only, for which Eq.~(\ref{eq: second moment}) reduces to
\begin{align}\label{eq: symmetric second moment}
    \cM^{(2)}_{\mu_{\UH}}(p_\psi^{\otimes 2})
    =\frac{1}{2^{n-1}(2^n+1)}\Pi_\sym\quad\quad \forall\psi\in\CP^{2^n-1}\; .
\end{align}
For $\{\mu_G(U)U\psi\}_{U\in G}$ with $\psi\in\CP^{\dim(\cH)-1}$ and $G<\UH$ to be a complex projective state $2$-design, the commutant of $G$ must therefore be constant on the symmetric subspace.

\begin{lemma}\label{lm: Schur for state designs}
    Let $G<\UH$ with $\cH=\C^{2^n}$ be a closed unitary subgroup. Then $\{\mu_GU\psi\}_{U\in G}$ with $\mu_G$ the unique normalised Haar measure$^{\ref{fn: uniform measure}}$ on $G$ defines a complex projective state $2$-design for every state $\psi\in\CP^{2^n-1}$ if and only if $G$ acts irreducibly on the symmetric subspace $\Sym^2(\C^{2^n})$, that is,
    \begin{align}\label{eq: symmetric commutant equivalence}
        \Comm(G,2)|_\sym
        =\Pi_\sym\Comm(G,2)\Pi_\sym
        =\C\cdot\Pi_\sym\; .
    \end{align}
\end{lemma}

\begin{proof}
    This follows from Eq.~(\ref{eq: moment operator as projector}) and Eq.~(\ref{eq: symmetric second moment}). In particular, note that Eq.~(\ref{eq: symmetric commutant equivalence}) is equivalent to $\End_G(V)=\C\cdot\one_V$ for $V=\Sym^2(\C^{2^n})$, where $\End_G(V)=\{A\in\mathrm{End}(V)\mid A\rho(g)=\rho(g)A\ \forall g\in G\}$ is the space of equivariant linear operators on $V$ for the restriction of the tensor-product representation $g\mapsto g^{\otimes 2}$ to $V$. By Schur's lemma this is equivalent to $G$ acting irreducibly on $V$.
\end{proof}

In fact, since global phases are irrelevant, it is sufficient to restrict to subgroups of the projective unitary group. It is thus natural to ask whether pure-state invariants give rise to state $k$-designs for some (closed) unitary subgroup $G\leq \U(2^n)$? This is indeed the case for the groups
\begin{align}\label{eq: unitary lift}
    G_a
    &:=\Sp_{W_aK}(2^{n-1})\cap\Cl_n\; ,
\end{align}
for all $a\in\cT_-$ and $k=2$, where, $G_a$ is a unitary lift of $\PStab(L_a)$. (Recall from Cor.~\ref{cor: Clifford intersections} that $\mathrm{P}(G_a)=\PStab(L_a)=\mathrm{P}(\Sp_{W_aK}(2^{n-1})\cap\Cl_n)$ for all $n\geq 2$; the case $n=1$ follows immediately from Lm.~\ref{lm: transitivity}.) For a projective unitary $[U]\in\mathrm{PU}(\cH)$, we write $[U]\psi$ for its induced action on $\CP^{d-1}$, that is, the action of any unitary lift $U$ on the ray $\psi$, equivalently $p_\psi\mapsto p_{U\psi}=[U]p_\psi[U]^\dagger$.

\begin{theorem}\label{thm: 2-design from invariant symmetries}
    For every $a\in\cT_-$ and pure state $\psi\in\CP^{2^n-1}$, $\{\frac{1}{|\PStab(L_a)|},[U]\psi\}_{[U]\in\PStab(L_a)}$ is a complex projective state $2$-design, and $(G_a,\mu_{G_a}=\frac{1}{|G_a|}=\frac{1}{2|\PStab(L_a)|})$ a symplectic 2-design.
\end{theorem}

\begin{proof}
    For $n=1$, $\PStab(L_a)=\pCl_1$ by Lm.~\ref{lm: transitivity} and the result thus holds since $\Cl_n$ is a unitary $3$-design for all $n\geq 1$ \cite{Webb2016,Zhu2017}. For $n\geq 2$, use that $\aP_n$ is a basis for $\LHsa$ to write any element $X\in\Comm(G_a,2)$ as $X=\sum_{P,Q\in\aP_n}\beta_{PQ}P\otimes Q$. Since $\mathrm{P}(G_a)=\PStab(L_a)$ contains the projective Pauli group, for every $P\in\aP_n$ there exists a phase $\lambda_P$ such that $\lambda_PP\in G_a$, hence, commutation with $G^{\otimes 2}_a$ implies commutation with $P^{\otimes 2}$ for every $P$, and therefore $\beta_{PQ}=0$ unless $P=Q$. Moreover, by Lm.~\ref{lm: transitivity} and Witt's extension theorem (see e.g. Thm.~7.4 in Ref.~\cite{Taylor1992}), $G_a$ acts transitively on $L_a$ and its complement $J^\circ_a=\aP^\circ_n\backslash L_a$, hence,
    \begin{align}\label{eq: stabiliser commutant algebra}
        \Comm(G_a,2)
        =\langle\one\otimes\one,\sum_{P\in L_a}P\otimes P,\sum_{P\in J_a}P\otimes P\rangle
        =\langle \one\otimes\one,\SWAP,\Pi_a\rangle
        =\langle \Pi_\sym,\Pi_\Lambda,\Pi_a\rangle\; ,
    \end{align}
    where we used the identity $\sum_{P\in\aP_n}P\otimes P=2^n\SWAP$ with $\SWAP=\Pi_\sym-\Pi_\Lambda$ as well as Lm.~\ref{lm: anti-symmetric eigenvalues}, by which $\sum_{P\in L_a}P\otimes P=M^-_a=\frac{1}{2}(2^n\SWAP-4^n\Pi_a)$ (similarly, $\sum_{P\in J_a}P\otimes P=M^+_a=\frac{1}{2}(2^n\SWAP+4^n\Pi_a)$, and where $\Pi_a=\frac{1}{2^n}\kett{a}\braa{a}$ for $\kett{a}=\vect(W_a)=(1\otimes W_a)|\Omega\rangle$ (cf. Eq.~(\ref{eq: vec})). Now, since both $\Pi_\Lambda$ and $\Pi_a$ have support on $\Lambda^2(\C^{2^n})$ only, Eq.~(\ref{eq: stabiliser commutant algebra}) implies that $\Comm(G_a,2)|_\sym\sim \Pi_\sym$, and the first statement thus follows from Lm.~\ref{lm: Schur for state designs}. 
    
    Finally, since $G_a<\Sp_{W_a}(2^{n-1})$ by definition, and $\Comm(\Sp(2^{n-1}),2)$ has dimension three for $n\geq 2$ \cite{GoodmanWallach2009}, the respective commutant algebras agree, and the second statement thus follows from Eq.~(\ref{eq: stabiliser commutant algebra}) and Eq.~(\ref{eq: moment operator as projector}). For $n=1$, $\Sp(1)\cong\SU(2)$, hence, $G_a=\Sp(1)\cap\Cl_1=\SU(2)\cap\Cl_1$, hence, $\mathrm{P}(G_a)=\Cl_1$ and the result thus follows since $\Cl_1$ is a symplectic $3$-design \cite{Webb2016,Zhu2017}.
\end{proof}

The proof of the second statement in Thm.~\ref{thm: 2-design from invariant symmetries} is analogous to that of Thm.~4 in Ref.~\cite{HashagenFlammiaGrossWallman2018}; what distinguishes orthogonal and symplectic cases is the first statement. In related work, it was recently shown that complex projective state $k$-designs can be constructed from symplectic $k$-designs (see Thm.~1 in Ref.~\cite{WestEtAl2024}), raising the question whether more efficient complex projective state $k$-designs can be constructed under the restriction to symplectic $k$-designs. Thm.~\ref{thm: 2-design from invariant symmetries} provides such an example for the case $k=2$, since (by Cor.~\ref{cor: Clifford intersections}) the unitary lifts $G_a$ in Eq.~(\ref{eq: unitary lift}) are not only Clifford unitaries (whose action on Pauli labels is given by symplectic transformations in $\Sp(2n,\zz_2)$), but elements in an ambient symplectic group $\Sp_{W_a}(2^{n-1})$. For an explicit representation of the various symplectic Lie algebras corresponding to the sets $L_a$, $a\in\cT_-$ and the relation between them, see App.~\ref{app: symplectic representation}.

Ref.~\cite{GargiuloHerringerZeier2026} independently proves the stronger fact that the corresponding Clifford transvection groups are symplectic $3$-designs. Here we identify this group from a different starting point---as the stabiliser of an update-closed pure-state invariant---and use its commutant to connect the design property directly with the purity geometry and the associated state 2-design orbit.

In turn, given Thm.~\ref{thm: 2-design from invariant symmetries}, one may ask if the groups $G_a$ exhaust all symplectic $2$-designs in the Clifford group. This is not the case, as there exist subgroups $H<O^-(2n,\zz_2)$ such that $\pi^{-1}(H)<G_a$ and whose action remains transitive on the level sets of quadratic refinements and thus, by the same arguments as in the proof of Thm.~\ref{thm: 2-design from invariant symmetries}, have the same commutant algebra as in Eq.~(\ref{eq: stabiliser commutant algebra}) \cite{LiebeckPraegerSaxl2010,GiudiciGlasbyPraeger2023}.

\subsection{On the reconstruction of quantum theory from questions}\label{sec: reconstruction}

We end our list of applications by recalling a key motivation for this work, and by placing our results within the physical context from which it arises.

The modern reconstruction programme seeks to derive the mathematical formalism of quantum theory from operational or information-theoretic principles rather than postulating Hilbert-space structure outright; representative approaches include Ref.~\cite{Hardy2016,MasanesMueller2010,ChiribellaEtAl2011,Hoehn2016,Hoehn2017,Muller2021}. The framework most directly relevant here is the reconstruction from `rules of information acquisition', developed in Ref.~\cite{Hoehn2016,Hoehn2017,HoehnWever2017}, for which the Brukner-Zeilinger information measure plays a fundamental role. As pointed out in Ref.~\cite{BruknerZeilinger1999,BruknerZeilinger2001}, when seeking to describe the information contained in a state over complementary observables, the classical Shannon entropy fails to account for measurement uncertainties. By contrast, the quadratic $I$ is essentially a sum of variances (see Sec.~\ref{sec: uncertainty relations}).

The surprising insight of Ref.~\cite{HoehnWever2017} is that a finite set of pure-state invariants - with respect to a given Pauli frame - is sufficient to reconstruct the pure-state manifold. More precisely, the authors of Ref.~\cite{HoehnWever2017} show that in the two-qubit case pure-state invariants, together with the demand that they be preserved under linear transformations of the Bloch vector, single out the projective unitary group, from which the pure-state manifold is then constructed as the orbit of a state of maximal information, e.g. a stabiliser state. This result is not only interesting for conceptual reasons, but also as an insight into the geometry of Bloch vectors, which for $n>2$ remains poorly understood. It is thus natural to ask whether the explicit construction in Ref.~\cite{HoehnWever2017} generalises beyond the two-qubit case, as well as to seek a more formal derivation of their result that highlights these geometric aspects more clearly. In particular, given that the most direct generalisation of the two-qubit pure-state invariants is given by the family of update-closed invariants in Def.~\ref{def: family of pure-state invariants} (see Ref.~\cite{FrembsHoehnNataleWever2026a}), one may hope that these questions have a positive answer already in this (from the perspective of Thm.~\ref{thm: all weighted quadratic invariants} and Thm.~\ref{thm: all Boolean pure-state invariants}) highly restricted setting. It turns out that this is indeed the case. Since the result has a rather different flavour and requires substantial additional machinery, we leave its proof to future work, and finish this section by placing the result within the broader conceptual context of reconstructions of quantum theory.

Unlike most reconstructions of quantum mechanics \cite{Hardy2016,MasanesMueller2010,ChiribellaEtAl2011,Hoehn2016,Muller2021}, the programme laid out in Ref.~\cite{Hoehn2017} assumes fewer `global' constraints and thus remains more constructive. Still, the explicitly worked-out two-qubit case in Ref.~\cite{HoehnWever2017} relies on a postulate of continuous symmetries. While justifiable on physical grounds on macroscopic scales, it is arguably the weak link in an otherwise constructivist approach. It is therefore tempting to seek a stronger reconstruction result that, instead of relying on continuous symmetries, uses only discrete notions such as the compatibility condition in Def.~\ref{def: update-closure}. This latter programme is, of course, much more ambitious, as there are fewer structural results that can be employed and one has to rely on combinatorial arguments instead. While this puts additional burden on proofs, it offers new structural insights (of which the characterisation of pure-state invariants in this and previous works \cite{HoehnWever2017,FrembsHoehnNataleWever2026a,FrembsHoehnNataleWever2026b} is a prime example). Indeed, if successful, it would yield a novel reconstruction of the non-associative Jordan product, and thus the formal structure of measurements in quantum mechanics, and possibly provide it with a more intuitive physical interpretation. This line of reasoning is not new, but has pre-cursors in e.g. Ref.~\cite{Leifer2006,LeiferSpekkens2013,Wetering2019,ParzygnatRusso2021,Baez2022,ParzygnatBuscemi2023}. A detailed development of this programme is left to future work.

\section{Conclusion}\label{sec: conclusion}

The set of $n$-qubit pure states comprises a space with rich geometric structure, of which only the cases for small $n$ seem to be well understood. Here, we provide novel insight into the intricate geometry of $n$-qubit pure quantum states by studying quadratic purity invariants that are diagonal in Bloch components. Surprisingly, this restriction yields a rich family of purity constraints, whose characterisation reveals deep connections with the theory of symplectic polar spaces. Restricting further to the subfamily of constraints that are closed under a physically motivated notion of `information update' (see Def.~\ref{def: update-closure}), reveals yet further previously hidden structure in the pure-state manifold: we show that the complements of their supports furnish Pauli generating sets, or ‘cross-sections’, of different embeddings of symplectic Pauli Lie algebras $\spf(2^{n-1})\hookrightarrow\su(2^n)$. We trace the appearance of this symplectic structure (and its relation with the well-known discrete symplectic structure of all Pauli observables $\aP_n$ \cite{DehaeneDeMoor2003,Gross2006}), to covariance with quaternionic structures on $\C^{2^n}$, and distinguish it from the corresponding real/orthogonal structures via the Arf invariant of quadratic refinements of the symplectic form on $\aP_n$.

Our analysis offers new structural insights not only into the geometry of the Bloch representation of pure quantum states but also the $n$-qubit Pauli group itself and its symmetries more generally. The underlying theory of symplectic polar spaces is not new, but its applications in quantum information have been noticed only very recently \cite{Cuypers2026,GargiuloHerringerZeier2026} (see also Ref.~\cite{SanigaPlanat2007}). Our analysis of purity, complementarity and uncertainty relations adds another physical dimension to this discussion.

Additionally, we presented a few selected examples showcasing how our results are of relevance for certain specific areas in quantum information including benchmarking, the study of uncertainty relations and state $k$-designs. These are likely exemplary cases only, and we expect more applications to be found in the future. Beyond the areas touched upon already, it would be interesting, for instance, to see whether the imprint of the ambient symplectic structure (within the unitary group) on the Pauli and Clifford groups allows to improve error correction protocols, to optimise Pauli sampling or classical simulation techniques.

\bibliographystyle{abbrv}
\bibliography{bibliography}

\appendix

\section{Cubic purity relations and Jordan structure}\label{app: cubic purity relations}

Thm.~\ref{thm: all weighted quadratic invariants}, Thm.~\ref{thm: all Boolean pure-state invariants} and Thm.~\ref{thm: all update-closed pure-state invariants} characterise purity in terms of quadratic constraints. On the other hand, purity can also be characterised by $\tr[\rho^2]=1$ and the analogous cubic relation $\tr[\rho^3]=1$. With this in mind, and motivated by the relations in Def.~\ref{def: PLJ algebra}, we seek a reformulation of the update-closed pure-state invariants in Def.~\ref{def: family of pure-state invariants} in terms of cubic relations that exposes their connection with Jordan structure more clearly.

To this end, we evaluate the purity constraint $p^2_\psi=p_\psi$ with $p_\psi=\frac{1}{2^n}(\one+\sum_{\one\neq P\in\aP_n}\alpha_P P)$ directly,
\begin{align*}
    p^2_\psi
    &=\frac{1}{4^n}(\one+\sum_{\one\neq P\in\aP_n}\alpha_P P)^2 \\
    &=\frac{1}{4^n}(\one+2\sum_{\one\neq P\in\aP_n}\alpha_P P+\sum_{\one\neq P,P'}\alpha_P\alpha_{P'}PP') \\
    &=\frac{1}{4^n}(2^n\one+2\sum_{\one\neq P\in\aP_n}\alpha_P P+\sum_{\one\neq P\neq P'\neq\one}\alpha_P\alpha_{P'}\{P,P'\})\; ,
\end{align*}
where we used $\{P,Q\}=\frac{1}{2}(PQ+QP)$ and that $\sum_{\one\neq P\in\aP_n}\alpha^2_P=2^n-1$ from $\tr[p^2_\psi]=1$ in the last line. Now, from $p^2_\psi=p_\psi$ one reads off the following relations: for every $\one\neq Q\in\aP_n$, it holds that
\begin{align}\label{eq: elementary quadratic constraints}
    (2^n-2)\alpha_Q
    =\frac{1}{2^n}\sum_{\one\neq P\neq P'\neq\one}\alpha_P\alpha_{P'} \tr[\{P,P'\}Q]\; .
\end{align}
Using again that $\sum_{\one\neq P\in\aP_n}\alpha^2_P=2^n-1$ (from $\tr[p^2_\psi]=1$) thus gives
\begin{align}\label{eq: triple product relation}
    T_\psi(\aP^\circ_n)
    :=\hspace{-.3cm}\sum_{\substack{\one\neq P,Q,R\in\aP_n\\ [P,Q]=[P,R]=[Q,R]=0}}\hspace{-.3cm}\alpha_P\alpha_Q\alpha_R\tr[PQR]
    =2^n(2^n-2)\sum_{\one\neq Q\in\aP_n}\alpha^2_Q
    =2^{n+1}(2^n-1)(2^{n-1}-1)\; .
\end{align}
Similarly, by including the identity element and using $I_\psi(\aP^\circ_n)=2^n-1$, one obtains
\begin{align}\label{eq: triple product relation incl 1}
    T_\psi(\aP_n)
    &:=\hspace{-.3cm}\sum_{\substack{P,Q,R\in\aP_n\\ [P,Q]=[P,R]=[Q,R]=0}}\hspace{-.3cm}\alpha_P\alpha_Q\alpha_R\tr[PQR]
    =T_{\psi}(\aP^\circ_n)+3\cdot 2^nI_\psi(\aP^\circ_n)+2^n
    =2^{3n}\; .
\end{align}
As a real constant, $T_\psi(\aP_n)$ is invariant under complex conjugation, that is, $T_\psi(\aP_n)=\overline{T_\psi(\aP_n)}$, hence, $T_\psi(\aP_n)=\frac{1}{2}(T_\psi(\aP_n)+\overline{T_\psi(\aP_n)})=T^+_\psi(\aP_n)$, where (by Hermiticity of the Pauli operators)
\begin{align*}
    T^+_\psi(\aP_n)
    &\ :=\hspace{-.2cm}\sum_{\substack{P,Q,R\in\aP_n\\ [P,Q]=[P,R]=[Q,R]=0\\(PQR)^T=PQR}}\hspace{-.3cm}\alpha_P\alpha_Q\alpha_R\tr[PQR] \\
    T^-_\psi(\aP_n)
    &\ :=\hspace{-.2cm}\sum_{\substack{P,Q,R\in\aP_n\\ [P,Q]=[P,R]=[Q,R]=0\\ (PQR)^T=-PQR}}\hspace{-.3cm}\alpha_P\alpha_Q\alpha_R\tr[PQR]\; .
\end{align*}
In other words, Eq.~(\ref{eq: triple product relation}) and Eq.~(\ref{eq: triple product relation incl 1}) are invariant under transposition. Now, recall that $\aP_n=L_a\cupdot J_a$ with $L_a=\{P\in\aP_n\mid\xi_a(P)=-1\}$ and $ J_a=\overline{L_a}=\{P\in\aP_n\mid \xi_a(P)=+1\}$ for $APA^{-1}=\xi_a(P)P^T$ with $A=W_a$, and where $\{J_a,J_a\}=J_a=\{L_a,L_a\}$ and $\{J_a,L_a\}=L_a$ (see Def.~\ref{def: PLJ algebra}). Invariance under transposition, equivalently the twisted $\xi_a$-relative complex conjugation (cf. Sec.~\ref{sec: quaternionic structure}), thus implies that the terms with an odd number of $L_a$ entries vanish, namely $T_\psi(L_a, J_a, J_a)$ and $T_\psi(L_a,L_a,L_a)$:
\begin{align*}
    \tr[PQR]
    &=\xi_a(P)\xi_a(Q)\xi_a(R)\tr[(AP^TA^{-1})(AQ^TA^{-1})(AR^TA^{-1})]\\
    &=\xi_a(P)\xi_a(Q)\xi_a(R)\tr[P^TQ^TR^T]\\
    &=\xi_a(P)\xi_a(Q)\xi_a(R)\tr[(RQP)^T]
    =\xi_a(P)\xi_a(Q)\xi_a(R)\tr[PQR]\; .
\end{align*}
Moreover, for every decomposition $\aP_n=L_a\cupdot J_a$, Eq.~(\ref{eq: triple product relation incl 1}) reduces to
\begin{align}\label{eq: symmetric decomposition}
\begin{split}
    \tT_\psi(\aP_n)
    &=\hspace{-.3cm}\sum_{\substack{Q\in J_a\\ \{P,P'\}=Q}}\hspace{-.3cm}\alpha_P\alpha_{P'}\alpha_Q\tr[PP'Q] \hspace{.2cm} + \hspace{-.1cm}\sum_{\substack{Q\in L_a\\ \{P,P'\}=Q}}\hspace{-.3cm}\alpha_P\alpha_{P'}\alpha_Q\tr[PP'Q] \\
    &=\hspace{-.3cm}\sum_{\substack{Q\in J_a\\P,P'\in J_a}}\hspace{-.3cm}\alpha_P\alpha_{P'}\alpha_Q\tr[\{P,P'\}Q]\hspace{.2cm} +
    \hspace{-.3cm}\sum_{\substack{Q\in J_a\\P,P'\in L_a}}\hspace{-.3cm}\alpha_P\alpha_{P'}\alpha_Q\tr[\{P,P'\}Q] \\
    &\hspace{.5cm} + \hspace{-.3cm}\sum_{\substack{P'\in J_a\\Q,P\in L_a}}\hspace{-.3cm}\alpha_P\alpha_{P'}\alpha_Q\tr[\{Q,P\}P'] + \hspace{-.3cm}\sum_{\substack{P\in J_a\\P',Q\in L_a}}\hspace{-.3cm}\alpha_P\alpha_{P'}\alpha_Q\tr[\{P',Q\}P] \\[.1cm]
    &=T_\psi( J_a, J_a, J_a)+3T_\psi( J_a,L_a,L_a)\; ,
\end{split}
\end{align}
where we used the cyclic symmetry in $T_\psi(\cdot,\cdot,\cdot)$ (cyclicity of the trace) in the last step. Note that Eq.~(\ref{eq: symmetric decomposition}) is valid both for (generating sets of) orthogonal and symplectic Pauli Lie algebras, as it merely relies on the algebraic relations in Def.~\ref{def: PLJ algebra}. Yet, the following stronger constraint,
\begin{align*}
    T_\psi( J_a, J_a, J_a)+T_\psi( J_a,L_a,L_a)
    &=\hspace{-.75cm}\sum_{\substack{P,P',Q\in J_a\\ [P,P']=[P,Q]=[P',Q]=0}}\hspace{-.5cm}\alpha_P\alpha_{P'}\alpha_Q\tr[PP'Q]+
    \hspace{-.5cm}\sum_{\substack{P,P'\in L_a,Q\in J_a\\ [P,P']=[P,Q]=[P',Q]=0}}\hspace{-.5cm}\alpha_P\alpha_{P'}\alpha_Q\tr[PP'Q]\\[.2cm]
    &=\sum_{\substack{P,P'\in\aP_n, Q\in J_a}}\alpha_P\alpha_{P'}\alpha_Q\tr[\{P,P'\}Q]\\
    &=4^n\sum_{Q\in J_a}\alpha^2_Q
    =4^nI_\psi( J_a)
    =2^{3n-1}\; ,
\end{align*}
uses the identities $\{J_a,J_a\}=J_a=\{L_a,L_a\}$ as well as the purity identities $I_\psi( J_a)=2^{n-1}$ in Eq.~(\ref{eq: elementary quadratic constraints}), and thus only applies to the symplectic case. Together with Eq.~(\ref{eq: triple product relation incl 1}) and Eq.~(\ref{eq: symmetric decomposition}), we thus obtain the following reformulation of the purity constraints as cubic relations,
\begin{align}\label{eq: restricted Jordan triple product}
\begin{split}
    T_\psi( J_a, J_a, J_a)
    =T_\psi( J_a,L_a,L_a)
    =\frac{2^{3n}}{4}
\end{split}\; .
\end{align}

While Eq.~(\ref{eq: triple product relation incl 1}) is simply the purity constraint $\tr[\rho^3]=1$, evaluated as an idempotency relation \cite{ByrdKhaneja2003,Kimura2003,KimuraKossakowski2005} on the full algebra, the analogous Eq.~(\ref{eq: restricted Jordan triple product}) refines this relation to symplectic Cartan decompositions (with Pauli shadow $\aP_n=L_a\cupdot J_a$ for $a=\cT_-$). This shows that pure-state invariants constrain not only the quadratic but also the cubic Bloch moments, in a manner adapted to these decompositions. From a more physical perspective, it shows that a substantial part of the Jordan structure of Pauli operators is already encoded by purity constraints, complementarity structures in the Pauli group and the associated uncertainty relations (see Ref.~\cite{FrembsHoehnNataleWever2026b}).

\section{On representations of symplectic Pauli Lie algebras}\label{app: symplectic representation}

We compare the orthogonal and symplectic Pauli Lie algebras by their action on the Bell basis. To do so, we interpret Pauli operators $Q\in\aP^\circ_n$ as generators of the special unitary Lie algebra (via $\aP_n\ni Q\ra\frac{i}{2}Q\in\su(2^n)$), and throughout this appendix adopt the convention that $[\cdot,\cdot]$ denotes the ordinary matrix Lie bracket, $[x,y]:=xy-yx$, in contrast with the normalised Pauli commutator introduced in Sec.~\ref{sec: Pauli observables}. We define the diagonal representation of $\su(2^n)$ on $\C^{2^n}\otimes\C^{2^n}$ by,
\begin{align}\label{eq: infinitesimals}
    \rho_\Delta(x)
    :=x\otimes\one+\one\otimes x\; ,
\end{align}
for every $x\in\su(2^n)$. This yields a Lie algebra homomorphism for $\su(2^n)$,
\begin{align}\label{eq: Lie alg homo}
    [\rho_\Delta(x),\rho_\Delta(y)]
    =[(x\otimes\one+\one\otimes x),(y\otimes\one+\one\otimes y)]
    =\left([x,y]\otimes\one+\one\otimes[x,y]\right)
    =\rho_\Delta([x,y])\; ,
\end{align}
and similarly for any Lie subalgebra $\fg\subset\su(2^n)$, specifically for $\fg=\spf(2^{n-1}),\ \sof(2^n)$. Since $[\rho_\Delta(x),\SWAP]=0$ for all $x\in\su(2^n)$, the representation $\rho_\Delta$ is not irreducible, but decomposes with respect to symmetric and anti-symmetric subspaces. For any $R\in\aP_n$ and using Eq.~(\ref{eq: vec}),
\begin{align}\label{eq: Bell basis rep}
    \rho_\Delta(\frac{i}{2}Q)\kett{R}
    :=\frac{i}{2}(Q\otimes\one+\one\otimes Q)\kett{R}
    =\kett{\frac{i}{2}(QR+ RQ^T)}\; ,
\end{align}
hence, either $\rho_\Delta(\frac{i}{2}Q)\kett{R}=0$ or $\rho_\Delta(\frac{i}{2}Q)\kett{R}\in\Sym^2(\C^{2^n})$ for $R^T=R$ symmetric and $\rho_\Delta(\frac{i}{2}Q)\kett{R}\in\Lambda^2(\C^{2^n})$ for $R^T=-R$ anti-symmetric. By Eq.~(\ref{eq: Bell basis rep}), the tensored Lie algebra representations of $\lf_A=\langle L_A\rangle\cong\spf(2^{n-1})$ for $L_A=\{X\in\aP_n\mid XA+AX^T=0\}$ with $A^T=-A$ and $\lf_S=\langle L_S\rangle\cong\sof(2^n)$ for $L_S=\{X\in\aP_n\mid XS+SX^T=0\}$ and $S^T=S$ decompose along the splitting $\C^{2^{2n}}=\Sym^2(\C^{2^n})\oplus\Lambda^2(\C^{2^n})$.\footnote{Here, we write $L_R:=L_r$ for $R=W_r\in\aP_n$.} In fact, for any nondegenerate (anti-)symmetric operator, we have the following fact.

\begin{lemma}\label{lm: adjoint rep in Bell basis}
    For any invertible $A^T=-A$, let $\lf^\C_A=\{x\in\mathfrak{sl}_{2^n}(\C)=\su(2^n)\otimes_\R\C\mid xA+Ax^T=0\}$. Similarly, for any invertible $S^T=S$, let $\lf^\C_S:=\{x\in\mathfrak{sl}_{2^n}(\C)\mid xS+Sx^T=0\}$. Then the maps
    \begin{align*}
        \Phi_A:\lf^\C_A&\longrightarrow\Sym^2(\C^{2^n}) &
        x&\longmapsto \kett{xA}\; ,\\
        \Phi_S:\lf^\C_S&\longrightarrow \Lambda^2(\C^{2^n}) &
        x&\longmapsto \kett{xS}\; ,
    \end{align*}
    are isomorphisms of $\lf^\C_R$-modules for $R=A,S$, where $\lf^\C_R$ acts on itself by the adjoint representation, and on $\Sym^2(\C^{2^n})$, respectively $\Lambda^2(\C^{2^n})$ by the diagonal representation $\rho_\Delta$.
\end{lemma}

\begin{proof}
    If $x\in\lf^\C_A$, then $xA+Ax^T=0$, from which
    \begin{align*}
        (xA)^T=A^Tx^T=-Ax^T=xA\; ,
    \end{align*}
    hence, $xA\in\Sym^2(\C^{2^n})$. Since $A$ is invertible, $x\mapsto xA$ is injective, and since $\dim_\C\lf^\C_A=\dim_\R\mathfrak{sp}(2^{n-1})=2^{n-1}(2^n+1)=\dim_\C\Sym^2(\C^{2^n})$, $\Phi_A$ is an isomorphism of vector spaces. Moreover, (after extending the definition of the diagonal representation in Eq.~(\ref{eq: infinitesimals}) to $\mathfrak{sl}_{2^n}(\C)=\su(2^n)\otimes_\R\C$), we have
    \begin{align*}
        \rho_\Delta(y)\kett{xA}
        =\kett{yxA+xAy^T}
        =\kett{yxA-xyA}
        =\kett{[y,x]A}\; ,
    \end{align*}
    for all $y\in\lf^\C_A$, equivalently $\Phi_A([y,x])=\rho_\Delta(y)\Phi_A(x)$. Consequently, $\Phi_A$ identifies the diagonal action on $\Sym^2(\C^{2^n})$ with the adjoint action on $\lf^\C_A$. The proof for $S$ is analogous. If $x\in\lf^\C_S$, then $xS+Sx^T=0$, from which
    \begin{align*}
        (xS)^T=Sx^T=-xS\; ,
    \end{align*}
    and thus $xS\in\Lambda^2(\C^{2^n})$. Invertibility of $S$ implies that $x\mapsto xS$ is injective, and since $\dim_\C\lf^\C_S=\dim_\R\mathfrak{so}(2^n)=2^{n-1}(2^n-1)=\dim_\C\Lambda^2(\C^{2^n})$, $\Phi_S$ is an isomorphism of vector spaces. Moreover,
    \begin{align*}
        \rho_\Delta(y)\kett{xS}
        =\kett{yxS+xSy^T}
        =\kett{yxS-xyS}
        =\kett{[y,x]S}\; ,
    \end{align*}
    for all $y\in\lf^\C_S$, equivalently $\Phi_S([y,x])=\rho_\Delta(y)\Phi_S(x)$, hence, $\Phi_S$ is an isomorphism of $\lf^\C_S$-modules.
\end{proof}

The module isomorphisms in Lm.~\ref{lm: adjoint rep in Bell basis} are the standard identifications $\spf(V)\simeq\Sym^2(V)$ and $\sof(V)\simeq\Lambda^2(V)$, obtained by lowering an index with the invariant bilinear form; closely related tensorial identifications underlie, respectively, bosonic and fermionic quadratic (Gaussian) dynamics.

\begin{corollary}\label{cor: irreducibility}
    The $\lf_A$-module $\Sym^2(\C^{2^n})$ is irreducible for every $A^T=-A\in\aP_n$ and $n\geq 1$. Moreover, the $\lf_S$-module $\Lambda^2(\C^{2^n})$ is irreducible for every $S^T=S\in\aP_n$ and $n\geq 3$.\footnote{For $n=1$, $\Lambda^2(\C^2)$ is one-dimensional and thus irreducible. For $n=2$, $\sof(4,\C)$ is not simple as $\sof(4,\C)\cong\mathrm{sl}_2(\C)\oplus\mathrm{sl}_2(\C)$, and for $n=1$, $\sof(2,\C)$ is Abelian, hence, any irreducible representation is one-dimensional.\label{fn: exceptional so cases}}
\end{corollary}

\begin{proof}
    By Lm.~\ref{lm: adjoint rep in Bell basis}, $\Sym^2(\C^{2^n})$ and $\Lambda^2(\C^{2^n})$ are isomorphic to the adjoint modules of $\lf^\C_A$ and $\lf^\C_S$, respectively. These complex Lie algebras are simple in the stated ranges \cite{GoodmanWallach2009}, hence their adjoint representations are irreducible. Finally, a complex-linear $\fg$-module is irreducible if and only if it is irreducible as a $\fg^\C$ module, which implies the result.
\end{proof}

While the diagonal representation of $\spf(2^{n-1})$ acts irreducibly on $\Sym^2(\C^{2^n})$, $\Lambda^2(\C^{2^n})$ decomposes under this representation into a one-dimensional representation spanned by $A$, since $\rho_\Delta(\frac{i}{2}Q)\kett{A}=0$ for all $Q\in L_A$, and an irreducible piece $\Lambda^2_0(\C^{2^n})$ (see Sec.~10.2.3 in Ref.~\cite{GoodmanWallach2009}),\footnote{$\Lambda^2_0(\C^{2^n})$ can be identified with the kernel of the restriction to $\Lambda^2(\C^{2^n})$ of the contraction $c_A:\Lambda^2(\C^{2^n})\ra\C$ given by $c_A(v\wedge w)=B_A(v,w)=w^TAv$. Since $B_A$ is nondegenerate, $\dim(\mathrm{im}(c_A))=1$, hence, $\Lambda^2_0(\C^{2^n})$ has codimension $1$.}
\begin{align}\label{eq: Lambda decomposition}
    \C^{2^{2n}}
    =\Sym^2(\C^{2^n})\oplus\left(\langle A\rangle\oplus\Lambda^2_0(\C^{2^n})\right)\; .
\end{align}
Similarly, by Cor.~\ref{cor: irreducibility}, the diagonal action of $\sof(2^n)$ is irreducible on $\Lambda^2(\C^{2^n})$ for $n\geq 3$, yet splits on $\Sym^2(\C^{2^n})$ into a one-dimensional summand spanned by $S$ and the irreducible component of traceless symmetric states $\Sym^2_0(\C^{2^n})$ (see Sec.~10.2.5 in Ref.~\cite{GoodmanWallach2009}),\footnote{$\Sym^2_0(\C^{2^n})$ is the kernel of the contraction map $c_S:\Sym^2(\C^{2^n})\ra\C$ given by $c_S(v\otimes v)=B_S(v,v)=v^TSv$.}
\begin{align}\label{eq: Sym decomposition}
    \C^{2^{2n}}
    =\left(\langle S\rangle\oplus\Sym^2_0(\C^{2^n})\right)\oplus\Lambda^2(\C^{2^n})\; .
\end{align}

\subsection{Pure-state invariants from quadratic Casimirs}\label{app: pure-state invariants from quadratic Casimirs}

Irreducibility of the $\lf_A$-representation on $\Sym^2(\C^{2^n})$ (by Cor.~\ref{cor: irreducibility}) is closely related with pure-state invariants. To see this, recall that the quadratic Casimir element of a Lie algebra $\fg$ is defined as $C_\fg:=\sum_{s\in B}ss^*$, where $B$ is any orthonormal basis of $\fg$, and $B^*$ its dual basis, defined with respect to a non-degenerate, symmetric, invariant, bilinear form on $\fg$. For simple Lie algebras, the Killing form, $K_\fg(x,y)=\tr_\fg[\ad_x\ad_y]$ defines such a form and is unique up to rescaling. Writing $\ad_x=L_x-R_x$ for $L_x(z)=xz$ and $R_x(z)=zx$ with $x,z\in\mathfrak{gl}(2^n)$ the left and right multiplication operators in $\End(\mathfrak{gl}(2^n))$, one obtains
\begin{align*}
    K_{\mathfrak{gl}(2^n)}(x,y)
    &=\tr_{\mathfrak{gl}(2^n)}[\ad_x\ad_y]\\
    &=\tr_{\mathfrak{gl}(2^n)}[L_xL_y]+\tr_{\mathfrak{gl}(2^n)}[R_xR_y]-\tr_{\mathfrak{gl}(2^n)}[L_xR_y]-\tr_{\mathfrak{gl}(2^n)}[R_xL_y]\\
    &=2^{n+1}\tr_\cH[xy]-2\tr_\cH[x]\tr_\cH[y]\; .
\end{align*}
for the Killing form on $\mathfrak{gl}(2^n)$ with $\cH=\C^{2^n}$. For $x,y\in\mathfrak{sl}(d,\C)$ the second term vanishes. Moreover, $\mathfrak{gl}(2^n,\C)=\mathfrak{sl}(2^n,\C)\oplus\C\one$ is invariant under $\ad_x$ for $x\in\mathfrak{sl}(d,\C)$,
and $\ad_x(\one)=0$, hence, the scalar summand contributes nothing to the trace, and $K_{\mathfrak{sl}(2^n)}(x,y)=2^{n+1}\tr_{\cH}[xy]$. Finally, since $\mathfrak{sl}(2^n)\cong\su(2^n)\otimes\C$ and the Killing form of a real Lie algebra is the restriction of its complexification, we have $K_{\su(2^n)}(x,y)=2^{n+1}\tr[xy]$.\footnote{For a simple Lie subalgebra, the restriction of the ambient Killing form is an invariant symmetric bilinear form and thus proportional to its Killing form. In particular, the Killing forms of $\lf_A\cong\spf(2^{n-1})$ and $\lf_S\cong\sof(2^n)$ are proportional to the Killing form on $\su(2^n)$, explicitly
\begin{align*}
    K_{\spf(2^{n-1})}
    &=\frac{2^{n-1}+1}{2^n}K_{\su(2^n)}|_{\spf(2^{n-1})} &
    K_{\sof(2^n)}
    &=\frac{2^{n-1}-1}{2^n}K_{\su(2^n)}|_{\sof(2^n)}
    \; .
\end{align*}}
Since the proportionality constant will not concern us here, a natural choice is to pick a multiple of the Killing form for which the Pauli basis is orthonormal in the respective Lie algebras, and we thus define $k_{\su(2^n)}\sim K_{\su(2^n)}$ by  $k_{\su(2^n)}(x,y)=-\frac{4}{2^n}\tr[xy]$ such that $k_{\su(2^n)}(\frac{i}{2}Q,\frac{i}{2}P)=\delta_{PQ}$. It then follows that the quadratic Casimir element of $\lf_A\subset\su(2^n)$ in the diagonal representation $\rho_\Delta$ takes the following form
\begin{equation}\label{eq: quadratic Casimir}
\begin{aligned}
    C_{\lf_A}
    =\sum_{Q\in L_A}\rho_\Delta\left(\frac{i}{2}Q\right)\rho^*_\Delta\left(\frac{i}{2}Q\right)
    &=\sum_{Q\in L_A}\rho^2_\Delta\left(\frac{i}{2}Q\right)\\
    &=-\frac{1}{4}\sum_{Q\in L_A}(Q\otimes\one+\one\otimes Q)^2
    =-\frac{1}{2}\sum_{Q\in L_A}(\one\otimes\one+Q\otimes Q)\; .
\end{aligned}
\end{equation}
(Here, $\rho^*_\Delta$ denotes the dual basis element with respect to the invariant form $k_{\su(2^n)}$ above.) Comparison with Eq.~(\ref{eq: diagonal invariant operator}) shows that the operator $M_A:=M(\zeta^-_a)=\sum_{Q\in L_A}Q\otimes Q$ (for $A=W_a$ and $\zeta^\pm_a:=\frac{1}{2}(1\pm\xi_a)$) corresponds with the quadratic Casimir in Eq.~(\ref{eq: quadratic Casimir}) up to a constant shift,
\begin{align*}
    M_A
    =\sum_{Q\in L_A}Q\otimes Q
    =-2C_{\lf_A}-|L_A|\one\otimes\one\; .
\end{align*}
Now, the crucial point is that since $\lf_A \cong\spf(2^{n-1})$ acts irreducibly on $\Sym^2(\C^{2^n})$ by Cor.~\ref{cor: irreducibility}, by Schur's lemma any $\spf(2^{n-1})$-invariant must thus act as a multiple of the projector $\Pi_\sym$ on the restriction to the symmetric subspace. Clearly, this implies $C_{\lf_A}|_{\Sym^2(\C^{2^n})}\sim\Pi_\sym$ and thus,
\begin{align*}
    I_\psi(\zeta^-_a)
    =\sum_{Q\in L_A}\tr[p_\psi Q]^2
    =\tr[p^{\otimes 2}_\psi M_A]
    =-2\tr[p^{\otimes 2}_\psi|C_{\lf_A}]-|L_A|
    =\mathrm{const}\; ,
\end{align*}
for all $\psi\in\CP^{2^n-1}$. In particular, with $|L_A|=2^{n-1}(2^n+1)$ we find $C_{\lf_A}|_{\Sym^2(\C^{2^n})}=-2^{n-1}(2^{n-1}+1)\Pi_\sym$ for all $A^T=-A\in\aP_n$. We have thus shown that pure-state invariants in Thm.~\ref{thm: all update-closed pure-state invariants} correspond with quadratic Casimir elements of $\spf(2^{n-1})$. By contrast, $\sof(2^n)$ does not act irreducibly on $\Sym^2(\C^{2^n})$, and its quadratic Casimir is not a scalar on $\Sym^2(\C^{2^n})$, but has distinct eigenvalues $0$ and $-4^{n-1}$ on $\langle S\rangle$ and $\Sym^2_0(\C^{2^n})$, respectively. Consequently, $\sof(2^n)$ does not give rise to pure-state invariants.

\subsection{Changing representation}\label{app: changing representation}

We close by relating representations and associated quadratic Casimir elements corresponding to different $A^T=-A\in\aP_n$. To this end, fix an element $0\neq x_0\in\lf^\C_A$, and define the operator
\begin{align*}
    \rho^-_\Delta(x_0)
    :=x_0\otimes\one-\one\otimes x_0\; ,
\end{align*}
such that $\rho^-_\Delta(x_0)\kett{A}=\kett{x_0A-Ax_0^T}=\kett{2x_0A}=2\Phi_A(x_0)\in\Sym^2(\C^{2^n})$. Then
\begin{align*}
    \Sym^2(\C^{2^n})
    =\U(\lf^\C_A)\cdot(\rho^-_\Delta(x_0)|A\rangle\rangle)\; ,
\end{align*}
where $\U(\lf^\C_A)$ denotes the universal enveloping algebra of the Lie algebra $\lf^\C_A$ acting on the diagonal representation $\rho_\Delta$. Under the identification $\Phi_A:\lf^\C_A\to\Sym^2(\C^{2^n})$, $x\mapsto\kett{xA}$ in Lm.~\ref{lm: adjoint rep in Bell basis}, the vectors obtained by iterated diagonal action correspond with nested commutators in $\lf^\C_A$,\footnote{In addition to $\rho_\Delta$ being a Lie algebra homomorphism, that is, $[\rho_\Delta(x),\rho_\Delta(y)]=\rho_\Delta([x,y])$ (see Eq.~(\ref{eq: Lie alg homo})), one easily verifies that $[\rho_\Delta(x),\rho^-_\Delta(y)]=\rho^-_\Delta([x,y])$ and $[\rho^-_\Delta(x),\rho^-_\Delta(y)]=\rho_\Delta([x,y])$ for all $x,y\in\su(2^n)$.}
\begin{align}\label{eq: nested commutators}
    \kett{(\ad_{y_k}\cdots\ad_{y_2}\ad_{y_1}x_0)A}\quad\quad y_j\in\lf^\C_A\; .
\end{align}
Since $\lf^\C_A$ is simple, the ideal generated by $x_0$ is all of $\lf^\C_A$. Hence, these nested commutators span all of $\Sym^2(\C^{2^n})$.\footnote{It is easy to see that two elements in $\aP_n$ are related by at most two nested commutators. For $n\geq 2$, the same applies under the restriction to $L_A$.} Similarly, for $0\neq x_0\in\lf^\C_S$, $\rho^-_\Delta(x_0)\kett{S}=\kett{2x_0S}\in\Lambda^2(\C^{2^n})$,
\begin{align*}
    \Lambda^2(\C^{2^n})
    =\U(\lf^\C_S)\cdot(\rho^-_\Delta(x_0)\kett{S})\; ,
\end{align*}
and for $n\geq 3$, that is, when $\lf^\C_S$ is simple,$^{\ref{fn: exceptional so cases}}$ $\Lambda^2(\C^{2^n})$ is spanned by nested commutators.

The diagonal representation $\rho_\Delta$ thus relates the different copies of the symplectic Pauli Lie algebras corresponding to different choices of anti-symmetric operators $A^T=-A(\in\aP_n)$ as follows. Recall (from Lm.~6 and Thm.~2 in Ref.~\cite{FrembsHoehnNataleWever2026a}) that the different Pauli Lie algebras $\lf_A\cong\spf(2^{n-1})$ for $A^T=-A\in\aP_n$ are related via conjugation by a Clifford unitary,
\begin{align*}
    (U\otimes U)C_{\lf_A}(U\otimes U)^\dagger
    &=(U\otimes U)\sum_{Q\in L_A}\rho^2_\Delta(\frac{i}{2}Q)(U\otimes U)^\dagger\\
    &=-\frac{1}{2}(U\otimes U)\left(|L_A|\one\otimes\one+\sum_{Q\in L_A}Q\otimes Q\right)(U\otimes U)^\dagger\\
    &=-\frac{1}{2}\left(|L_A|\one\otimes\one+\sum_{Q'\in L_{UAU^T}}Q'\otimes Q'\right)
    =\sum_{Q'\in L_{UAU^T}}\rho^2_\Delta(\frac{i}{2}Q')
    =C_{\lf_{UAU^T}}\; .
\end{align*}
This suggests to extend the representation $\rho_\Delta:\su(2^n)\ra\mathrm{End}(\C^{2^{2n}})$ equivariantly with respect to the adjoint action for all $x\in\lf_A$,
\begin{align*}
    \rho_\Delta(\Ad_U(x))
    =\Ad_{U\otimes U}(\rho_\Delta(x))\; .
\end{align*}
While the different conjugated Lie algebras and corresponding quadratic Casimir elements differ, the latter remain proportional to the projector $\Pi_\sym$ under the restriction to $\Sym^2(\C^{2^n})$, since $\Ad_{U\otimes U}$ commutes with $\SWAP$. Explicitly, using the isomorphism in Lm.~\ref{lm: adjoint rep in Bell basis} and the representation via nested commutators in Eq.~(\ref{eq: nested commutators}), the different copies of the symplectic Lie algebra $\spf(2^{n-1})$ are distinguished by different choices of `cyclic seed vectors' $\kett{xA}$.

Finally, we remark that Lm.~\ref{lm: adjoint rep in Bell basis} and Cor.~\ref{cor: irreducibility} are not tied to qubit systems. The finite Pauli geometry of the pure-state invariants in Def.~\ref{def: I} is thus best seen as a discrete shadow of the above irreducibility statements about the action of the symplectic Lie algebra. Generalising the pure-state invariants in Def.~\ref{def: I} beyond qubit systems is a natural direction for future work.

\end{document}